\documentclass[letterpaper]{article} 
\usepackage{aaai2026}  
\usepackage{times}  
\usepackage{helvet}  
\usepackage{courier}  
\usepackage[hyphens]{url}  
\usepackage{graphicx} 
\usepackage{natbib}  
\usepackage{caption} 
\usepackage{algorithm}
\usepackage{algorithmic}

\usepackage{amsmath}
\usepackage{amssymb}
\usepackage{mathtools}
\usepackage{amsthm}
\theoremstyle{plain}
\newtheorem{theorem}{Theorem}[section]
\newtheorem{proposition}[theorem]{Proposition}

\theoremstyle{definition}

\theoremstyle{remark}

\usepackage{xcolor}
\usepackage{sverb, longtable}
\usepackage{rotating}
\usepackage{subcaption}

\usepackage{epstopdf}
\usepackage{url}
\usepackage{multirow}%
\usepackage{amssymb}%
\usepackage{amsthm}%
\usepackage{mathrsfs}%
\usepackage{amsmath}
\usepackage{mathtools}
\usepackage{booktabs}%
\usepackage{bm}
\usepackage{dsfont}
\usepackage[capitalize,noabbrev]{cleveref}
\usepackage{tikz}
\usepackage{pgfplots}
\usetikzlibrary{calc}
\usetikzlibrary{arrows,shapes,chains}
\usetikzlibrary{decorations.pathmorphing}
\usetikzlibrary{patterns}

\newcommand{\cb}[1]{\ifmmode {\boldsymbol{#1}}\else ${\boldsymbol{#1}}$\fi}
\newcommand{\cp}[1]{\ifmmode {\mathcal{#1}}\else ${\mathcal{#1}}$\fi}

\newcommand{\orangeso}{\color[RGB]{203,75,22}}

\newcommand{\argmax}{\mathop{\mathrm{argmax}}\limits}

\usepackage{makecell}

\usepackage{newfloat}
\usepackage{listings}
\DeclareCaptionStyle{ruled}{labelfont=normalfont,labelsep=colon,strut=off} 
\floatstyle{ruled}
\newfloat{listing}{tb}{lst}{}
\floatname{listing}{Listing}
\title{Disturbance-based Discretization, 
Differentiable IDS Channel, \\
and an IDS-Correcting Code for DNA-based Storage}
\author{
    Alan J.X. Guo\textsuperscript{\rm 1,\rm 2}\thanks{Corresponding author.},
    Mengyi Wei\textsuperscript{\rm 1},
    Yufan Dai\textsuperscript{\rm 1},
    Yali Wei\textsuperscript{\rm 1},
    Pengchen Zhang\textsuperscript{\rm 1}
}
\affiliations{

    \textsuperscript{\rm 1}Center for Applied Mathematics, 
    Tianjin University, 
    Tianjin 300072, China\\
    \textsuperscript{\rm 2}State Key Laboratory of Synthetic Biology, 
    Tianjin University, 
    Tianjin 300072, China\\

    \{jiaxiang.guo, mengyi.wei, daiyufan, yaliwei222\_\_, zhangpengchen\}@tju.edu.cn

}

\usepackage{bibentry}

\begin{document}

\maketitle

\begin{abstract}
    With recent advancements in next-generation data storage, especially in biological molecule-based storage, 
    insertion, deletion, and substitution (IDS) error-correcting codes have garnered increased attention. 
    However, a universal method for designing tailored IDS-correcting codes across varying channel settings 
    remains underexplored. 
    We present an autoencoder-based approach, THEA-code, 
    aimed at efficiently generating IDS-correcting codes for complex IDS channels. 
    In the work, a disturbance-based discretization
    is proposed to discretize the features of the autoencoder, 
    and a simulated differentiable IDS channel is developed as a differentiable alternative for IDS operations. 
    These innovations facilitate the successful convergence of the autoencoder, 
    producing channel-customized IDS-correcting codes that demonstrate 
    commendable performance across complex IDS channels, 
    particularly in realistic DNA-based storage channels. 
\end{abstract}


\section{Introduction}\label{sec:introduction}
Biological molecule-based storage, a method that uses the synthesis and sequencing of biological molecules 
for information storage and retrieval, 
has attracted significant attention~\cite{church2012next,goldman2013towards,grass2015robust,erlich2017dna,organick2018random,dong2020dna,chen2021artificial,shaikh2022high,welzel2023dna}. 
Currently, most applications in this field are focused on DNA-based information storage~\cite{meiser2022synthetic}. 

Due to the involvement of biochemical procedures, the storage pipeline 
can be viewed as an insertions, deletions, or substitutions (IDS) channel~\cite{blawat2016forward} 
over $4$-ary sequences with the alphabet $\{\mathrm{A},\mathrm{T},\mathrm{G},\mathrm{C}\}$. 
Consequently, an IDS-correcting encoding/decoding method plays a key role in 
biological molecule-based storage. 

However, despite the existence of excellent combinatorial IDS-correcting 
codes~\cite{varshamov1965code,levenshtein1966binary,sloane2000single, mitzenmacher2009survey,cai2021correcting,garbys2023beyond,barlev2023size}, 
applying them in DNA-based storage remains challenging. 
The biochemical channel in DNA-based storage is more complex than those studied in previous works, 
with factors such as inhomogeneous error probabilities across error types, base indices, and even 
sequence patterns~\cite{hirao1992extraordinarily, press2020hedges, blawat2016forward, cai2021correcting, hamoum2021channel}. 
Additionally, most of the aforementioned combinatorial codes focus on correcting either a single error 
or a burst of errors, whereas multiple independent errors within the same DNA sequence are common in DNA-based storage. 
To address this, an outer code is usually employed to correct residual errors 
that are beyond the capability of the inner IDS code. 

Given the complexity of the IDS channel, we leverage the universality of deep learning methods
by employing an autoencoder~\cite{baldi2012autoencoders} as the foundation for an end-to-end IDS-correcting code. 
This approach enables researchers to train customized codes tailored to various IDS channels 
through a unified training procedure, 
rather than manually designing specific combinatorial codes for each IDS channel setting, 
many of which remain unexplored. 

To realize this approach, two novel techniques are developed, which we believe offer greater contributions 
to the communities than the code itself. 

Firstly, the discretization effect of applying disturbance 
in a non-generative model is investigated in this work. 
It is observed that introducing disturbance to the logistic feature 
forces the non-generative model to reduce the disturbance caused indeterminacy 
by producing more confident logits, thereby achieving discretization. 
This aligns with the discrete codewords of an error-correcting code (ECC) in this work, 
and provides an alternative approach for bridging the gap between continuous 
models and discrete applications. 

Secondly, a differentiable IDS channel using a Transformer-based model~\cite{vaswani2017attention} is developed. 
The non-differentiable nature of IDS operations presents a key challenge for deploying deep learning models 
that rely on gradient descent training. 
To tackle this, a model is trained in advance to mimic the IDS operations according to a given error profile. 
It can serve as a plug-in module for the IDS channel and 
is backpropagable within the network. 
This differentiable IDS channel has the potential to act as a general module 
for addressing IDS or DNA-related problems using deep learning methods. 
For instance, researchers could build generative models on this module to simulate the biochemical processes 
involved in manipulating biosequences. 

Overall, this work implements a heuristic end-to-end autoencoder as an IDS-correcting code, referred to as THEA-Code. 
The encoder maps the source DNA sequence into a longer codeword sequence. 
After introducing IDS errors to the codeword, 
a decoder network is employed to reconstruct the original source sequence from the codeword. 
During the training of this autoencoder, 
disturbance-based discretization 
is applied to 
the codeword sequence to produce one-hot-like vectors, and the differentiable IDS channel 
serves as a substitute for conventional IDS channel, enabling gradient backpropagation. 


To the best of our knowledge, this work presents the first end-to-end autoencoder solution 
for an IDS-correcting code. 
It introduces the disturbance-based discretization, 
and proposes the first differentiable IDS channel. 
It is also the first universal method for designing tailored 
IDS-correcting codes across varying channel settings. 
Experiments across multiple complex IDS channels, 
particularly in the realistic DNA-based storage channel, 
demonstrate the effectiveness of the proposed THEA-Code.

\section{Related Works}\label{sec:related-works}
Many established IDS-correcting codes are rooted in the Varshamov-Tenengolts (VT) code~\cite{varshamov1965code,levenshtein1966binary}, 
including~\cite{calabi1969family,tanaka1976synchronization,sloane2000single,cai2021correcting,garbys2023beyond}. 
These codes often rely on rigorous mathematical deduction and provide firm proofs for their coding schemes. 
However, the stringent hypotheses they use tend to restrict their practical applications. 
Heuristic IDS-correcting codes for DNA-based storage, such as those proposed 
in~\cite{pfister2021polar,yan2022segmented,maarouf2022concatenated,welzel2023dna},
usually incorporate synchronization markers~\cite{sellers1962bit,srinivasavaradhan2021trellis,haeupler2021synchronization}, 
watermarks~\cite{davey2001reliable}, or positional information~\cite{press2020hedges} 
within their encoded sequences. 
Recently, directly correcting errors in retrieved DNA reads without sequence reconstruction 
has been investigated, demonstrating promising performance~\cite{welter2024end}.

In recent years, deep learning methods have found increasing applications 
in coding theory~\cite{ibnkahla2000applications,simeone2018very,akrout2023domain,park2025crossmpt}. 
Several architectures have been employed as decoders or sub-modules of conventional codes on 
the additive white Gaussian noise (AWGN) channel.
In~\cite{cammerer2017scaling}, 
the authors applied neural networks to replace sub-blocks in the conventional iterative decoding algorithm 
for polar codes. 
Recurrent neural networks (RNN) were used for decoding convolutional and turbo codes~\cite{kim2018communication}. 
Both RNNs and Transformer-based models have served as belief propagation decoders 
for linear codes~\cite{nachmani2018deep,choukroun2022error,choukroun2023denoising,choukroun2024deep,choukroun2024a,choukroun2024learning}. 
Hypergraph networks were also utilized as decoders for block codes in~\cite{nachmani2019hyper}. 
Despite these advancements, end-to-end deep learning solutions remain relatively less explored. 
As mentioned in~\cite{jiang2019turbo}, 
direct applications of multi-layer perceptron (MLP) and convolutional neural network (CNN) 
are not comparable to conventional methods. 
To address this, the authors in~\cite{jiang2019turbo} used deep models to replace sub-modules of 
a turbo code skeleton, 
and trained an end-to-end encoder-decoder model. 
Similarly, in~\cite{makkuva2021ko}, neural networks were employed to replace the Plotkin mapping for the Reed-Muller code. 
Both of these works inherit frameworks from conventional codes and utilize neural networks as replacements for 
key modules. 
In~\cite{balevi2020autoencoder}, researchers proposed an autoencoder-based inner code with one-bit quantization for 
the AWGN channel. 
Confronting challenges arising from quantization, they utilized interleaved training on the encoder and decoder. 

\section{Disturbance-based Discretization}\label{sec:gumbel}
In this work, it is observed that introducing disturbance to 
the categorical distribution feature produced by 
a non-generative model causes 
the feature to resemble a one-hot vector. 

Intuitively, the non-generative model may attempt to reduce the indeterminacy 
introduced by the disturbance 
by generating more confident logits. 
When a logit $x$ is perturbed by a noise term $\epsilon$ before producing the categorical distribution, 
a fully converged model, aiming to generate outputs with high certainty, 
have to produce logits $x$ with significantly larger magnitudes
to diminishing the relative proportion of the disturbance $\epsilon$. 
From this perspective, the logit $x$ becomes more confident, producing probabilities that are closer to one-hot vectors 
and exhibit lower entropy. 
This effect is confirmed by monitoring the entropy of the categorical distribution 
in the experiments presented in 
\cref{subsec:expgs}.

Let $\bm{x}$ be the logits that produce the probabilities $\bm{\pi}=\{\pi_1,\pi_2,\ldots,\pi_k\}$ 
via the softmax function, 
\begin{equation}
    \pi_i = \frac{\exp{x_i}}{\sum_{j=1}^{k}\exp{x_j}}, \quad i=1,2,\ldots,k. 
\end{equation}
In this work, the non-generative disturbance is introduced to $\bm{\pi}$ by sampling 
from the Gumbel distribution~\cite{gumbel1935valeurs}. 
It follows the same formula as the Gumbel-Softmax, 
which has been widely used in generative models for generating samples~\cite{jang2017categorical,maddison2017the}. 
Specifically, 
the non-generative disturbance is applied to $\bm{x}$ 
using the following formula: 
\begin{equation}\label{eqn:gsequation}
\mathrm{GS}(\bm{x})_i = \frac{\exp{((x_i+g_i)/\tau)}}{\sum_{j=1}^{k}\exp{((x_j+g_j)/\tau)}}, \quad i=1,2,\ldots,k,
\end{equation}
where $g_1,g_2,\ldots,g_k$ are i.i.d. samples drawn from the Gumbel distribution $G(0,1)$ and 
$\tau$ is the temperature that controls the entropy. 

Applying $\mathrm{GS}(\bm{x})$ in a non-generative model  
is found to induce the model to produce more confident logits $\bm{x}$ and, consequently, 
probabilities $\bm{\pi}$ that resemble one-hot vectors,  
as stated in Proposition \ref{thm:gs}. 
\begin{proposition}\label{thm:gs}
    By introducing disturbance to a non-generative autoencoder's feature logits $\bm{x}$ 
    via $\mathrm{GS}(\bm{x})$, 
    the autoencoder, upon non-trivial convergence, produces confident 
    logits $\bm{x}$, 
    resulting in one-hot-like probabilities $\bm{\pi}$. 
\end{proposition}

\begin{proof}[Brief proof:]
    Consider the binary case with temperature $\tau=1$, 
    and let $\bm{x} = (x_1, x_2)$ be the logits from the upstream model, 
    with Gumbel noise added to compute $\bm{y} = \mathrm{GS}(\bm{x})$. 
    At convergence, the gradient of the loss $\mathcal{L} = f(\bm{y})$ 
    with respect to $\bm{x}$ approaches zero. 
    By computing $\partial \mathcal{L}/\partial x_1$, 
    we find that it depends on $y_1 y_2$ and the derivatives of $f$, 
    implying that either the output probabilities $y_i$ are near 0/1, or $f(\bm{y})$ is insensitive to its inputs. 
    The former leads to low-entropy, one-hot-like distributions in $\bm{y}$. 
    In the latter case, since $\bm{y}$ varies due to the Gumbel noise, an $f(\bm{y})$ that is insensitive to its inputs
    implies that the model has converged to a trivial solution, contradicting the hypothesis. 
    Further, the logits $\bm{x}$ are bounded by the probability that $\bm{y}$ deviates from a one-hot-like distributions, 
    indicating that the model produces confident logits to suppress the effect 
    of the Gumbel noise. 
\end{proof}

A full version of the proof is provided in 
\cref{app:proof}. 
Based on this, a converged model that applies the disturbance in \cref{eqn:gsequation} 
to its feature logits $\bm{x}$ 
will be constrained to produce one-hot-like probability vectors
when \cref{eqn:gsequation} is replaced with the $\mathrm{softmax}$ during inference.

\section{Differentiable IDS Channel \\on $3$-Simplex $\Delta^3$}\label{sec:ids}
It is evident that the operations of insertion and deletion are not differentiable. 
Consequently, a conventional IDS channel, 
which modifies a sequence by directly applying IDS operations, 
hinders gradient propagation and cannot be seamlessly integrated into deep learning-based methods. 

Leveraging the logical capabilities inherent in Transformer-based models, 
a sequence-to-sequence model is employed to simulate the conventional IDS channel. 
Built on deep models, this simulated IDS channel is differentiable. 
In the following discussion, we use the notation $\mathrm{CIDS}(\cdot,\cdot)$ to represent the {C}onventional {IDS} channel, 
and $\mathrm{DIDS}(\cdot,\cdot;\theta)$ for the simulated {D}ifferentiable {IDS} channel. 
The simulated channel is trained independently before being integrated into the autoencoder, 
whose learned parameters remain fixed during the optimization of the autoencoder. 

As the model utilizes probability vectors rather than discrete letters, 
we need to promote conventional IDS operations onto 
the $3$-simplex $\Delta^3$, where $\Delta^3$ is defined as the collection $4$-dimentional probability vectors
\begin{equation}
\Delta^3 = \{\bm{\pi}|\, \pi_i\geq 0, {\sum_{i=1}^4 \pi_i = 1, i=1,2,3,4}\}.
\end{equation} 
For a sequence of probability vectors $\bm{C}=(\bm{\pi}_1,\bm{\pi}_2,\ldots,\bm{\pi}_k)$, 
where each $\bm{\pi}_i$ is an element from the simplex $\Delta^3$, 
the IDS operations are promoted as follows. 

Insertion at index $i$ involves adding a one-hot vector representing the inserted symbol 
from the alphabet $\{\mathrm{A},\mathrm{T},\mathrm{G},\mathrm{C}\}$ before index $i$. 
Deletion at index $i$ simply removes the vector $\bm{\pi}_i$ from $\bm{C}$. 
For substitution, the probability vector $\bm{\pi}_i$ is rolled by corresponding offsets 
for the three types of substitutions, which correspond to substitute $\#1,\#2,$ and $\#3$ in 
\cref{fig:profile} 
from 
\cref{app:dataset}. 
For example, applying a type-$\#1$ substitution at index $i$ rolls the original vector 
$\bm{\pi}_i = (\pi_{i1},\pi_{i2},\pi_{i3},\pi_{i4})$ into 
$(\pi_{i4},\pi_{i1}, \pi_{i2},\pi_{i3})$. 
It is straightforward to verify that the promoted IDS operations degenerate to standard IDS operations 
when the probability vectors are constrained to a one-hot representation. 

\begin{figure}[tb!]
    \vskip -0.in
        \centering
{\linespread{1}
        \centering
        \tikzstyle{format}=[circle,draw,thin,fill=white]
        \tikzstyle{format_gray}=[circle,draw,thin,fill=gray]
        \tikzstyle{format_rect}=[rectangle,draw,thin,fill=white,align=center]
        \tikzstyle{arrowstyle} = [->,thick]
        \tikzstyle{network} = [rectangle, minimum width = 3cm, minimum height = 1cm, text centered, draw = black,align=center,rounded corners,fill=green_so,fill opacity=0.5,text opacity=1]
        \tikzstyle{training_batch} = [trapezium, trapezium left angle = 30, trapezium right angle = 150, minimum width = 3cm, text centered, draw = black, fill = cyan_so, fill opacity=0.3,text opacity=1,align=center]		
        \tikzstyle{class_features} = [trapezium, trapezium left angle = 30, trapezium right angle = 150, minimum width = 3cm, text centered, draw = black, fill = cyan_so, fill opacity=0.3,text opacity=1,align=center]
        \tikzstyle{pixel} = [rectangle, draw = black, fill = orange_so, fill opacity=0.5,text opacity=0,align=center]	
        \tikzstyle{pixel_red} = [rectangle, draw = black, fill = red_so, fill opacity=1,text opacity=0,align=center]	
        \tikzstyle{feature} = [rectangle, draw = black, fill = orange_so, fill opacity=0.3,text opacity=0,align=center,rounded corners]	
        \tikzstyle{feature_sfp} = [rectangle, draw = black, fill = violet_so, fill opacity=0.3,text opacity=0,align=center,rounded corners]					
        \tikzstyle{arrow1} = [thick, ->, >= stealth]
        \tikzstyle{arrow1_thick} = [thick, ->, >= stealth, line width=1.2pt]
        \tikzstyle{arrow2} = [thick, dashed, ->, >= stealth]
        \tikzstyle{thick_line} = [thick, line width=1.5pt]
        \tikzstyle{channel} = [fill=white,fill opacity = 0.7, rounded corners=3pt]
        \tikzstyle{channel_shadow} = [fill = gray_so, fill opacity = 0.1, rounded corners]
        \tikzstyle{channel_selected} = [fill = orange_so, fill opacity = 0.5]

        \tikzstyle{dna} = [decoration={coil}, decorate, thick, decoration={aspect=0, segment length=0.5*0.87cm, post length=0.,pre length=0.}]
        \scalebox{0.87} 
        {
        \begin{tikzpicture}[auto,>=latex', thin, start chain=going below, every join/.style={norm}]
                \definecolor{gray_so}{RGB}{88,110,117}
                \definecolor{lightgray_so}{RGB}{207,221,221}
                \definecolor{yellow_so}{RGB}{181,137,0}
                \definecolor{cyan_so}{RGB}{42,161,152}
                \definecolor{orange_so}{RGB}{203,75,22}
                \definecolor{green_so}{RGB}{133,153,0}
                \definecolor{red_so}{RGB}{220,50,47}
                \definecolor{magenta_so}{RGB}{211,54,130}
                \definecolor{violet_so}{RGB}{108,113,196}
                \definecolor{yellow_ad}{RGB}{242,228,201}
                \definecolor{pink_ad}{RGB}{242,182,160}
                \definecolor{green_ad}{RGB}{146,195,185}
                \definecolor{dgreen_ad}{RGB}{104,166,148}
                \definecolor{purple_ad}{RGB}{115,72,88}
                \definecolor{att_blue}{RGB}{185,233,248}
                \definecolor{nature_orange}{RGB}{252,140,98}
                \definecolor{nature_green}{RGB}{102,195,170}
                \definecolor{nature_blue}{RGB}{142,160,204}
                \definecolor{nature_yellow}{RGB}{253,184,55}
                \useasboundingbox  (0,0) rectangle (7.5,3.6);

                \scope[transform canvas={scale=1}]
                
                \coordinate (zero) at (0,0);
                \coordinate (half) at (0,1.8);

                \coordinate (halfDnaCenter) at ($(half)+(0.5,0)$);

                \node at ($(halfDnaCenter)+(0.5,1)$) {$\bm{C}$};
                \node at ($(halfDnaCenter)+(0.5,-1)$) {$\bm{p}$};

                \filldraw[channel_selected,fill=nature_green,fill opacity = 0.1] ($(halfDnaCenter)+(0,-0.25)+(-0.125,0)$) rectangle ($(halfDnaCenter)+(0.25,0)+(-0.125,0)$);
                \filldraw[channel_selected,fill=nature_green,fill opacity = 0.1] ($(halfDnaCenter)+(0,0)+(-0.125,0)$) rectangle ($(halfDnaCenter)+(0.25,0.25)+(-0.125,0)$);
                \filldraw[channel_selected,fill=nature_green,fill opacity = 0.05] ($(halfDnaCenter)+(0,0.25)+(-0.125,0)$) rectangle ($(halfDnaCenter)+(0.25,0.5)+(-0.125,0)$);
                \filldraw[channel_selected,fill=nature_green,fill opacity = 0.75] ($(halfDnaCenter)+(0,0.5)+(-0.125,0)$) rectangle ($(halfDnaCenter)+(0.25,0.75)+(-0.125,0)$);
                \filldraw[channel_selected,fill=gray,fill opacity = 0.] ($(halfDnaCenter)+(0,-0.75)+(-0.125,0)$) rectangle ($(halfDnaCenter)+(0.25,-0.5)+(-0.125,0)$);
                \filldraw[channel_selected,fill=gray,fill opacity = 0.] ($(halfDnaCenter)+(0,-0.75)+(0.-0.375,0)$) rectangle ($(halfDnaCenter)+(0.25,-0.5)+(0.375,0)$);

                \filldraw[channel_selected,fill=nature_orange,fill opacity = 0.8] ($(halfDnaCenter)+(0,-0.25)+(0.125,0)$) rectangle ($(halfDnaCenter)+(0.25,0)+(0.125,0)$);
                \filldraw[channel_selected,fill=nature_orange,fill opacity = 0.1] ($(halfDnaCenter)+(0,0)+(0.125,0)$) rectangle ($(halfDnaCenter)+(0.25,0.25)+(0.125,0)$);
                \filldraw[channel_selected,fill=nature_orange,fill opacity = 0.05] ($(halfDnaCenter)+(0,0.25)+(0.125,0)$) rectangle ($(halfDnaCenter)+(0.25,0.5)+(0.125,0)$);
                \filldraw[channel_selected,fill=nature_orange,fill opacity = 0.05] ($(halfDnaCenter)+(0,0.5)+(0.125,0)$) rectangle ($(halfDnaCenter)+(0.25,0.75)+(0.125,0)$);
                \filldraw[channel_selected,fill=gray,fill opacity = 0.] ($(halfDnaCenter)+(0,-0.75)+(0.125,0)$) rectangle ($(halfDnaCenter)+(0.25,-0.5)+(0.125,0)$);

                \filldraw[channel_selected,fill=nature_green,fill opacity = 0.01] ($(halfDnaCenter)+(0,-0.25)+(0.375,0)$) rectangle ($(halfDnaCenter)+(0.25,0)+(0.375,0)$);
                \filldraw[channel_selected,fill=nature_green,fill opacity = 0.05] ($(halfDnaCenter)+(0,0)+(0.375,0)$) rectangle ($(halfDnaCenter)+(0.25,0.25)+(0.375,0)$);
                \filldraw[channel_selected,fill=nature_green,fill opacity = 0.04] ($(halfDnaCenter)+(0,0.25)+(0.375,0)$) rectangle ($(halfDnaCenter)+(0.25,0.5)+(0.375,0)$);
                \filldraw[channel_selected,fill=nature_green,fill opacity = 0.9] ($(halfDnaCenter)+(0,0.5)+(0.375,0)$) rectangle ($(halfDnaCenter)+(0.25,0.75)+(0.375,0)$);
                \filldraw[channel_selected,fill=gray,fill opacity = 0.5] ($(halfDnaCenter)+(0,-0.75)+(0.375,0)$) rectangle ($(halfDnaCenter)+(0.25,-0.5)+(0.375,0)$);

                \filldraw[channel_selected,fill=nature_blue,fill opacity = 0.1] ($(halfDnaCenter)+(0,-0.25)+(0.625,0)$) rectangle ($(halfDnaCenter)+(0.25,0)+(0.625,0)$);
                \filldraw[channel_selected,fill=nature_blue,fill opacity = 0.1] ($(halfDnaCenter)+(0,0)+(0.625,0)$) rectangle ($(halfDnaCenter)+(0.25,0.25)+(0.625,0)$);
                \filldraw[channel_selected,fill=nature_blue,fill opacity = 0.6] ($(halfDnaCenter)+(0,0.25)+(0.625,0)$) rectangle ($(halfDnaCenter)+(0.25,0.5)+(0.625,0)$);
                \filldraw[channel_selected,fill=nature_blue,fill opacity = 0.2] ($(halfDnaCenter)+(0,0.5)+(0.625,0)$) rectangle ($(halfDnaCenter)+(0.25,0.75)+(0.625,0)$);
                \filldraw[channel_selected,fill=gray,fill opacity = 0.] ($(halfDnaCenter)+(0,-0.75)+(0.625,0)$) rectangle ($(halfDnaCenter)+(0.25,-0.5)+(0.625,0)$);

                \filldraw[channel_selected,fill=nature_yellow,fill opacity = 0.1] ($(halfDnaCenter)+(0,-0.25)+(0.875,0)$) rectangle ($(halfDnaCenter)+(0.25,0)+(0.875,0)$);
                \filldraw[channel_selected,fill=nature_yellow,fill opacity = 0.8] ($(halfDnaCenter)+(0,0)+(0.875,0)$) rectangle ($(halfDnaCenter)+(0.25,0.25)+(0.875,0)$);
                \filldraw[channel_selected,fill=nature_yellow,fill opacity = 0.0] ($(halfDnaCenter)+(0,0.25)+(0.875,0)$) rectangle ($(halfDnaCenter)+(0.25,0.5)+(0.875,0)$);
                \filldraw[channel_selected,fill=nature_yellow,fill opacity = 0.1] ($(halfDnaCenter)+(0,0.5)+(0.875,0)$) rectangle ($(halfDnaCenter)+(0.25,0.75)+(0.875,0)$);
                \filldraw[channel_selected,fill=gray,fill opacity = 0.] ($(halfDnaCenter)+(0,-0.75)+(0.875,0)$) rectangle ($(halfDnaCenter)+(0.25,-0.5)+(0.875,0)$);

                \coordinate (cCenterRight) at ($(halfDnaCenter)+(0,0.25)+(1.125,0)$);
                \coordinate (profileCenterRight) at ($(cCenterRight)+(0,-0.875)$);

                \draw[line width=1.2, rounded corners] ($(cCenterRight)+(0.1,0)+(0,0.5)$) -- ($(cCenterRight)+(0.2,0)+(0,0.5)$) -- ($(cCenterRight)+(0.2,0)+(0,-0.25)$) -- ($(cCenterRight)+(0.4,0)+(0,-0.25)$);
                \draw[line width=1.2, rounded corners] ($(profileCenterRight)+(0.1,0)+(0,-0.125)$) -- ($(profileCenterRight)+(0.2,0)+(0,-0.125)$) -- ($(cCenterRight)+(0.2,0)+(0,-0.25)$) -- ($(cCenterRight)+(0.4,0)+(0,-0.25)$);

                \coordinate (upperModelCenter) at ($(cCenterRight)+(0.1,0)+(0,0.5)+(2,0)$);
                \coordinate (lowerModelCenter) at ($(profileCenterRight)+(0.1,0)+(0,-0.125)+(2,0)$);
                \coordinate (modelCenter) at ($(upperModelCenter)!0.5!(lowerModelCenter)$);

                \draw[line width=1.2, rounded corners] ($(cCenterRight)+(0.4,0)+(0,-0.25)$) -- ($(modelCenter)+(-1-0.6,0)$);
                \draw[arrow1_thick, rounded corners] ($(modelCenter)+(-1-0.6,0)$) -- ($(modelCenter)+(-1-0.4,0)$) -- ($(upperModelCenter)+(-1-0.4,0)$) -- ($(upperModelCenter)+(-1-0.1,0)$);
                \draw[arrow1_thick, rounded corners, dashed] ($(modelCenter)+(-1-0.6,0)$) -- ($(modelCenter)+(-1-0.4,0)$) -- ($(lowerModelCenter)+(-1-0.4,0)$) -- ($(lowerModelCenter)+(-1-0.1,0)$);

                \filldraw[channel, fill=nature_green] ($(upperModelCenter)+(-1,-0.618)$) rectangle ($(upperModelCenter)+(1,0.618)$);
                \node at ($(upperModelCenter)+(0,0.25)$) {differentiable};
                \node at ($(upperModelCenter)+(0,-0.25)$) {$\mathrm{DIDS}(\cdot,\cdot;\theta)$};
                \filldraw[channel, fill=white] ($(lowerModelCenter)+(-1,-0.618)$) rectangle ($(lowerModelCenter)+(1,0.618)$);
                \node at ($(lowerModelCenter)+(0,0.25)$) {conventional};
                \node at ($(lowerModelCenter)+(0,-0.25)$) {$\mathrm{CIDS}(\cdot,\cdot)$};

                \coordinate (weightsCenter) at ($(upperModelCenter)!0.5!(lowerModelCenter)$);

                \coordinate (upperFeatureCenterLeft) at ($(upperModelCenter)+(2.1,0)$);
                \coordinate (lowerFeatureCenterLeft) at ($(lowerModelCenter)+(2.1,0)$);
                \coordinate (upperFeatureCenter) at ($(upperFeatureCenterLeft)+(0.125,-0.25)$);
                \coordinate (lowerFeatureCenter) at ($(lowerFeatureCenterLeft)+(0.125,-0.25)$);

                \node at ($(upperFeatureCenterLeft)+(0.625,0.75)$) {$\hat{\bm{C}}_{\mathrm{DIDS}}$};
                \node at ($(lowerFeatureCenterLeft)+(0.625,-0.75)$) {$\hat{\bm{C}}_{\mathrm{CIDS}}$};
                \draw[arrow1_thick, rounded corners] ($(upperModelCenter)+(1.1,0)$) -- ($(upperFeatureCenterLeft)+(-0.1,0)$);
                \draw[arrow1_thick, rounded corners, dashed] ($(lowerModelCenter)+(1.1,0)$) -- ($(lowerFeatureCenterLeft)+(-0.1,0)$);

                \filldraw[channel_selected,fill=nature_green,fill opacity = 0.1] ($(upperFeatureCenter)+(0,-0.25)+(-0.125,0)$) rectangle ($(upperFeatureCenter)+(0.25,0)+(-0.125,0)$);
                \filldraw[channel_selected,fill=nature_green,fill opacity = 0.1] ($(upperFeatureCenter)+(0,0)+(-0.125,0)$) rectangle ($(upperFeatureCenter)+(0.25,0.25)+(-0.125,0)$);
                \filldraw[channel_selected,fill=nature_green,fill opacity = 0.05] ($(upperFeatureCenter)+(0,0.25)+(-0.125,0)$) rectangle ($(upperFeatureCenter)+(0.25,0.5)+(-0.125,0)$);
                \filldraw[channel_selected,fill=nature_green,fill opacity = 0.75] ($(upperFeatureCenter)+(0,0.5)+(-0.125,0)$) rectangle ($(upperFeatureCenter)+(0.25,0.75)+(-0.125,0)$);

                \filldraw[channel_selected,fill=nature_orange,fill opacity = 0.8] ($(upperFeatureCenter)+(0,-0.25)+(0.125,0)$) rectangle ($(upperFeatureCenter)+(0.25,0)+(0.125,0)$);
                \filldraw[channel_selected,fill=nature_orange,fill opacity = 0.1] ($(upperFeatureCenter)+(0,0)+(0.125,0)$) rectangle ($(upperFeatureCenter)+(0.25,0.25)+(0.125,0)$);
                \filldraw[channel_selected,fill=nature_orange,fill opacity = 0.05] ($(upperFeatureCenter)+(0,0.25)+(0.125,0)$) rectangle ($(upperFeatureCenter)+(0.25,0.5)+(0.125,0)$);
                \filldraw[channel_selected,fill=nature_orange,fill opacity = 0.05] ($(upperFeatureCenter)+(0,0.5)+(0.125,0)$) rectangle ($(upperFeatureCenter)+(0.25,0.75)+(0.125,0)$);

                \filldraw[channel_selected,fill=nature_green,fill opacity = 0.1] ($(upperFeatureCenter)+(0,-0.25)+(0.375,0)$) rectangle ($(upperFeatureCenter)+(0.25,0)+(0.375,0)$);
                \filldraw[channel_selected,fill=nature_green,fill opacity = 0.5] ($(upperFeatureCenter)+(0,0)+(0.375,0)$) rectangle ($(upperFeatureCenter)+(0.25,0.25)+(0.375,0)$);
                \filldraw[channel_selected,fill=nature_green,fill opacity = 0.2] ($(upperFeatureCenter)+(0,0.25)+(0.375,0)$) rectangle ($(upperFeatureCenter)+(0.25,0.5)+(0.375,0)$);
                \filldraw[channel_selected,fill=nature_green,fill opacity = 0.1] ($(upperFeatureCenter)+(0,0.5)+(0.375,0)$) rectangle ($(upperFeatureCenter)+(0.25,0.75)+(0.375,0)$);

                \filldraw[channel_selected,fill=nature_blue,fill opacity = 0.1] ($(upperFeatureCenter)+(0,-0.25)+(0.625,0)$) rectangle ($(upperFeatureCenter)+(0.25,0)+(0.625,0)$);
                \filldraw[channel_selected,fill=nature_blue,fill opacity = 0.1] ($(upperFeatureCenter)+(0,0)+(0.625,0)$) rectangle ($(upperFeatureCenter)+(0.25,0.25)+(0.625,0)$);
                \filldraw[channel_selected,fill=nature_blue,fill opacity = 0.6] ($(upperFeatureCenter)+(0,0.25)+(0.625,0)$) rectangle ($(upperFeatureCenter)+(0.25,0.5)+(0.625,0)$);
                \filldraw[channel_selected,fill=nature_blue,fill opacity = 0.2] ($(upperFeatureCenter)+(0,0.5)+(0.625,0)$) rectangle ($(upperFeatureCenter)+(0.25,0.75)+(0.625,0)$);

                \filldraw[channel_selected,fill=nature_yellow,fill opacity = 0.1] ($(upperFeatureCenter)+(0,-0.25)+(0.875,0)$) rectangle ($(upperFeatureCenter)+(0.25,0)+(0.875,0)$);
                \filldraw[channel_selected,fill=nature_yellow,fill opacity = 0.8] ($(upperFeatureCenter)+(0,0)+(0.875,0)$) rectangle ($(upperFeatureCenter)+(0.25,0.25)+(0.875,0)$);
                \filldraw[channel_selected,fill=nature_yellow,fill opacity = 0.0] ($(upperFeatureCenter)+(0,0.25)+(0.875,0)$) rectangle ($(upperFeatureCenter)+(0.25,0.5)+(0.875,0)$);
                \filldraw[channel_selected,fill=nature_yellow,fill opacity = 0.1] ($(upperFeatureCenter)+(0,0.5)+(0.875,0)$) rectangle ($(upperFeatureCenter)+(0.25,0.75)+(0.875,0)$);

                \filldraw[channel_selected,fill=nature_green,fill opacity = 0.1] ($(lowerFeatureCenter)+(0,-0.25)+(-0.125,0)$) rectangle ($(lowerFeatureCenter)+(0.25,0)+(-0.125,0)$);
                \filldraw[channel_selected,fill=nature_green,fill opacity = 0.1] ($(lowerFeatureCenter)+(0,0)+(-0.125,0)$) rectangle ($(lowerFeatureCenter)+(0.25,0.25)+(-0.125,0)$);
                \filldraw[channel_selected,fill=nature_green,fill opacity = 0.05] ($(lowerFeatureCenter)+(0,0.25)+(-0.125,0)$) rectangle ($(lowerFeatureCenter)+(0.25,0.5)+(-0.125,0)$);
                \filldraw[channel_selected,fill=nature_green,fill opacity = 0.75] ($(lowerFeatureCenter)+(0,0.5)+(-0.125,0)$) rectangle ($(lowerFeatureCenter)+(0.25,0.75)+(-0.125,0)$);

                \filldraw[channel_selected,fill=nature_orange,fill opacity = 0.8] ($(lowerFeatureCenter)+(0,-0.25)+(0.125,0)$) rectangle ($(lowerFeatureCenter)+(0.25,0)+(0.125,0)$);
                \filldraw[channel_selected,fill=nature_orange,fill opacity = 0.1] ($(lowerFeatureCenter)+(0,0)+(0.125,0)$) rectangle ($(lowerFeatureCenter)+(0.25,0.25)+(0.125,0)$);
                \filldraw[channel_selected,fill=nature_orange,fill opacity = 0.05] ($(lowerFeatureCenter)+(0,0.25)+(0.125,0)$) rectangle ($(lowerFeatureCenter)+(0.25,0.5)+(0.125,0)$);
                \filldraw[channel_selected,fill=nature_orange,fill opacity = 0.05] ($(lowerFeatureCenter)+(0,0.5)+(0.125,0)$) rectangle ($(lowerFeatureCenter)+(0.25,0.75)+(0.125,0)$);

                \filldraw[channel_selected,fill=nature_green,fill opacity = 0.03] ($(lowerFeatureCenter)+(0,-0.25)+(0.375,0)$) rectangle ($(lowerFeatureCenter)+(0.25,0)+(0.375,0)$);
                \filldraw[channel_selected,fill=nature_green,fill opacity = 0.02] ($(lowerFeatureCenter)+(0,0)+(0.375,0)$) rectangle ($(lowerFeatureCenter)+(0.25,0.25)+(0.375,0)$);
                \filldraw[channel_selected,fill=nature_green,fill opacity = 0.05] ($(lowerFeatureCenter)+(0,0.25)+(0.375,0)$) rectangle ($(lowerFeatureCenter)+(0.25,0.5)+(0.375,0)$);
                \filldraw[channel_selected,fill=nature_green,fill opacity = 0.9] ($(lowerFeatureCenter)+(0,0.5)+(0.375,0)$) rectangle ($(lowerFeatureCenter)+(0.25,0.75)+(0.375,0)$);

                \filldraw[channel_selected,fill=nature_blue,fill opacity = 0.1] ($(lowerFeatureCenter)+(0,-0.25)+(0.625,0)$) rectangle ($(lowerFeatureCenter)+(0.25,0)+(0.625,0)$);
                \filldraw[channel_selected,fill=nature_blue,fill opacity = 0.1] ($(lowerFeatureCenter)+(0,0)+(0.625,0)$) rectangle ($(lowerFeatureCenter)+(0.25,0.25)+(0.625,0)$);
                \filldraw[channel_selected,fill=nature_blue,fill opacity = 0.6] ($(lowerFeatureCenter)+(0,0.25)+(0.625,0)$) rectangle ($(lowerFeatureCenter)+(0.25,0.5)+(0.625,0)$);
                \filldraw[channel_selected,fill=nature_blue,fill opacity = 0.2] ($(lowerFeatureCenter)+(0,0.5)+(0.625,0)$) rectangle ($(lowerFeatureCenter)+(0.25,0.75)+(0.625,0)$);

                \filldraw[channel_selected,fill=nature_yellow,fill opacity = 0.1] ($(lowerFeatureCenter)+(0,-0.25)+(0.875,0)$) rectangle ($(lowerFeatureCenter)+(0.25,0)+(0.875,0)$);
                \filldraw[channel_selected,fill=nature_yellow,fill opacity = 0.8] ($(lowerFeatureCenter)+(0,0)+(0.875,0)$) rectangle ($(lowerFeatureCenter)+(0.25,0.25)+(0.875,0)$);
                \filldraw[channel_selected,fill=nature_yellow,fill opacity = 0.0] ($(lowerFeatureCenter)+(0,0.25)+(0.875,0)$) rectangle ($(lowerFeatureCenter)+(0.25,0.5)+(0.875,0)$);
                \filldraw[channel_selected,fill=nature_yellow,fill opacity = 0.1] ($(lowerFeatureCenter)+(0,0.5)+(0.875,0)$) rectangle ($(lowerFeatureCenter)+(0.25,0.75)+(0.875,0)$);
                \endscope
        \end{tikzpicture}	
        }
}
\caption{\label{fig:diff-ids} The differentiable IDS channel. 
The $\hat{\bm{C}}_{\mathrm{DIDS}}$ and $\hat{\bm{C}}_{\mathrm{CIDS}}$ are generated by 
the differentiable and conventional IDS channels, respectively. 
Optimizing the difference between $\hat{\bm{C}}_{\mathrm{DIDS}}$ and $\hat{\bm{C}}_{\mathrm{CIDS}}$ 
trains the differentiable channel. 
}
\vskip -0.in
\end{figure}

As illustrated in \cref{fig:diff-ids}, 
both the conventional IDS channel $\mathrm{CIDS}$ and the simulated IDS channel $\mathrm{DIDS}$ 
take the sequence $\bm{C}$ of probability vectors and an error profile $\bm{p}$ as their inputs. 
The error profile consists of a sequence of letters that record the types of errors encountered while processing $\bm{C}$. 
Complicated IDS channels can be deduced by specifying the rules for generating error profiles. 
The probability sequence $\bm{C}$ is expected to be modified by the simulated IDS channel to 
$\hat{\bm{C}}_{\mathrm{DIDS}} = \mathrm{DIDS}(\bm{C},\bm{p};\theta)$
according to the error profile $\bm{p}$ in the upper stream of \cref{fig:diff-ids}. 
In the lower stream, the sequence $\bm{C}$ is modified as 
$\hat{\bm{C}}_{\mathrm{CIDS}} = \mathrm{CIDS}(\bm{C},\bm{p})$ 
with respect to 
the error profile $\bm{p}$ using the previously defined promoted IDS operations. 

To train the model $\mathrm{DIDS}(\cdot,\cdot;\theta)$, 
the Kullback–Leibler divergence~\cite{kullback1997information} of 
$\hat{\bm{C}}_{\mathrm{DIDS}}$ from $\hat{\bm{C}}_{\mathrm{CIDS}}$ 
can be utilized as the optimization target
\begin{equation}\label{eqn:ids-kld-loss}
\mathcal{L}_{\mathrm{KLD}}(\hat{\bm{C}}_{\mathrm{DIDS}},\hat{\bm{C}}_{\mathrm{CIDS}}) 
= \frac{1}{k}\sum_{i} \hat{\bm{\pi}}_{i\mathrm{CIDS}}^T \log \frac{\hat{\bm{\pi}}_{i\mathrm{CIDS}}}
{\hat{\bm{\pi}}_{i\mathrm{DIDS}}}. 
\end{equation} 
By optimizing \cref{eqn:ids-kld-loss} on randomly generated probability vector sequences $\bm{C}$ 
and error profiles $\bm{p}$, 
the parameters $\theta$ of the 
differentiable IDS channel are trained to $\hat{\theta}$. 
Following this, the model $\mathrm{DIDS}(\cdot,\cdot;\hat{\theta})$ simulates 
the conventional IDS channel $\mathrm{CIDS}(\cdot,\cdot)$. 
The significance of such an IDS channel lies in its differentiability. 
Once optimized independently, the parameters of the IDS channel are fixed for downstream applications. 
In the following text, we use $\mathrm{DIDS}(\cdot,\cdot)$ to refer to the trained IDS channel for simplicity. 

In practice, the differentiable IDS channel is implemented as a sequence-to-sequence model, 
employing one-layer Transformers for both its encoder and decoder.\footnote{
    Here, the encoder and decoder refer specifically to the modules of the sequence-to-sequence model, 
    not the modules of the autoencoder. 
    We trust that readers will be able to distinguish between them based on the context. 
}
The model takes a padded vector sequence and error profile, whose embeddings are concatenated 
along the feature dimension as its input. 
To generate the output, that represents the sequence with errors, 
learnable position embedding vectors are utilized as the queries 
(omitted from \cref{fig:diff-ids}). 

\section{THEA-Code}\label{sec:thea-code}

\subsection{Framework}\label{subsec:framework}
The flowchart of the proposed code is illustrated in \cref{fig:thea-code}. 
Based on the principles of DNA-based storage, which synthesizes DNA molecules of fixed length, 
the proposed model is designed to handle source sequences and codewords of constant lengths. 
Essentially, 
the proposed method encodes source sequences into codewords; 
the IDS channel introduces IDS errors to these codewords; 
and a decoder is employed to reconstruct the recovered sequences according to the corrupted codewords. 

Let $f_\mathrm{en}(\cdot;\phi)$ denote the encoder, 
where $\phi$ represents the encoder's parameters. 
The source sequence $\bm{s}$ is first encoded into the codeword $\bm{c}=f_\mathrm{en}(\bm{s};\phi)$ 
by the encoder,\footnote{For simplicity, 
we do not distinguish between notations for sequences represented as letters, 
one-hot vectors, or probability vectors in the following text.}
where the codeword $\bm{c}$ is obtained using \cref{eqn:gsequation} during the training phase and $\argmax$ during the testing phase. 
Next, a random error profile $\bm{p}$ is generated, which records the positions and types of errors 
that will occur on codeword $\bm{c}$. 
Given the error profile $\bm{p}$, the codeword $\bm{c}$ is transformed into the corrupted codeword 
$\hat{\bm{c}} = \mathrm{DIDS}(\bm{c},\bm{p};\hat{\theta})$ 
by the simulated differentiable IDS channel, implemented as a sequence-to-sequence model with trained parameters $\hat{\theta}$. 
Finally, a decoder $f_\mathrm{de}(\cdot;\psi)$ with parameters $\psi$ decodes the corrupted codeword $\hat{\bm{c}}$ 
back into the recovered sequence $\hat{\bm{s}} = f_\mathrm{de}(\hat{\bm{c}};\psi)$. 

\begin{figure*}[htb!]
    \vskip -0.in
        \centering
{\linespread{1}
        \centering
        \tikzstyle{format}=[circle,draw,thin,fill=white]
        \tikzstyle{format_gray}=[circle,draw,thin,fill=gray]
        \tikzstyle{format_rect}=[rectangle,draw,thin,fill=white,align=center]
        \tikzstyle{arrowstyle} = [->,thick]
        \tikzstyle{network} = [rectangle, minimum width = 3cm, minimum height = 1cm, text centered, draw = black,align=center,rounded corners,fill=green_so,fill opacity=0.5,text opacity=1]
        \tikzstyle{training_batch} = [trapezium, trapezium left angle = 30, trapezium right angle = 150, minimum width = 3cm, text centered, draw = black, fill = cyan_so, fill opacity=0.3,text opacity=1,align=center]		
        \tikzstyle{class_features} = [trapezium, trapezium left angle = 30, trapezium right angle = 150, minimum width = 3cm, text centered, draw = black, fill = cyan_so, fill opacity=0.3,text opacity=1,align=center]
        \tikzstyle{pixel} = [rectangle, draw = black, fill = orange_so, fill opacity=0.5,text opacity=0,align=center]	
        \tikzstyle{pixel_red} = [rectangle, draw = black, fill = red_so, fill opacity=1,text opacity=0,align=center]	
        \tikzstyle{feature} = [rectangle, draw = black, fill = orange_so, fill opacity=0.3,text opacity=0,align=center,rounded corners]	
        \tikzstyle{feature_sfp} = [rectangle, draw = black, fill = violet_so, fill opacity=0.3,text opacity=0,align=center,rounded corners]					
        \tikzstyle{arrow1} = [thick, ->, >= stealth]
        \tikzstyle{arrow1_thick} = [thick, ->, >= stealth, line width=1.2pt]
        \tikzstyle{arrow2} = [thick, dashed, ->, >= stealth]
        \tikzstyle{thick_line} = [thick, line width=1.5pt]
        \tikzstyle{channel} = [fill=white,fill opacity = 0.7, rounded corners=3pt]
        \tikzstyle{channel_shadow} = [fill = gray_so, fill opacity = 0.1, rounded corners]
        \tikzstyle{channel_selected} = [fill = orange_so, fill opacity = 0.5]

        \tikzstyle{dna} = [decoration={coil}, decorate, thick, decoration={aspect=0, segment length=0.5*0.87cm, post length=0.,pre length=0.}]
        \scalebox{0.77} 
        {
        \begin{tikzpicture}[auto,>=latex', thin, start chain=going below, every join/.style={norm}]
                \definecolor{gray_so}{RGB}{88,110,117}
                \definecolor{lightgray_so}{RGB}{207,221,221}
                \definecolor{yellow_so}{RGB}{181,137,0}
                \definecolor{cyan_so}{RGB}{42,161,152}
                \definecolor{orange_so}{RGB}{203,75,22}
                \definecolor{green_so}{RGB}{133,153,0}
                \definecolor{red_so}{RGB}{220,50,47}
                \definecolor{magenta_so}{RGB}{211,54,130}
                \definecolor{violet_so}{RGB}{108,113,196}
                \definecolor{yellow_ad}{RGB}{242,228,201}
                \definecolor{pink_ad}{RGB}{242,182,160}
                \definecolor{green_ad}{RGB}{146,195,185}
                \definecolor{dgreen_ad}{RGB}{104,166,148}
                \definecolor{purple_ad}{RGB}{115,72,88}
                \definecolor{att_blue}{RGB}{185,233,248}
                \definecolor{nature_orange}{RGB}{252,140,98}
                \definecolor{nature_green}{RGB}{102,195,170}
                \definecolor{nature_blue}{RGB}{142,160,204}
                \definecolor{nature_yellow}{RGB}{253,184,55}
                \useasboundingbox (0,0) rectangle (14,4);

                \scope[transform canvas={scale=1}]
                
                \coordinate (zero) at (0,0);
                \coordinate (half) at (0,2.2);
                \coordinate (upperhalf) at ($(half)+(0,1)$);
                \coordinate (lowerhalf) at ($(half)+(0,-1)$);

                \coordinate (upperDnaCenter) at ($(upperhalf)+(0.6,0)$);
                \coordinate (lowerDnaCenter) at ($(lowerhalf)+(0.6,0)$);
                \coordinate (halfDnaCenter) at ($(half)+(0.6,0)$);

                \filldraw[channel_selected,fill=nature_blue,fill opacity = 0.5] ($(halfDnaCenter)+(0,0)+(-0.125,-0.5)$) rectangle ($(halfDnaCenter)+(0.25,0.25)+(-0.125,-0.5)$);
                \filldraw[channel_selected,fill=nature_green,fill opacity = 0.5] ($(halfDnaCenter)+(0,0.25)+(-0.125,-0.5)$) rectangle ($(halfDnaCenter)+(0.25,0.5)+(-0.125,-0.5)$);
                \filldraw[channel_selected,fill=nature_yellow,fill opacity = 0.5] ($(halfDnaCenter)+(0,0.5)+(-0.125,-0.5)$) rectangle ($(halfDnaCenter)+(0.25,0.75)+(-0.125,-0.5)$);
                \filldraw[channel_selected,fill=nature_orange,fill opacity = 0.5] ($(halfDnaCenter)+(0,0.75)+(-0.125,-0.5)$) rectangle ($(halfDnaCenter)+(0.25,1)+(-0.125,-0.5)$);
                \node at ($(halfDnaCenter)+(-0.4,0.2)$) {$\bm{s}$};
                \node at ($(halfDnaCenter)+(0,-0.9)$) {source};
                \node at ($(halfDnaCenter)+(0,-1.2)$) {information};

                \coordinate (upperModelCenter) at ($(upperDnaCenter)+(2.125,0)$);
                \coordinate (lowerModelCenter) at ($(lowerDnaCenter)+(2.125,0)$);
                \coordinate (encoderCenter) at ($(halfDnaCenter)+(2.125,0)$);
                \filldraw[channel, fill=nature_orange] ($(encoderCenter)+(-1,-0.618)$) rectangle ($(encoderCenter)+(1,0.618)$);
                \node at ($(encoderCenter)+(0,-0.0)$) {Encoder};
                \draw[arrow1_thick, rounded corners] ($(halfDnaCenter)+(0.125+0.1,0)$) -- ($(encoderCenter)+(-1.1,0)$);

                \coordinate (weightsCenter) at ($(upperModelCenter)!0.5!(lowerModelCenter)$);

                \coordinate (upperFeatureCenter) at ($(encoderCenter)+(2.125,0)$);
                \coordinate (featureCenter) at ($(encoderCenter)+(2.125,0)$);
                \coordinate (upperAuxCenter) at ($(upperFeatureCenter)+(0,0.75)$);
                \coordinate (upperCodewordCenter) at ($(upperFeatureCenter)+(0,-0.5)$);
                \coordinate (upperRealCodewordCenter) at ($(upperCodewordCenter)+(0,-0.125)$);

                \filldraw[channel_selected,fill=white] ($(upperAuxCenter)+(-0.125,-0.5)$) rectangle ($(upperAuxCenter)+(0.125,0.5)$);
                \filldraw[channel_selected,fill=nature_blue,fill opacity = 0.5] ($(upperAuxCenter)+(0,0)+(-0.125,-0.5)$) rectangle ($(upperAuxCenter)+(0.25,0.25)+(-0.125,-0.5)$);
                \filldraw[channel_selected,fill=nature_green,fill opacity = 0.5] ($(upperAuxCenter)+(0,0.25)+(-0.125,-0.5)$) rectangle ($(upperAuxCenter)+(0.25,0.5)+(-0.125,-0.5)$);
                \filldraw[channel_selected,fill=nature_yellow,fill opacity = 0.5] ($(upperAuxCenter)+(0,0.5)+(-0.125,-0.5)$) rectangle ($(upperAuxCenter)+(0.25,0.75)+(-0.125,-0.5)$);
                \filldraw[channel_selected,fill=nature_orange,fill opacity = 0.5] ($(upperAuxCenter)+(0,0.75)+(-0.125,-0.5)$) rectangle ($(upperAuxCenter)+(0.25,1)+(-0.125,-0.5)$);

                \filldraw[channel_selected,fill=nature_blue,fill opacity = 1] ($(upperCodewordCenter)+(0,0)+(-0.125,-0.5)$) rectangle ($(upperCodewordCenter)+(0.25,0.25)+(-0.125,-0.5)$);
                \filldraw[channel_selected,fill=nature_green,fill opacity = 1] ($(upperCodewordCenter)+(0,0.25)+(-0.125,-0.5)$) rectangle ($(upperCodewordCenter)+(0.25,0.5)+(-0.125,-0.5)$);
                \filldraw[channel_selected,fill=nature_orange,fill opacity = 1] ($(upperCodewordCenter)+(0,0.5)+(-0.125,-0.5)$) rectangle ($(upperCodewordCenter)+(0.25,0.75)+(-0.125,-0.5)$);
                \filldraw[channel_selected,fill=nature_green,fill opacity = 1] ($(upperCodewordCenter)+(0,0.75)+(-0.125,-0.5)$) rectangle ($(upperCodewordCenter)+(0.25,1)+(-0.125,-0.5)$);
                \filldraw[channel_selected,fill=nature_yellow,fill opacity = 1] ($(upperCodewordCenter)+(0,-0.25)+(-0.125,-0.5)$) rectangle ($(upperCodewordCenter)+(0.25,0)+(-0.125,-0.5)$);
                \filldraw[channel_selected,fill=white,fill opacity = 1, style=dotted, pattern=north west lines] ($(upperCodewordCenter)+(0,1)+(-0.125,-0.5)$) rectangle ($(upperCodewordCenter)+(0.25,1.25)+(-0.125,-0.5)$);
                \node at ($(upperAuxCenter)+(-0.4,0.2)$) {$\bm{r}$};
                \node at ($(upperAuxCenter)+(0,0.7)$) {auxiliary reconstruction};

                \node at ($(upperCodewordCenter)+(-0.4,0.25)$) {$\bm{c}$};
                \node at ($(upperCodewordCenter)+(0,-1)$) {codeword};
                \node at ($(upperCodewordCenter)+(0,-1)+(0.6,0.6)$) {\orangeso\textbf{GS}};

                \draw[arrow1_thick, rounded corners] ($(encoderCenter)+(1+0.1,0)$) -- ($(upperFeatureCenter)+(-0.125-0.1,0)$);

                \coordinate (idsCenter) at ($(featureCenter)+(2.125+0.1,0)$);

                \coordinate (lowerProfileCenter) at ($(idsCenter)+(0,-1.5)$);
                \filldraw[channel_selected,fill=gray,fill opacity = 0.2] ($(lowerProfileCenter)+(0,0)+(-0.25,-0.125)$) rectangle ($(lowerProfileCenter)+(0.25,0.25)+(-0.25,-0.125)$);
                \filldraw[channel_selected,fill=gray,fill opacity = 0.5] ($(lowerProfileCenter)+(-0.25,0)+(-0.25,-0.125)$) rectangle ($(lowerProfileCenter)+(0,0.25)+(-0.25,-0.125)$);
                \filldraw[channel_selected,fill=gray,fill opacity = 0.1] ($(lowerProfileCenter)+(-0.5,0)+(-0.25,-0.125)$) rectangle ($(lowerProfileCenter)+(-0.25,0.25)+(-0.25,-0.125)$);
                \filldraw[channel_selected,fill=gray,fill opacity = 0.8] ($(lowerProfileCenter)+(0.25,0)+(-0.25,-0.125)$) rectangle ($(lowerProfileCenter)+(0.5,0.25)+(-0.25,-0.125)$);
                \filldraw[channel_selected,fill=gray,fill opacity = 0.7] ($(lowerProfileCenter)+(0.5,0)+(-0.25,-0.125)$) rectangle ($(lowerProfileCenter)+(0.75,0.25)+(-0.25,-0.125)$);
                \filldraw[channel_selected,fill=gray,fill opacity = 0.7] ($(lowerProfileCenter)+(0.75,0)+(-0.25,-0.125)$) rectangle ($(lowerProfileCenter)+(1,0.25)+(-0.25,-0.125)$);
                \node at ($(lowerProfileCenter)+(-1,0)$) {$\bm{p}$};
                \node at ($(lowerProfileCenter)+(0,-0.4)$) {random error profile};
                \draw[line width=1.2, rounded corners] ($(lowerProfileCenter)+(-0.75,0.125+0.1)$) -- ($(lowerProfileCenter)+(-0.75,0.125+0.2)$) -- ($(lowerProfileCenter)+(0,0.125+0.2)$) -- ($(lowerProfileCenter)+(0,0.125+0.3)$);
                \draw[arrow1_thick, rounded corners] ($(lowerProfileCenter)+(0.75,0.125+0.1)$) -- ($(lowerProfileCenter)+(0.75,0.125+0.2)$) -- ($(lowerProfileCenter)+(0,0.125+0.2)$) -- ($(idsCenter)+(0,-0.618-0.1)$);

                \draw[line width=1.2, rounded corners] ($(upperRealCodewordCenter)+(0.125+0.1,-0.625)$) -- ($(upperRealCodewordCenter)+(0.125+0.2,-0.625)$) -- ($(upperRealCodewordCenter)+(0.125+0.2,0)$) -- ($(upperRealCodewordCenter)+(0.125+0.3,0)$);
                \draw[arrow1_thick, rounded corners] ($(upperRealCodewordCenter)+(0.125+0.1,0.625)$) -- ($(upperRealCodewordCenter)+(0.125+0.2,0.625)$) -- ($(upperRealCodewordCenter)+(0.125+0.2,0)$) -- ($(upperRealCodewordCenter)+(0.125+0.65,0)$) -- ($(idsCenter)+(-1.45,0)$) -- ($(idsCenter)+(-1.1,0)$);

                \filldraw[channel, fill=nature_green] ($(idsCenter)+(-1,-0.618)$) rectangle ($(idsCenter)+(1,0.618)$);
                \node at ($(idsCenter)+(0,0.33)$) {Differentiable};
                \node at ($(idsCenter)+(0,-0.0)$) {IDS};
                \node at ($(idsCenter)+(0,-0.33)$) {Channel};

                \coordinate (corruptCenter) at ($(idsCenter)+(2.125,0)$);
                \coordinate (corruptCenterAdjust) at ($(corruptCenter)+(0,0.125)$);
                \filldraw[channel_selected,fill=nature_yellow,fill opacity = 1] ($(corruptCenterAdjust)+(0,-0.25)+(-0.125,-0.5)$) rectangle ($(corruptCenterAdjust)+(0.25,0)+(-0.125,-0.5)$);
                \filldraw[channel_selected,fill=nature_blue,fill opacity = 1] ($(corruptCenterAdjust)+(0,0)+(-0.125,-0.5)$) rectangle ($(corruptCenterAdjust)+(0.25,0.25)+(-0.125,-0.5)$);
                \filldraw[channel_selected,fill=nature_yellow,fill opacity = 1] ($(corruptCenterAdjust)+(0,0.25)+(-0.125,-0.5)$) rectangle ($(corruptCenterAdjust)+(0.25,0.5)+(-0.125,-0.5)$);
                \filldraw[channel_selected,fill=nature_orange,fill opacity = 1] ($(corruptCenterAdjust)+(0,0.5)+(-0.125,-0.5)$) rectangle ($(corruptCenterAdjust)+(0.25,0.75)+(-0.125,-0.5)$);
                \filldraw[channel_selected,fill=nature_green,fill opacity = 1] ($(corruptCenterAdjust)+(0,0.75)+(-0.125,-0.5)$) rectangle ($(corruptCenterAdjust)+(0.25,1)+(-0.125,-0.5)$);
                \node at ($(corruptCenter)+(-0.4,0.4)$) {$\hat{\bm{c}}$};
                \node at ($(corruptCenter)+(0,-0.9)$) {corrupted};
                \node at ($(corruptCenter)+(0,-1.2)$) {codeword};
                \draw[arrow1_thick, rounded corners] ($(idsCenter)+(1+0.1,0)$) -- ($(corruptCenter)+(-0.125-0.1,0)$);

                \coordinate (decoderCenter) at ($(corruptCenter)+(2.125,0)$);
                \filldraw[channel, fill=nature_orange] ($(decoderCenter)+(-1,-0.618)$) rectangle ($(decoderCenter)+(1,0.618)$);
                \node at ($(decoderCenter)+(0,-0.0)$) {Decoder};
                \draw[arrow1_thick, rounded corners] ($(corruptCenter)+(0.125+0.1,0)$) -- ($(decoderCenter)+(-1-0.1,0)$);

                \coordinate (outputCenter) at ($(decoderCenter)+(2.125,0)$);
                \filldraw[channel_selected,fill=nature_blue,fill opacity = 0.5] ($(outputCenter)+(0,0)+(-0.125,-0.5)$) rectangle ($(outputCenter)+(0.25,0.25)+(-0.125,-0.5)$);
                \filldraw[channel_selected,fill=nature_green,fill opacity = 0.7] ($(outputCenter)+(0,0.25)+(-0.125,-0.5)$) rectangle ($(outputCenter)+(0.25,0.5)+(-0.125,-0.5)$);
                \filldraw[channel_selected,fill=nature_yellow,fill opacity = 0.5] ($(outputCenter)+(0,0.5)+(-0.125,-0.5)$) rectangle ($(outputCenter)+(0.25,0.75)+(-0.125,-0.5)$);
                \filldraw[channel_selected,fill=nature_orange,fill opacity = 0.7] ($(outputCenter)+(0,0.75)+(-0.125,-0.5)$) rectangle ($(outputCenter)+(0.25,1)+(-0.125,-0.5)$);
                \node at ($(outputCenter)+(-0.4,0.3)$) {$\hat{\bm{s}}$};
                \node at ($(outputCenter)+(0,-0.9)$) {recovered};
                \node at ($(outputCenter)+(0,-1.2)$) {information};
                \draw[arrow1_thick, rounded corners] ($(decoderCenter)+(1+0.1,0)$) -- ($(outputCenter)+(-0.125-0.1,0)$);

                \endscope
        \end{tikzpicture}	
        }
}
\caption{\label{fig:thea-code} The flowchart of THEA-Code, including the encoder, the pretrained IDS channel, 
and the decoder. All of these modules are implemented using Transformer-based models. 
The ``GS'' is where the disturbance based discretization applied in the pipeline. 
}
\vskip -0.in
\end{figure*}

Following this pipeline, a natural optimization target is the cross-entropy loss
\begin{equation}\label{eqn:ce-loss}
    \mathcal{L}_{\mathrm{CE}}(\hat{\bm{s}},\bm{s}) = - \sum_i \sum_j \mathds{1}_{j=s_i} \log{\hat{s}_{ij}}, 
\end{equation}
which evaluates the reconstruction disparity of the source sequence $\bm{s}$ (in its label representation) 
by the recovered sequence $\hat{\bm{s}}$ (in its one-hot probability distribution). 

However, merely optimizing such a loss function will not yield the desired outcomes. 
While the encoder and decoder of an autoencoder typically collaborate on a unified task in most applications, 
in this work, we expect them to follow distinct underlying logic. 
Particularly, when imposing constraints to enforce greater discreteness in the codeword, 
the joint training of the encoder and decoder becomes challenging, 
where the optimization of each relies on the other during the training phase. 

\subsection{Auxiliary reconstruction of source sequence by the encoder}\label{subsec:auxiliary}
To address the aforementioned issue, 
we introduce a supplementary task exclusively for the encoder, 
aimed at initializing it with some foundational logical capabilities. 
Inspired by the systematic code which embed the input message within the codeword,
a straightforward task for the encoder is to replicate the input sequence at the output, 
ensuring that the model preserves all information from its input without reduction.
With this in mind, we incorporate a reconstruction task into the encoder's training process. 

In practice, the encoder is designed to output a longer sequence, which is subsequently split 
into two parts: the codeword representation $\bm{c}$ and an auxiliary reconstruction $\bm{r}$ of the input source sequence, 
as shown in \cref{fig:thea-code}. 
The auxiliary reconstruction loss is calculated using the cross-entropy loss as
\begin{equation}\label{eqn:aux-loss}
\mathcal{L}_{\mathrm{Aux}}(\bm{r},\bm{s}) = - \sum_i \sum_j \mathds{1}_{j=s_i} \log{r_{ij}}, 
\end{equation}
which quantifies the difference between the reconstruction $\bm{r}$ (in its one-hot probability distribution) 
and the input sequence 
$\bm{s}$ (in its label representation). 

Considering that the auxiliary loss may not have negative effects on the encoder for its simple logic, 
we don't use a separate training stage for optimizing the $\mathcal{L}_{\mathrm{Aux}}$. 
The auxiliary loss defined in \cref{eqn:aux-loss} is 
incorporated into the overall loss function and applied consistently throughout the entire training phase. 

\subsection{The encoder and decoder}\label{subsec:autoencoder}
In this approach, both the encoder and decoder are implemented using Transformer-based sequence-to-sequence models. 
Each consists of (3+3)-layer Transformers with sinusoidal positional encoding.
The embedding of the DNA bases is implemented through a fully connected layer without bias 
to ensure compatibility with probability vectors. 
Learnable position index embeddings are employed to query the outputs. 

\subsection{Training phase}\label{subsec:train}
The training process is divided into two phases. 
Firstly, the differentiable IDS channel is fully trained by optimizing
\begin{equation}
    \hat{\theta} = \underset{\theta}{\arg\min}\, 
    \mathcal{L}_{\mathrm{KLD}}(\hat{\bm{C}}_{\mathrm{DIDS}},\hat{\bm{C}}_{\mathrm{CIDS}}) 
\end{equation}
on randomly generated codewords $\bm{c}$ and profiles $\bm{p}$. 
Once the differentiable IDS channel is trained, its parameters are fixed. 
The remaining components of the autoencoder are then trained by optimizing a weighted sum 
of \cref{eqn:ce-loss} and \cref{eqn:aux-loss}, 
\begin{equation}
    \hat{\phi},\hat{\psi} = \underset{\phi,\psi}{\arg\min}\,
    \mathcal{L}_{\mathrm{CE}}(\hat{\bm{s}},\bm{s}) + \mu \mathcal{L}_{\mathrm{Aux}}(\bm{r},\bm{s}), 
\end{equation}
where $\mu$ is a hyperparameter representing the weight of the auxiliary reconstruction loss. 
The autoencoder is trained on randomly generated input sequences $\bm{s}$ and profiles $\bm{p}$. 

\subsection{Testing phase}\label{subsec:test}
In the testing phase, the differentiable IDS channel is replaced with the conventional IDS channel. 
The process begins with the encoder mapping the source sequence $\bm{s}$ to the codeword $\bm{c}$ 
in the form of probability vectors. 
An $\argmax$ function is then applied to convert $\bm{c}$ into a discrete letter sequence,  
removing any extra information from the probability vectors.  
Next, the conventional IDS operations are performed on $\hat{\bm{c}} = \mathrm{CIDS}(\bm{c},\bm{p})$ 
according to a randomly generated error profile $\bm{p}$. 
The one-hot representation of $\hat{\bm{c}}$ is then passed into the decoder, 
which reconstructs the recovered sequence $\hat{\bm{s}}$. 
Finally, metrics are computed to measure the differences between the original source sequence $\bm{s}$ 
and the reconstructed sequence $\hat{\bm{s}}$, providing an evaluation of the method's performance.

Since the sequences are randomly generated from an enormous pool of possible terms, 
the training and testing sets are separated using different random seeds. 
For example, in the context of this work, the source sequence is a $100$-long $4$-ary sequence, 
providing $1.6\times 10^{60}$ possible sequences. 
Given this vast space, sets of randomly generated sequences using different seeds are unlikely to overlap. 

\section{Experiments on the Differentiable IDS Channel}\label{sec:experimentDIDS}
\subsection{Accuracy of the channel}\label{sec:AccDIDS}
The differentiable IDS channel is expected to faithfully modify the input sequence according to the given profiles. 
To explicitly demonstrate the performance, 
accuracy is evaluated under various profile settings. 


The results is illustrated in \cref{fig:AccDIDS}. 
It is suggested that the differentiable IDS channel edits the input sequence 
faithfully according 
to the profile when the total channel error rate is no more than 20\%. 
When the error rate exceeds 20\%, the accuracy of the differentiable IDS channel declines 
as the channel error rate increases.
It is worth noting that realistic DNA-based storage channels typically do not exhibit error rates above 20\%.

\begin{figure}[htb!]
    \vskip 0.in
    \centering
        \includegraphics[width=1.\linewidth,trim={0 0 0 0},clip]{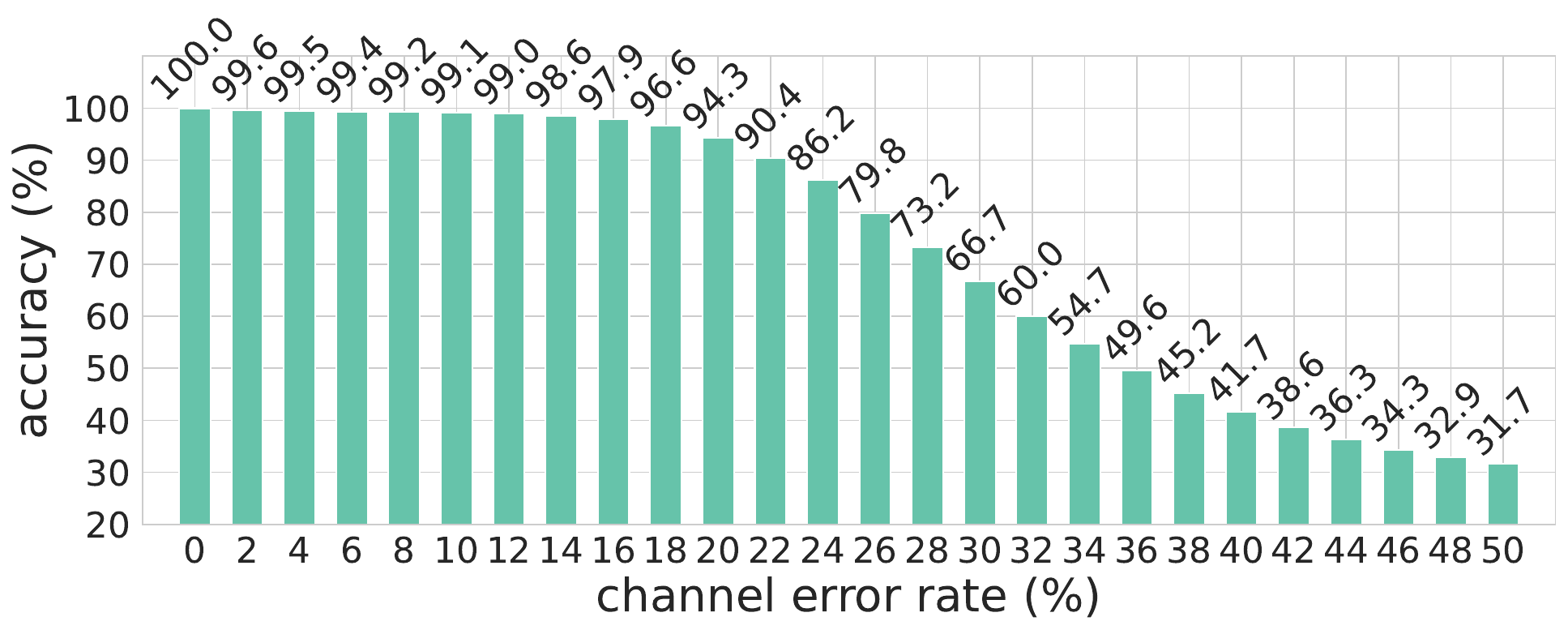}
    \caption{The accuracy of the differentiable IDS channel under various channel error rates.
    Accuracy is calculated by comparing the outputs of the differentiable IDS channel 
    with those of the conventional IDS channel.
    }
    \label{fig:AccDIDS}
    \vskip -0.in
\end{figure}
\section{Experiments on the IDS-Correcting Code}\label{sec:experiments}
Commonly used methods for synthesizing DNA molecules in DNA-based storage pipelines 
typically yield sequences of lengths ranging from $100$ to $200$~\cite{welter2024end}. 
In this study, 
we choose the number $150$ as the codeword length, aligning with these established practices. 
Unless explicitly stated otherwise, 
all the following experiments adhere to the default setting: 
source sequence length $\ell_s=100$, 
codeword length $\ell_c=150$, 
auxiliary loss weight $\mu=1$, 
and the error profile is generated with a $1\%$ probability of errors occurring at each position, 
with insertion, deletion, and substitution errors equally likely. 

To evaluate performance, the nucleobase error rate (NER) is employed as a metric, 
analogous to the bit error rate (BER), but replacing bits with nucleobases. 
For a DNA sequence $\bm{s}$ and its decoded counterpart $\hat{\bm{s}}$, the NER is defined as
\begin{equation}\label{eqn:nber}
    \mathrm{NER} (\bm{s}, \hat{\bm{s}}) = \frac{\# \{s_i\neq \hat{s}_i\}}{\# \{s_i\}}.
\end{equation}
The NER represents the proportion of nucleobase errors 
corresponding to base substitutions in the source DNA sequence. 
It's worth noting that these errors can be post-corrected using a mature conventional outer code. 

The source code is uploaded at \url{https://github.com/aalennku/THEA-Code}. 

\subsection{Performance with different channel settings}\label{subsec:performance}
The code rate is the proportion of non-redundant data in the codeword, 
calculated by dividing the source length $\ell_s$ by the codeword length $\ell_c$. 
We explored variable source lengths $\ell_s$ while keeping the codeword length $\ell_c=150$ fixed. 
The results in \cref{tab:coderate}
reveal a trend that the NER increases from $0.09\%$ to $2.81\%$ as the code rate increases from $0.33$ to $0.83$. 

            

By applying an outer conventional ECC to address the remaining NER, 
which is a common technique in DNA-based storage~\cite{press2020hedges,pfister2021polar,yan2022segmented,welzel2023dna}, 
a complete solution for DNA-based storage is achieved. 
Here, the IDS-correcting code is focused. 


By controlling the generation process of the error profile $\bm{p}$ for different channel settings, 
we can evaluate whether THEA-Code learns channels' attributes and produces customized codes 
based on the models' performance. 

\textbf{Results on IDS channels with position related errors.} 
Along with the default setting, where error rates are position-insensitive (denoted as $\mathrm{Hom}$), 
two other IDS channels parameterized by ascending ($\mathrm{Asc}$) and descending ($\mathrm{Des}$) 
error rates along the sequence are 
considered.\footnote{These settings simplify DNA-based storage channels, 
as a DNA sequence is marked with a 3' end and a 5' end. 
Some researchers believe that the error rate accumulates towards the sequence end during synthesis~\cite{meiser2020dna}.}
The $\mathrm{Asc}$ channel has error rates increasing from $0\%$ to $2\%$ along the sequence, 
with the average error rate matching that of the default setting $\mathrm{Hom}$. 
The $\mathrm{Des}$ channel follows a similar pattern but has decreasing error rates along the sequence. 



To verify that the proposed method customizes codes for different channels, 
cross-channel testing was conducted, with the results shown in \cref{tab:diffchannel}. 
The numbers in the matrix represent the NER of a model trained with the channel of the row 
and tested on the channel of the column. 


The diagonal of \cref{tab:diffchannel} shows the results of the model trained and tested with a consistent channel, 
suggesting that the learned THEA-Code exhibits varying 
performance depending on the specific channel configuration. 
The columns of \cref{tab:diffchannel} suggest that, for each testing channel, 
models trained with the channel configuration consistently achieve the best performance among the three channel settings. 
Considering the $\mathrm{Hom}$ channel is a midway setting between $\mathrm{Asc}$ and $\mathrm{Des}$, 
the first and third columns (and rows) show that the more dissimilar the training and testing channels are, 
the worse the model's performance becomes, even though the overall error rates are the same across the three channels. 
These findings verify that the deep learning-based method effectively customizes codes for specific channels, 
which could advance IDS-correcting code design into a more fine-grained area. 

\textbf{Results on IDS channels with various IDS error rates.} 
IDS channels with larger error probabilities were also tested. 
The experiments were extended to include channels with error probabilities in 
$\{0.5\%, 1\%, 2\%, 4\%, 8\%, 16\%\}$, with results listed in \cref{tab:diffchannelerror}. 

It is suggested that models trained on channels with higher error probabilities exhibit compatibility with channels 
with lower error probabilities. 
In most cases, models trained and tested on similar channels achieve better performance. 

\begin{table}[tb!]
    \vskip -0in
    \setlength{\tabcolsep}{1.1pt}
    \centering
    {
    \begin{tabular}{rcccc}
        \toprule
        $\ell_s$ & 50     & 75    & 100   & 125 \\
        code rate & 0.33   & 0.50  & 0.67 & 0.83 \\ \midrule
        NER($\%$)& 0.12 $\pm$ 0.03  & 0.51 $\pm$ 0.03 & 1.15 $\pm$ 0.08 & 3.71 $\pm$ 0.59 \\
        \bottomrule
    \end{tabular}
    }
    \caption{The testing NER for different source lengths $\ell_s$, 
    with the codeword length fixed at $\ell_c = 150$. 
    The code rate is calculated as $\ell_s / \ell_c$, ranging from 0.33 to 0.83.
    }\label{tab:coderate}
    \vskip -0.in
\end{table}
\begin{table}[tb!]
    \centering
    {
    \begin{tabular}{r|ccc}\toprule
        NER($\%$)                          & $\mathrm{Asc}$           & $\mathrm{Hom}$           & $\mathrm{Des}$           \\\midrule
        $\mathrm{Asc}$ & \textbf{0.90 $\pm$ 0.09}  & 1.46 $\pm$ 0.08 & 2.09 $\pm$ 0.44 \\
        $\mathrm{Hom}$ & 1.03 $\pm$ 0.20 & \textbf{1.15 $\pm$ 0.08}  & 1.30 $\pm$ 0.03 \\
        $\mathrm{Des}$ & 1.72 $\pm$ 0.12  & 1.32 $\pm$ 0.07 & \textbf{1.01 $\pm$ 0.05}  \\
        \bottomrule
    \end{tabular}
    }
    \caption{The testing NER across different channels. 
    Each entry is the NER of a model trained (resp. tested) with 
    the row (resp. column) header channel. 
    }\label{tab:diffchannel}
    \vskip -0.in
\end{table} 

\begin{table}[tb!]
    \vskip -0.in
    \centering
    {\small
    \begin{tabular}{r|cccccc}\toprule
        NER($\%$) & 0.5\%         & 1\%           & 2\%           & 4\%           & 8\%           & 16\%           \\ \midrule
        0.5\%         & \textbf{0.68} & 1.59          & 4.26          & 11.67         & 26.87         & 45.61          \\
        1\%           & 0.52          & \textbf{1.15} & 2.90           & 8.12          & 21.19         & 41.03          \\
        2\%           & 0.67          & 1.43          & \textbf{3.16} & 7.79          & 18.7          & 36.89          \\
        4\%           & 1.25          & 1.76          & 2.88          & \textbf{5.53} & 12.39         & 28.31          \\
        8\%           & 2.74          & 3.24          & 4.30          & 6.62          & \textbf{12.2} & 25.41          \\
        16\%          & 11.57         & 11.93         & 12.61         & 14.4          & 17.22         & \textbf{25.51} \\ \bottomrule
        \end{tabular}
    }
    \caption{The testing NER across different IDS error probabilities. 
    The row and column headers correspond to channels configured with respective 
    probabilities of errors. 
    Each entry represents the NER of a model trained (resp. tested) on the channel specified 
    by the row (resp. column) header. 
    }\label{tab:diffchannelerror}
    \vskip -0.in
\end{table}

\begin{table}[tb!]
    \vskip -0.in
    \centering
    {\small
    \begin{tabular}{rcccccc}\toprule
        code rate  & 0.33 & 0.50 & 0.6  & 0.67 & 0.75 & 0.83 \\\midrule
        Cai       & 0.44 & 1.00 & -    & 2.53 & -    & 8.65 \\
        DNA-LM    & 0.55 & 1.03 & -    & 2.29 & -    & 7.43 \\
        HEDGES    & 0.28 & \textbf{0.25} & \textbf{0.65} & -    & 3.43 & -    \\
        THEA-Code & \textbf{0.09} & 0.46 & 1.00    & \textbf{1.06} & \textbf{2.03}    & \textbf{2.81} \\\bottomrule
    \end{tabular}
    }
    \caption{The testing error rates compared with different established codes, through the default $1\%$ IDS channel. 
    }\label{tab:comparison}
    \vskip -0.in
\end{table}

\begin{table*}[tb!]
    \vskip -0.in
    \centering
    {\small
    \begin{tabular}{r|ccc|ccc|ccc}
        \toprule
            & \multicolumn{3}{c|}{$r = 0.33$}               & \multicolumn{3}{c|}{$r=0.50$}                 & \multicolumn{3}{c}{$r=0.67$}                     \\ \midrule
        NER($\%$)& C111           & C253           & MemSim            & C111           & C253           & MemSim            & C111            & C253            & MemSim             \\ \midrule
        C111 & \textbf{2.28} & 3.02          & 15.9          & \textbf{7.60} & 8.77          & 24.85         & \textbf{15.19} & 16.96          & 34.46          \\
        C253 & 2.73          & \textbf{2.93} & 17.3          & 9.15          & \textbf{9.13} & 25.09         & 16.87          & \textbf{16.90} & 32.86          \\
        MemSim  & 5.60           & 6.64          & \textbf{1.55} & 14.78         & 16.62         & \textbf{6.11} & 24.89          & 25.91          & \textbf{12.02} \\ \bottomrule
        \end{tabular}
    }
    \caption{The testing NER across different channels including C111, C253, and MemSim, under varying code rates. 
    Each entry represents the NER of a model trained (resp. tested) on the channel specified 
    by the row (resp. column) header. 
    }\label{tab:memsim}
    \vskip -0.in
\end{table*}

\begin{table*}[tb!]
    \vskip -0.in
    \centering
    {\small
    \begin{tabular}{r|ccc|ccc|ccc}\toprule
        & \multicolumn{3}{c|}{$r=0.33$}                  & \multicolumn{3}{c|}{$r=0.50$}                  & \multicolumn{3}{c}{$r=0.67$}                     \\\midrule
        & C111          & C253          & MemSim        & C111          & C253          & MemSim        & C111           & C253           & MemSim         \\\midrule
        Cai       & 17.01         & 17.52         & 72.74         & 29.00         & 29.57         & 74.40         & 40.12          & 42.62          & 73.90          \\
        DNA-LM    & 32.24         & 37.33         & 60.13         & 45.32         & 51.13         & 64.27         & 56.34          & 60.22          & 68.72          \\
        HEDGES    & 3.21          & 4.56          & 29.42         & 27.22         & 27.79         & 99.56         & 54.35          & 55.66          & 99.62          \\
        THEA-Code & \textbf{2.28} & \textbf{2.93} & \textbf{1.55} & \textbf{7.60} & \textbf{9.13} & \textbf{6.11} & \textbf{15.19}          & \textbf{16.90}          & \textbf{12.02} \\ \bottomrule
    \end{tabular}
    }
    \caption{The testing error rates compared with established code through channels 
    including C111, C253, and MemSim, under varying code rates. 
    }\label{tab:realcom}
    \vskip -0.in
\end{table*}

\textbf{Results on realistic IDS channels.} 
We also conducted experiments using IDS channels that more closely resemble 
realistic IDS channels in DNA-based storage. 
A memory channel was proposed in \cite{hamoum2021channel}, 
relying on statistical data obtained via a realistic storage pipeline. 
It models the IDS errors based on the $k$-mers of sequences and adjacent edits. 
In this work, we utilize the publicly released trained memory channel from \cite{hamoum2021channel}, 
filtering out apparent outlier sequences with Levenshtein distance greater than $20$. 
This simulated channel is referred to as MemSim. 

In practice, a DNA sequence $\bm{c}$ is input into MemSim to produce  
the output sequence $\hat{\bm{c}}$ from the channel. 
By comparing $\bm{c}$ and $\hat{\bm{c}}$, an error profile $\bm{p}$ is inferred. 
Using the sequence $\bm{c}$ and the error profile $\bm{p}$ 
in the procedure depicted in \cref{fig:thea-code}, 
an IDS-correcting code for MemSim is customized. 

For comparison, two simple channels, partially aligned with MemSim, were also considered. 
The overall IDS error rate for MemSim is $10.36\%$, 
with the proportions of insertion, deletion, and substitution being 
$1.66\%$, $5.31\%$, and $3.38\%$, respectively. 
We refer to the context-free channel with these specific error proportions as channel C253. 
Channel C111 is defined as having the same overall IDS error rate $10.36\%$, 
but with equal proportions of insertion, deletion, and substitution.
It is evident that MemSim is the closest approximation to a realistic channel, 
followed by C253, while C111 deviates the most from a realistic channel, 
despite all having the same overall IDS error rate. 

The results across channels, including C111, C253, and MemSim, are presented in \cref{tab:memsim}. 
The results suggest that THEA-Code performs better when the model is trained on the same channel used for testing. 
Specifically, for the realistic channel, codes trained on the simpler channels C253 and C111 fail to deliver 
satisfactory results. 
Overall, THEA-Code trained and tested with MemSim achieves the best results, 
demonstrating that the proposed model significantly benefits
from customizing the code for the realistic channel.

\subsection{Comparison experiments}\label{sec:comparison}

Comparison experiments were conducted against prior works include: the combinatorial code from~\cite{cai2021correcting}, 
the segmented code method DNA-LM from~\cite{yan2022segmented}, 
and the efficient heuristic method HEDGES from~\cite{press2020hedges}.

Such methods are typically designed to operate under discrete, fixed configurations,  
making it challenging to align them within the same setting. 
We made every effort to align these methods, and 
present a subset of the comparison results in \cref{tab:comparison}, 
which is tested through the default $1\%$ error channel. 
Detailed configurations and results across multiple channels are provided in 
 \cref{app:comparison}. 

\cref{tab:comparison} demonstrates the effectiveness of the proposed method. 
The performance of THEA-Code and HEDGES outperform the other methods by a large margin. 
At lower code rates, THEA-Code achieves a comparable error rate to HEDGES. 
At higher code rates, the proposed method outperforms HEDGES, 
achieving much lower error rates.

\textbf{Comparison through the realistic channel.}
We also compared these codes across the channels C111, C253, and MemSim introduced in \cref{subsec:performance}, 
all of which have an overall channel error rate of $10.36\%$. 
Specifically, MemSim simulates the IDS channel from a realistic storage pipeline. 

The results are illustrated in~\cref{tab:realcom}. 
It can be observed that high-error-rate channels severely degrade the performance of compared codes, 
while the proposed THEA-Code outperforms them by a significant margin. 
Moreover, the compared codes, lacking the ability to adapt to specific channels, 
show a noticeable decline in performance as the channel transitions from the simpler C111/C253 to the more realistic MemSim. 
In contrast, THEA-Code leverages customized channel-specific designs, 
achieving the best performance on MemSim across all three channels.

\section{More Experiments in the Appendices}
In this work, the disturbance-based discretization, the differentiable IDS channel, 
and the auxiliary reconstruction loss are newly proposed. 
Comprehensive experiments on these modules are presented in the Appendices provided in 
\url{https://arxiv.org/abs/2407.18929}. 
Following is a brief overview. 

The full proof of Proposition \ref{thm:gs} is given in 
 \cref{app:proof}. 
Details of the comparison experiments discussed in \cref{sec:comparison} are provided in 
\cref{app:comparison}. 

\textbf{Disturbance-based discretization.} 
The ablation studies and hyperparameter optimization for the disturbance-based discretization, 
including the discretization effect compared to vanilla softmax, 
the optimization of temperature $\tau$, 
and the results of potential alternative, 
are in 
 \cref{app:gs}. 

\textbf{Differentiable IDS channel.} 
Experiments on the differentiable IDS channel, 
including 
the gradient trace under specific error profiles, 
and the gradients with respect to the error profile through an identity channel, 
are in 
 \cref{app:experimentDIDS}. 

\textbf{Auxiliary reconstruction loss.} 
For the auxiliary reconstruction loss, 
ablation studies, weight optimization of the loss term $\mu$, 
and experiments on different auxiliary patterns 
are provided in 
 \cref{app:ablationaux}. 

\textbf{Dataset and model.} 
The construction of datasets and the definition of error profiles are 
detailed in 
 \cref{app:dataset}. 
A brief introduction to the Transformer model, 
as well as complexity analysis and time consumption, is presented in 
 \cref{app:complexity}. 

\section*{Acknowledgements}
This work was supported by the National Key Research and Development Program of China under Grant 
2020YFA0712100 and 2025YFC3409900, the National Natural Science Foundation of China, 
and the Emerging Frontiers Cultivation Program of Tianjin University Interdisciplinary Center.
\bibliography{AlanGuo-bio}


\clearpage
\appendix
\setcounter{secnumdepth}{2}

\section{Detailed Proof of Proposition~\ref{thm:gs}}\label{app:proof}
Let $\bm{x}=(x_1,x_2,\ldots,x_n)$ denote the logits output by the upstream model, 
and let $\bm{y}=(y_1,y_2,\ldots,y_n) = \mathrm{GS}(\bm{x})$ represent the Gumbel-Softmax of $\bm{x}$
\begin{equation}\label{eqn:yy}
    y_i = \frac{\exp{((x_i+g_i)/\tau)}}{\sum_{j=1}^{n}\exp{((x_j+g_j)/\tau)}}, \quad i=1,2,\ldots,n,
\end{equation}
where the $g_i$ are i.i.d. samples drawn from the Gumbel distribution $G(0,1)$, and 
$\tau$ is the temperature parameter controlling the entropy. 
Let $\mathcal{L}=f(y_1,y_2,\ldots,y_n)$ be the optimization target, 
which is the composite function of the downstream model and the loss function.

Without loss of generality, consider the partial derivative of $\mathcal{L}$ with respect to $x_1$, which is
\begin{align}
    \frac{\partial \mathcal{L}}{\partial x_1} =& \sum_{i=1}^{n} \frac{\partial f}{\partial y_i} \frac{\partial y_i}{\partial x_1}\\
    =& \frac{1}{\tau}\frac{\partial f}{\partial y_1}
    \frac{\exp{((x_1+g_1)/\tau)}\sum_{j\neq 1}\exp{((x_j+g_j)/\tau)}}{(\sum_{i}\exp{((x_i+g_i)/\tau)})^2}\\
    &- \frac{1}{\tau}\sum_{j\neq 1}\frac{\partial f}{\partial y_j} \frac{\exp{((x_1+g_1)/\tau)}\exp{((x_j+g_j)/\tau)}}{(\sum_{i}\exp{((x_i+g_i)/\tau)})^2}\\
    =& \frac{1}{\tau}y_1 \left(\sum_{j\neq 1}y_j\left(\frac{\partial f}{\partial y_1}-\frac{\partial f}{\partial y_j}\right)\right).\label{eqn:chain}
\end{align}
At convergence, the model reaches a (local) minimum where gradients vanish. 
Therefore, according to \cref{eqn:chain}, 
either $y_1$ 
or the term $\left(\sum_{j\neq 1}y_j\left(\frac{\partial f}{\partial y_1}-\frac{\partial f}{\partial y_j}\right)\right)$ 
should be zero. 

(1). Consider the case that 
$\left(\sum_{j\neq 1}y_j\left(\frac{\partial f}{\partial y_1}-\frac{\partial f}{\partial y_j}\right)\right)=0$. 
Since $y_1 = 1-\sum_{j\neq 1} y_j$, we can express $\mathcal{L}$ as a function of $y_2,\ldots,y_n$: 
\begin{equation}
    \mathcal{L}(y_2,\ldots,y_n)=f(1-\sum_{j\neq 1} y_j,y_2,\ldots,y_n).
\end{equation}
Then the dot product of the gradient of $\mathcal{L}$ with the vector $(y_2,\ldots,y_n)'$ becomes
\begin{align}\label{eqn:dotproduct}
    &\nabla \mathcal{L} \cdot (y_2,\ldots,y_n)' \\
    =& -\left(\frac{\partial f}{\partial y_1}-\frac{\partial f}{\partial y_2},\ldots,\frac{\partial f}{\partial y_1}-\frac{\partial f}{\partial y_n}\right) \cdot (y_2,\ldots,y_n)'\\
    =& - \left(\sum_{j\neq 1}y_j\left(\frac{\partial f}{\partial y_1}-\frac{\partial f}{\partial y_j}\right)\right) = 0.\label{eqn:graddotvar}
\end{align}
Now consider two subcases: 
\begin{itemize}
    \item If $(y_2, \ldots, y_n) = \bm{0}$, then $y_1 = 1$, 
    meaning that the output is exactly one-hot even under the randomness of Gumbel noise. 
    \item If $(y_2, \ldots, y_n) \neq \bm{0}$ and $y_1 \neq 1$, 
    then due to the randomness introduced by the Gumbel noise, 
    the condition in \cref{eqn:dotproduct} holds for varying $y_i$ values.
    Let $\phi(t) = \mathcal{L}(ty_2,\ldots,ty_n)$, 
    by the mean value theorem for multivariable functions, we have
    \begin{align}
        &\phi(t) - \phi(0) = \phi'(\tau) \cdot t \\
        =& t \nabla \mathcal{L}(\tau y_2,\ldots,\tau y_n)\cdot(y_2,\ldots,y_n)' \\
        =& \frac{t}{\tau} \nabla\mathcal{L}(\tau y_2,\ldots,\tau y_n)\cdot(\tau y_2,\ldots,\tau y_n)', \quad \tau\in[0,t]. 
    \end{align}
    According to \cref{eqn:dotproduct}, the dot product is zero, 
    so $\phi(t)-\phi(0)=0$, and thus $\phi(t)$ is constant with regard to $t$. 
    Since this reasoning holds for all such directions under the Gumbel-softmax distribution, 
    it follows that $\mathcal{L}$ is constant, 
    which means the downstream network and loss function degenerates into 
    a trivial model, contradicted to the Proposition hypothesis of non-trivial convergence. 
\end{itemize}

(2). Consider the case that $y_1 = 0$. 
Since the previous derivation was made without loss of generality for $x_1$ and $y_1$, 
the same reasoning applies to any $x_i$ and $y_i$.
Therefore, it follows that there should be some index $i_0$ such that $y_{i_0} = 1$ 
and $y_j = 0$ for all $j\neq i_0$. 

Now it can be inferred that one of $y_1,\ldots,y_n$ is close to $1$, while the others are close to $0$. 
Without loss of generality, assume $y_i < \epsilon_1$ for all $i\neq 1$, it can be calculated that
\begin{equation}
\frac{1}{y_i} = 1 + \sum_{j\neq i} \exp{((x_j-x_i + g_j - g_i)/\tau)} > \frac{1}{\epsilon_1}. 
\end{equation}
From this inequality, there must exist some $j\neq i$ such that 
\begin{equation}
    \exp{((x_j-x_i + g_j - g_i)/\tau)}> \frac{1}{n\epsilon_1}. 
\end{equation}
Let $M_1 = (\frac{1}{n\epsilon_1})^\tau$, then we have
\begin{equation}
    g_j-g_i > \ln\frac{M_1}{\exp{(x_j-x_i)}}.
\end{equation}
Note that the $g_i$s are i.i.d. samples from the Gumbel distribution $G(0,1)$, 
so $g_i-g_j$ follows a $\mathrm{Logistic}(0,1)$ distribution with CDF 
\begin{equation}
    F_{\mathrm{Logistic}(0,1)} (x) = \frac{1}{1+\exp{(-x)}}.
\end{equation}
The probability that $g_j-g_i>\ln{\frac{M_1}{\exp{(x_j-x_i)}}}$ is 
\begin{align}
    &P\left(g_j-g_i>\ln{\frac{M_1}{\exp{(x_j-x_i)}}}\right) \\
    =& 1 - \frac{M_1}{M_1+\exp{(x_j-x_i)}}. 
\end{align}
Requiring this probability to be greater than $1 - \epsilon_2$, we obtain 
\begin{equation}
    \frac{M_1}{M_1+\exp{(x_j-x_i)}} < \epsilon_2, 
\end{equation}
which is equivalent to 
\begin{equation}
    \exp{(x_j-x_i)} > M_1\left(\frac{1}{\epsilon_2}-1\right), \quad i\neq j. 
\end{equation}
Under this condition, applying the $\mathrm{softmax}$ function to $\bm{x}$ yields  
\begin{align}
    \pi_i &= \frac{\exp{x_i}}{\sum_{k}\exp{x_k}} < \frac{\exp{x_i}}{\exp{x_j}}\\
    &<\frac{1}{M_1(1/\epsilon_2-1)} = \frac{(n\epsilon_1)^{\tau}\epsilon_2}{1-\epsilon_2}\\
    &<(n\epsilon_1)^\tau \epsilon_2, 
\end{align}
which indicates that $\pi_i$ is minimum. 
And the autoencoder, under the Gumbel-Softmax disturbance, produces more confident and hence more discretized logits $\bm{x}$ 
and the corresponding categorical distribution $\bm{\pi}$. 

\section{Comparison Experiments}\label{app:comparison}
To evaluate the effectiveness of the proposed methods, we conducted comparison experiments against three prior works, 
which are: 
\begin{itemize}
    \item a combinatorial code that can correct single IDS errors over a 4-ary alphabet from Cai~\cite{cai2021correcting}; 
    \item a segmented method for correcting multiple IDS errors, called DNA-LM from~\cite{yan2022segmented};
    \item a well-known, efficient heuristic method, called HEDGES, from a DNA-based storage research~\cite{press2020hedges}. 
\end{itemize}

These methods typically offer only a few discrete, fixed configurations. 
We made efforts to align their settings as closely as possible. 
For Cai's combinatorial code, the code rates are fixed based on the code lengths.
In our experiments on Cai, only the code rates are matched, with the code length determined according to the code rate.\footnote{It 
is important to note that code length plays a critical role in these experiments, 
as longer codewords are more likely to encounter multiple errors that cannot be corrected. 
Thus, Cai's performance here is just a baseline statistic of multi-errors with respect to the length, 
and performance may degrade with increased length. }
For DNA-LM, we maintained the codeword length around 150, adjusting the number of segments to match code rates. 
For HEDGES, only binary library is publicly available, 
and it supports fixed code rates in $\{0.75, 0.6, 0.5, 1/3, 0.25, 1/6\}$. 
HEDGES' inner code was tested independently for comparison. 
We list all the source lengths $\ell_s$, codeword lengths $\ell_c$, 
and code rate $r$ used in the experiments in \cref{tab:setting}. 

\begin{table*}[htb!]
    \vskip 0.in
    \centering
    {
    \begin{tabular}{rllll}\toprule
                  & $r_1=\ell_{s1}/\ell_{c1}$ & $r_2=\ell_{s2}/\ell_{c2}$ & $r_3=\ell_{s3}/\ell_{c3}$  & $r_4=\ell_{s4}/\ell_{c4}$  \\\midrule
        Cai       & 0.33=7/21   & 0.50=16/32  & 0.67=32/48   & 0.83=85/102  \\
        DNA-LM    & 0.34=50/148 & 0.51=77/152 & 0.68=96/142  & 0.84=124/148 \\
        HEDGES    & 0.34=52/155 & 0.50=76/152 & \textsl{0.60=92/153}  & \textsl{0.75=115/155} \\
        THEA-Code & 0.33=50/150 & 0.50=75/150 & 0.67=100/150 & 0.83=125/150 \\\bottomrule
    \end{tabular}
    }
    \vskip 0.in
    \caption{The testing configurations for the comparison experiments. 
    Each cell includes the code rate, message length, and code length. 
    The settings are tried to be aligned, except the Cai configuration has a code length that does not align with 150,
    and 
    HEDGES uses fixed code rates of $0.60$ and $0.75$, which are not aligned.
    }\label{tab:setting}
\end{table*}

The experiments were conducted on the default IDS channel with $1\%$ error probability, 
as well as its variations, $\mathrm{Asc}$ and $\mathrm{Des}$, introduced in \cref{subsec:performance}. 
The results is illustrated in \cref{fig:comparison}. 
The experiments handled failed corrections by directly using the corrupted codeword as the decoded message.

\begin{figure}[htb!]
    \vskip -0.in
    \centering
        \includegraphics[width=0.87\linewidth,trim={0 0 0 0},clip]{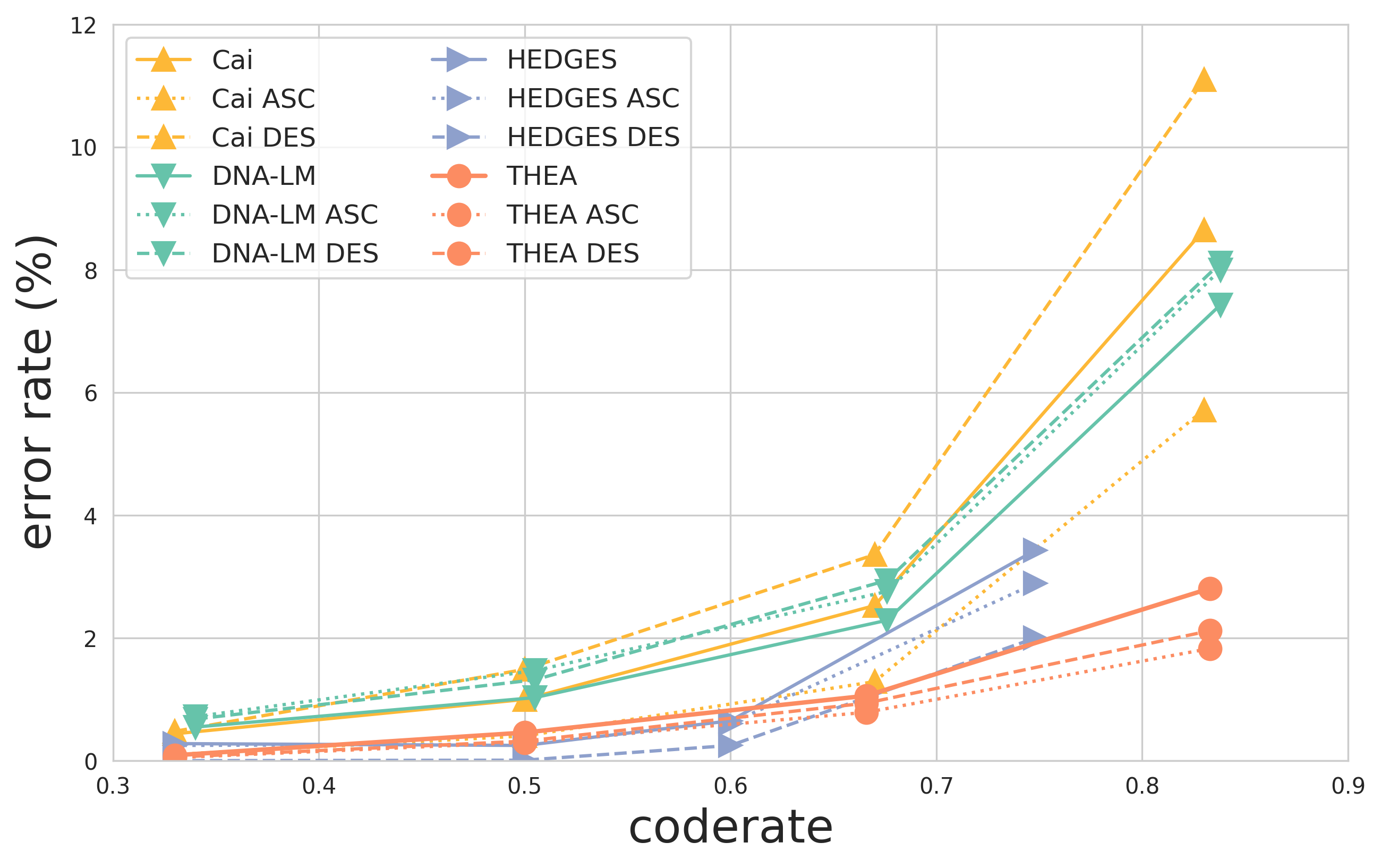}
    \caption{The error rates of the comparison experiments. 
    Results for Cai, DNA-LM, HEDGES, and THEA-Code are shown across $\mathrm{Hom}$, $\mathrm{Asc}$, and $\mathrm{Des}$ channels, 
    with respect to their code rates. 
    }
    \label{fig:comparison}
    \vskip -0.in
\end{figure}

The results for Cai’s method indicate that directly applying classical combinatorial codes 
to a $1\%$ IDS error probability channel with a codeword length of $150$ is impractical. 
The observed error rates are high, even though these values were obtained with shorter code lengths than $150$. 
The segmented method with sync markers in DNA-LM supports a codeword length of $150$ 
and can correct multiple errors across different segments. 
However, it also exhibits a high error rate, indicating a nonnegligible likelihood 
of multi-errors occurring within the same segment. 
For HEDGES, while the results are commendable, the code rate is restricted to a limited set of fixed values. 
The results of THEA-Code demonstrate the effectiveness of the proposed method. 
At lower code rates, THEA-Code achieves a comparable error rate to HEDGES. 
At higher code rates, the proposed method outperforms HEDGES, 
achieving a lower error rate at a higher code rate, specifically $2.81\%$ error rate at $0.83$ code rate for THEA-Code 
v.s. $3.43\%$ error rate at $0.75$ code rate for HEDGES.

\section{Ablation Study on the Disturbance-Based Discretization}\label{app:gs}

\subsection{Effects of the disturbance-based discretization}\label{subsec:expgs}
The ablation study on utilizing the disturbance-based discretization was conducted 
analyzing the discreteness of the codewords.
During training, the entropy \cite{shannon1948mathematical} of the codewords 
\begin{equation}\label{eqn:entropy}
H(\bm{\pi}) = -\sum_{i=1}^{k} \pi_i\log{\pi_i}
\end{equation}
was recorded. 
This entropy measures the level of discreteness in the codewords. 
Lower entropy implies a distribution that is closer to a one-hot style probability vector, 
which indicates greater discreteness. 
In addition to entropy, two other metrics were also recorded, 
as they are the reconstruction loss $\mathcal{L}_{\mathrm{CE}}$ and the NER. 
The results, plotted in \cref{fig:gumbelablation}, 
compare the default disturbance setting (Gumbel-Softmax) against a vanilla softmax approach.

\begin{figure*}[htb!]
\vskip 0.in
\centering
\subcaptionbox{Gumbel-Softmax}[1\linewidth]
{
    \includegraphics[width=0.3\linewidth]{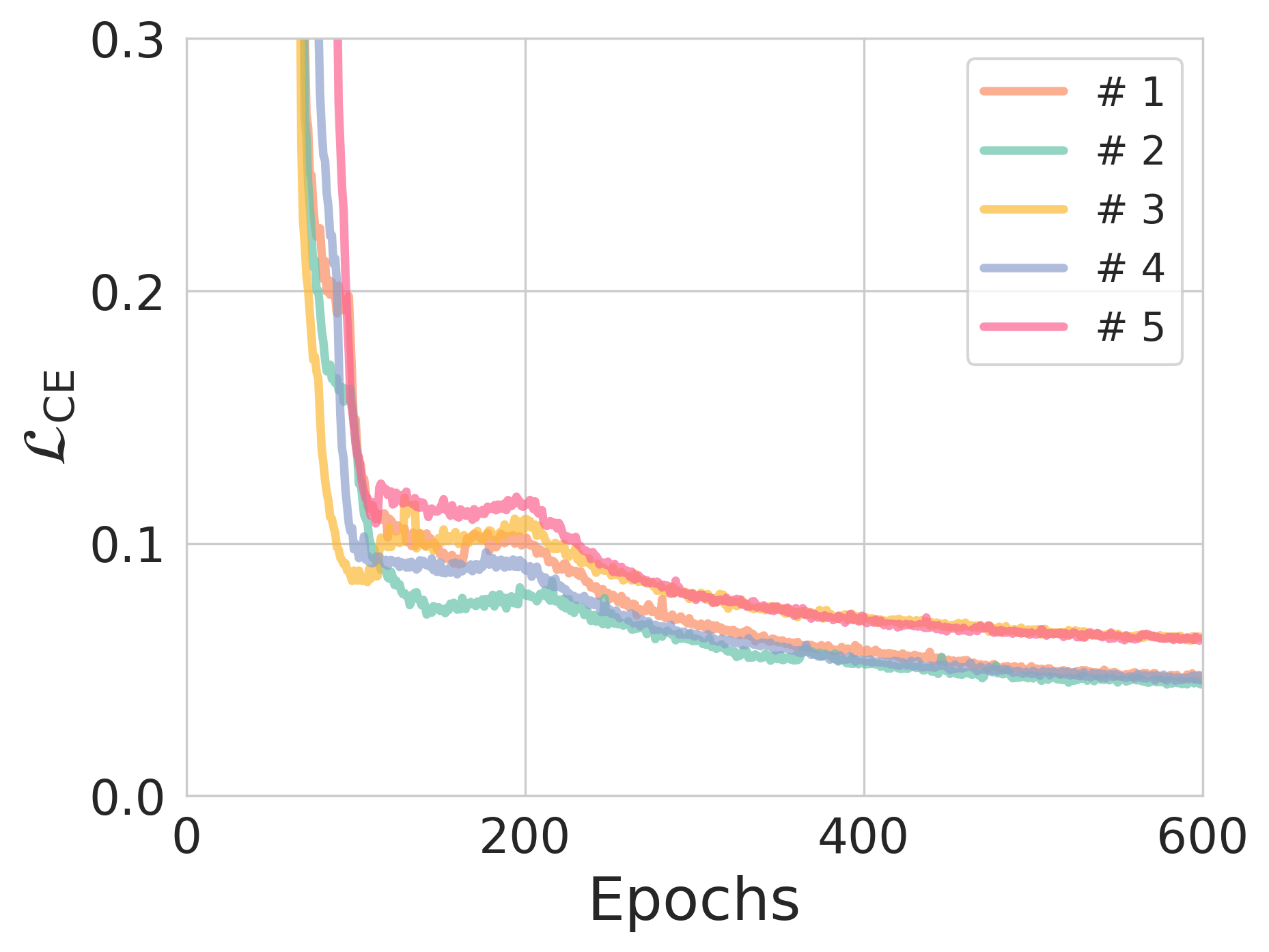}
    \includegraphics[width=0.3\linewidth]{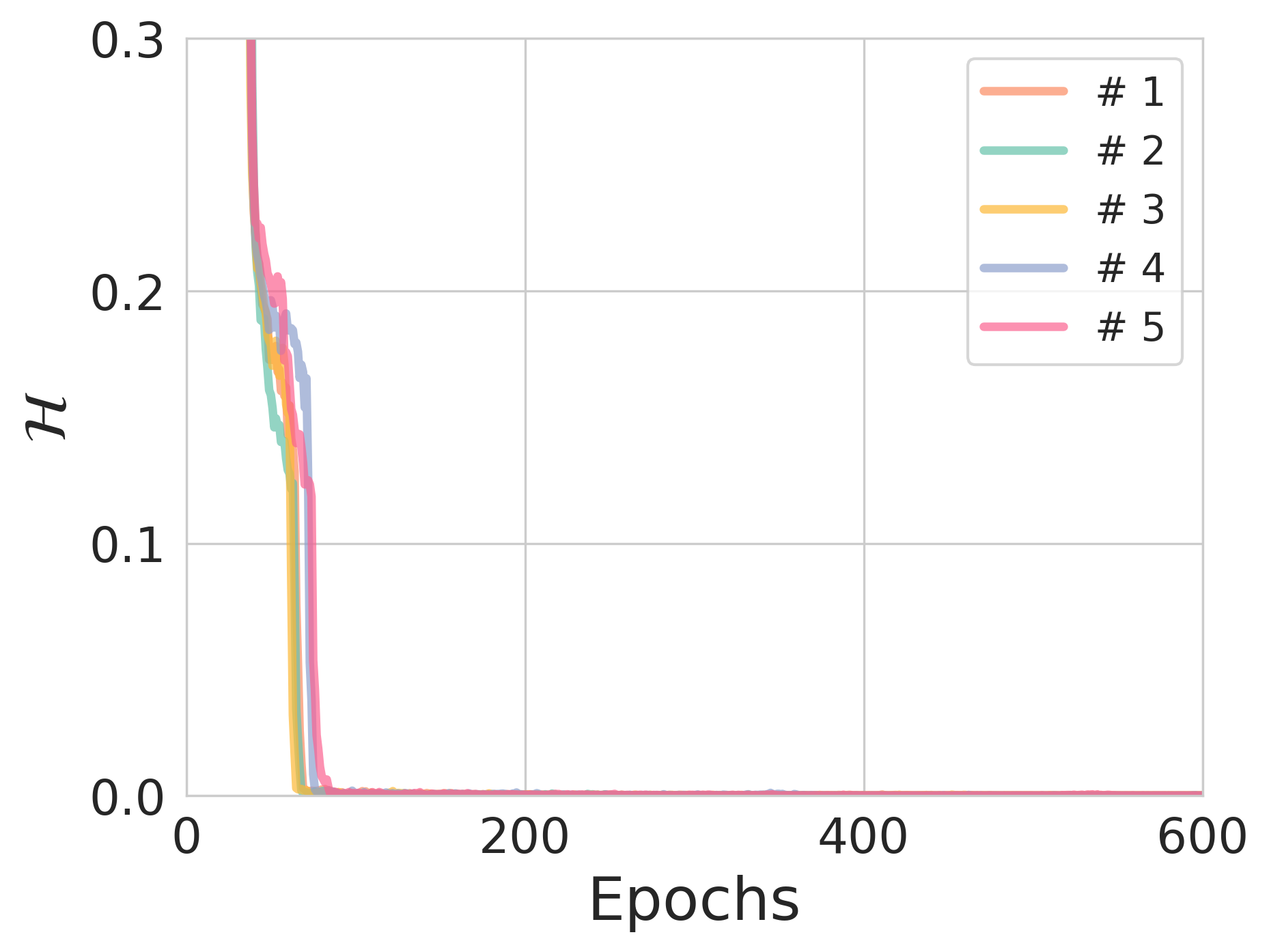}
    \includegraphics[width=0.3\linewidth]{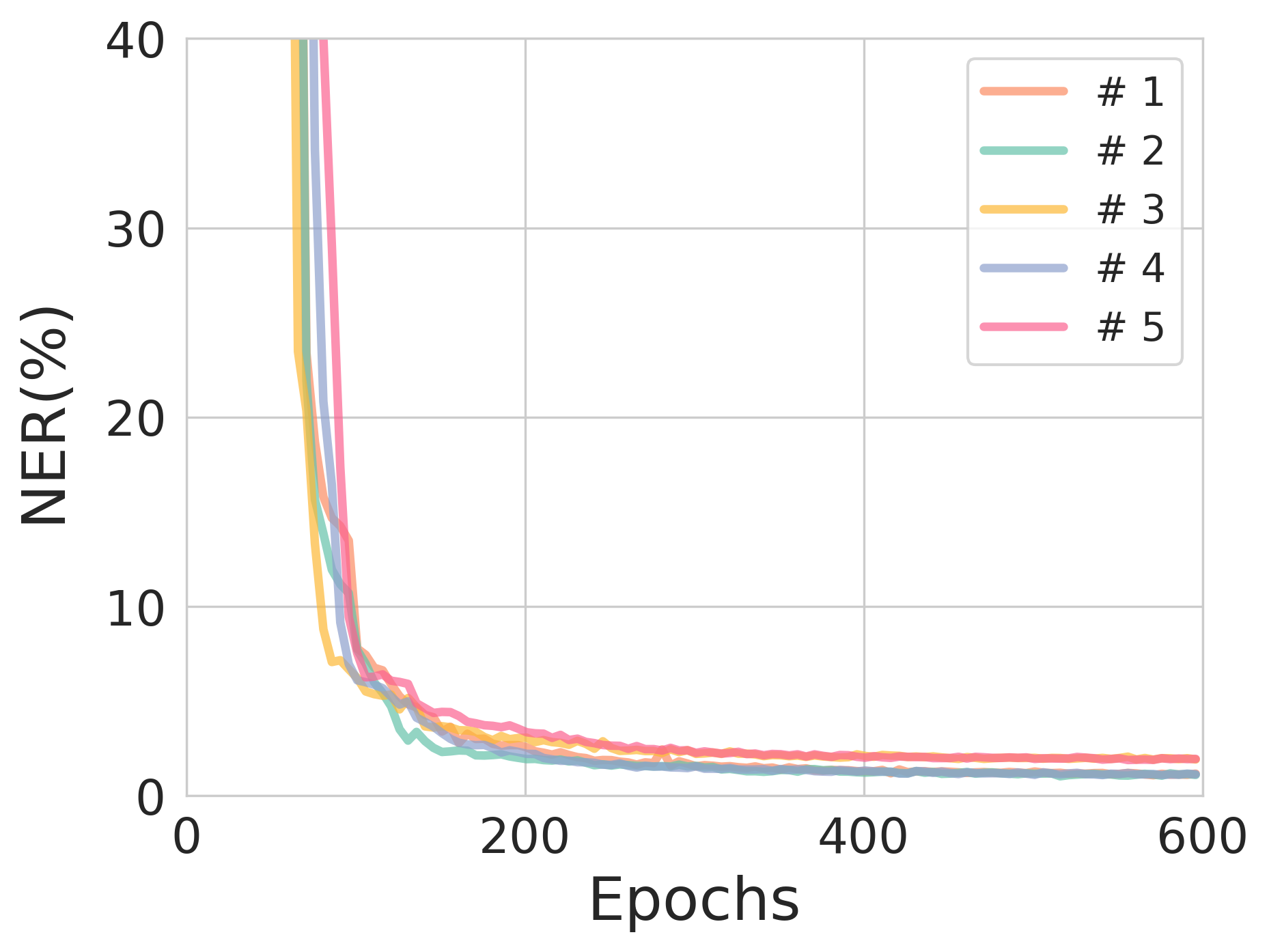}
}%

\subcaptionbox{vanilla softmax}[1\linewidth]
{
    \includegraphics[width=0.3\linewidth]{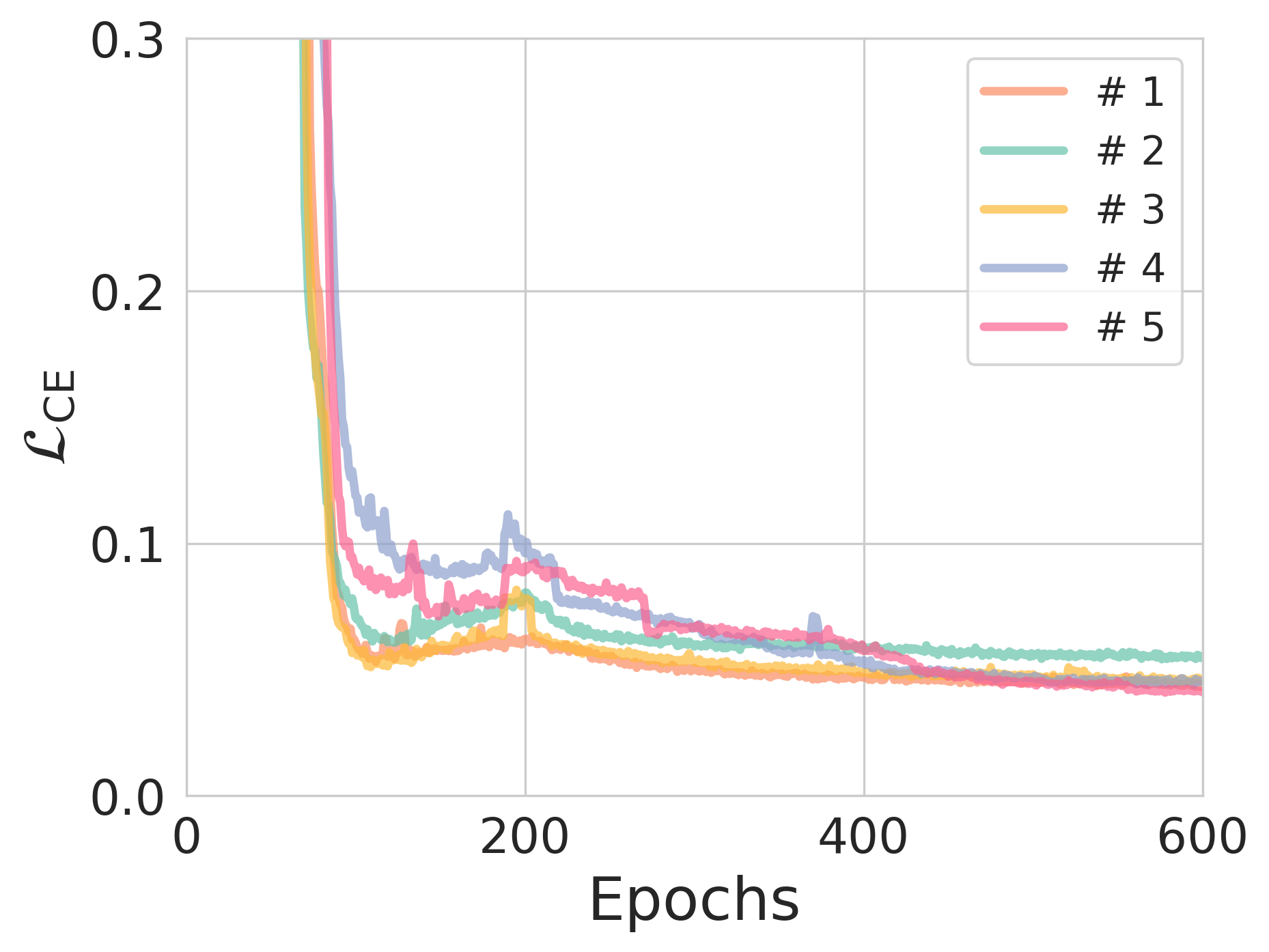}
    \includegraphics[width=0.3\linewidth]{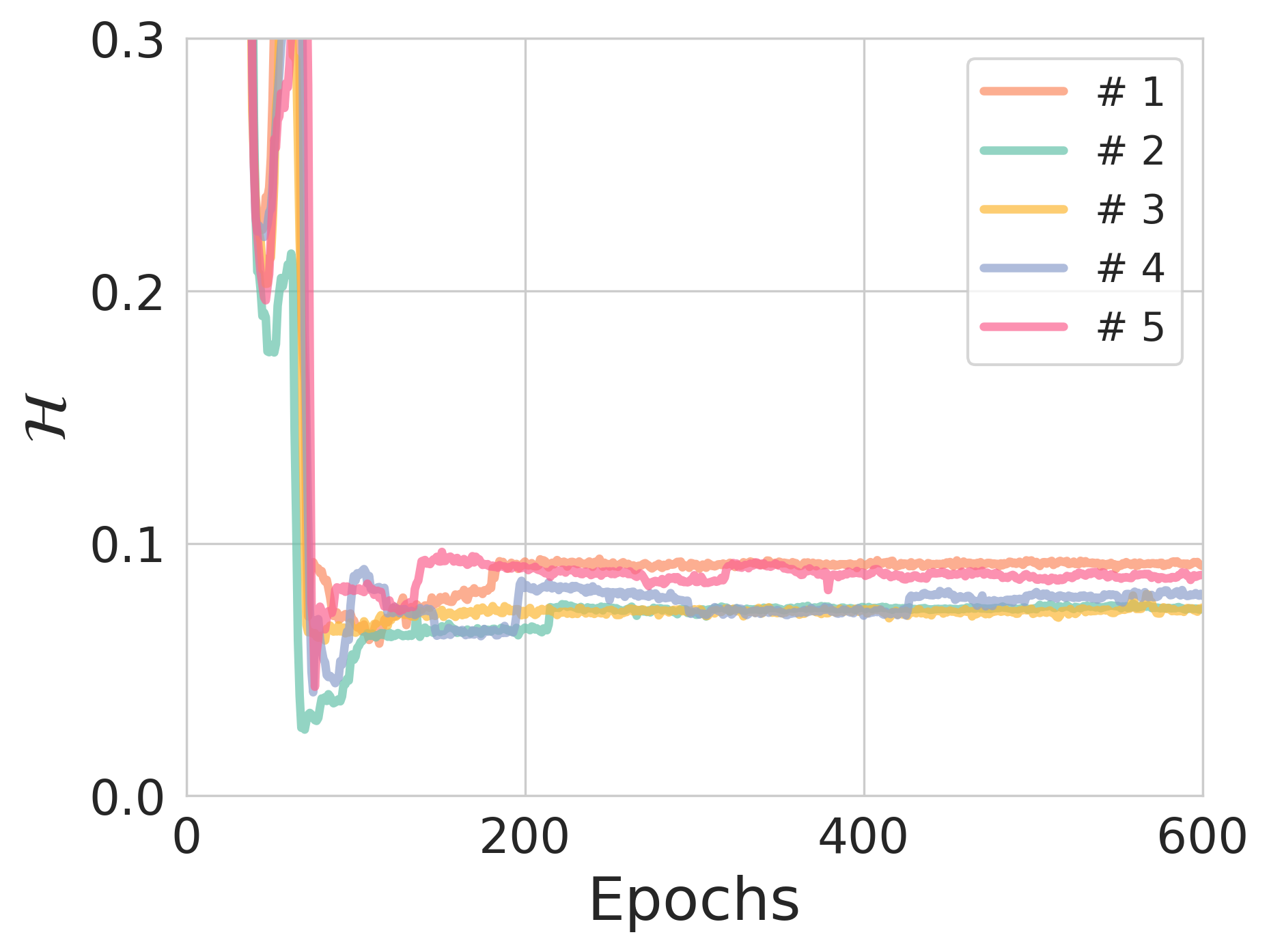}
    \includegraphics[width=0.3\linewidth]{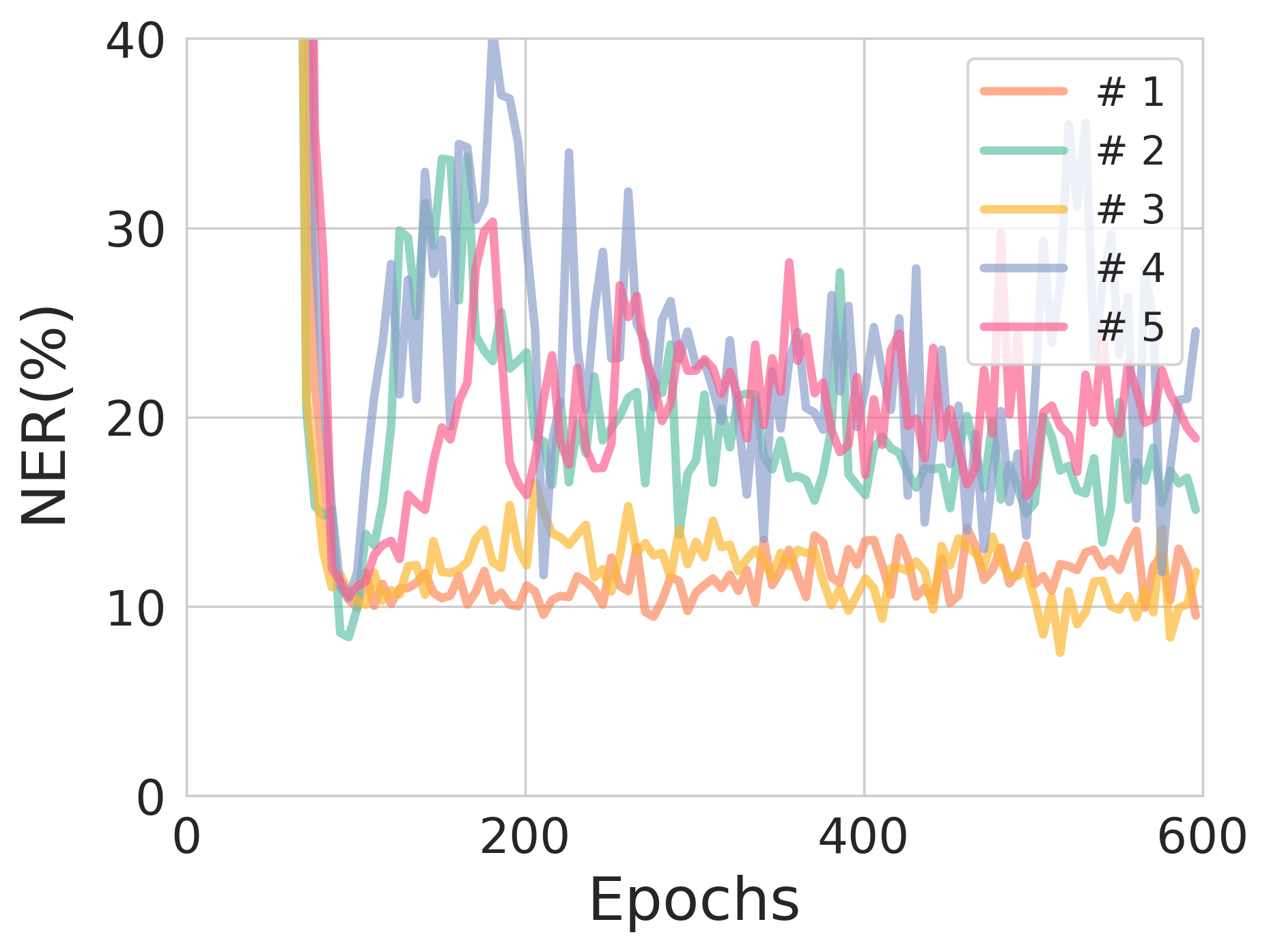}
}%
\caption{The reconstruction loss, codeword entropy, and validation NER 
comparing the Gumbel-Softmax setting against a vanilla softmax approach. 
5 runs were recorded. 
}
\label{fig:gumbelablation}
\vskip -0in
\end{figure*}

The first column of \cref{fig:gumbelablation} indicates that using disturbance marginally increases the reconstruction loss 
$\mathcal{L}_{\mathrm{CE}}$ in the continuous mode, 
which is expected since Gumbel-Softmax introduces additional noise into the system. 
When comparing the average entropy $\mathcal{H}$ of the learned codeword, 
applying disturbance-based discretization significantly reduces the entropy, 
suggesting that the codewords behave more like one-hot vectors. 
The NER is calculated in the discrete mode by replacing the softmax 
with an $\argmax$ operation on the codewords. 
The third column clearly shows that when codewords are closer to a one-hot style, 
the model is more consistent between the continuous and discrete modes, 
leading to better performance during the testing phase. 


\subsection{Optimization of hyperparameter temperature $\tau$ in the Gumbel-Softmax formula}\label{app:tau}

To examine the impact of different temperature values $\tau$ in \cref{eqn:gsequation}, 
experiments were conducted with various settings of $\tau \in \{0.25,0.5,1,2,4,8\}$. 
Since the disturbance-based discretization is designed to encourage greater discretization of the codeword, 
the codeword entropy $\mathcal{H}$, as defined in \cref{eqn:entropy}, and the validation NER 
were tracked throughout the training phase

As shown in \cref{fig:gumbelt}, lower temperature ($\tau=0.25$) has an effect in discretization, 
but result in unstable and poor model performance, while higher temperatures ($\tau\in\{2,4,8\}$) 
lead to both poor discretization and high NER. 

\begin{figure*}[htb!]
    \vskip 0.in
    \centering
    \subcaptionbox{$\tau=0.25$}[0.3\linewidth]
    {
        \includegraphics[width=1\linewidth]{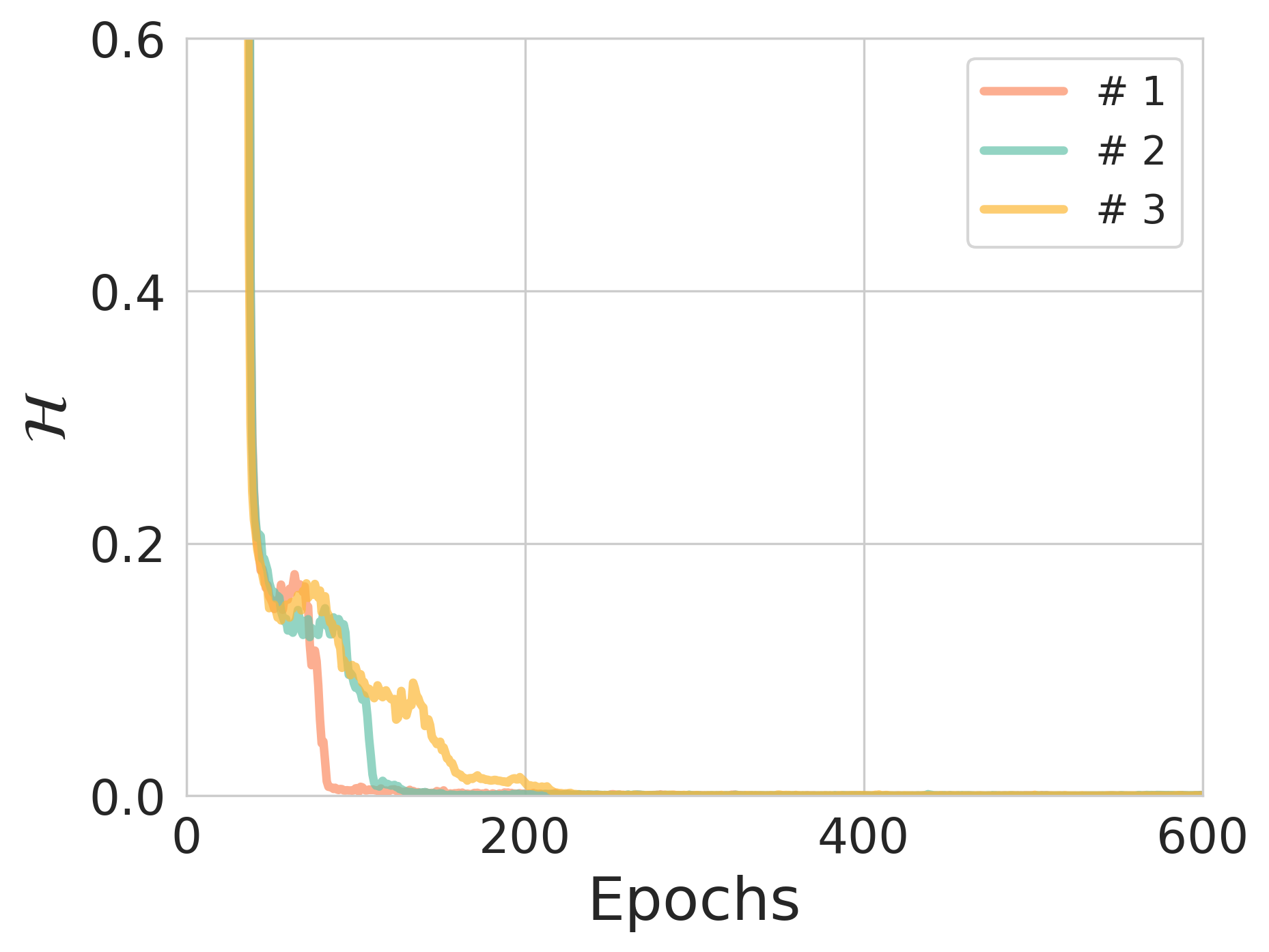}
        \includegraphics[width=1\linewidth]{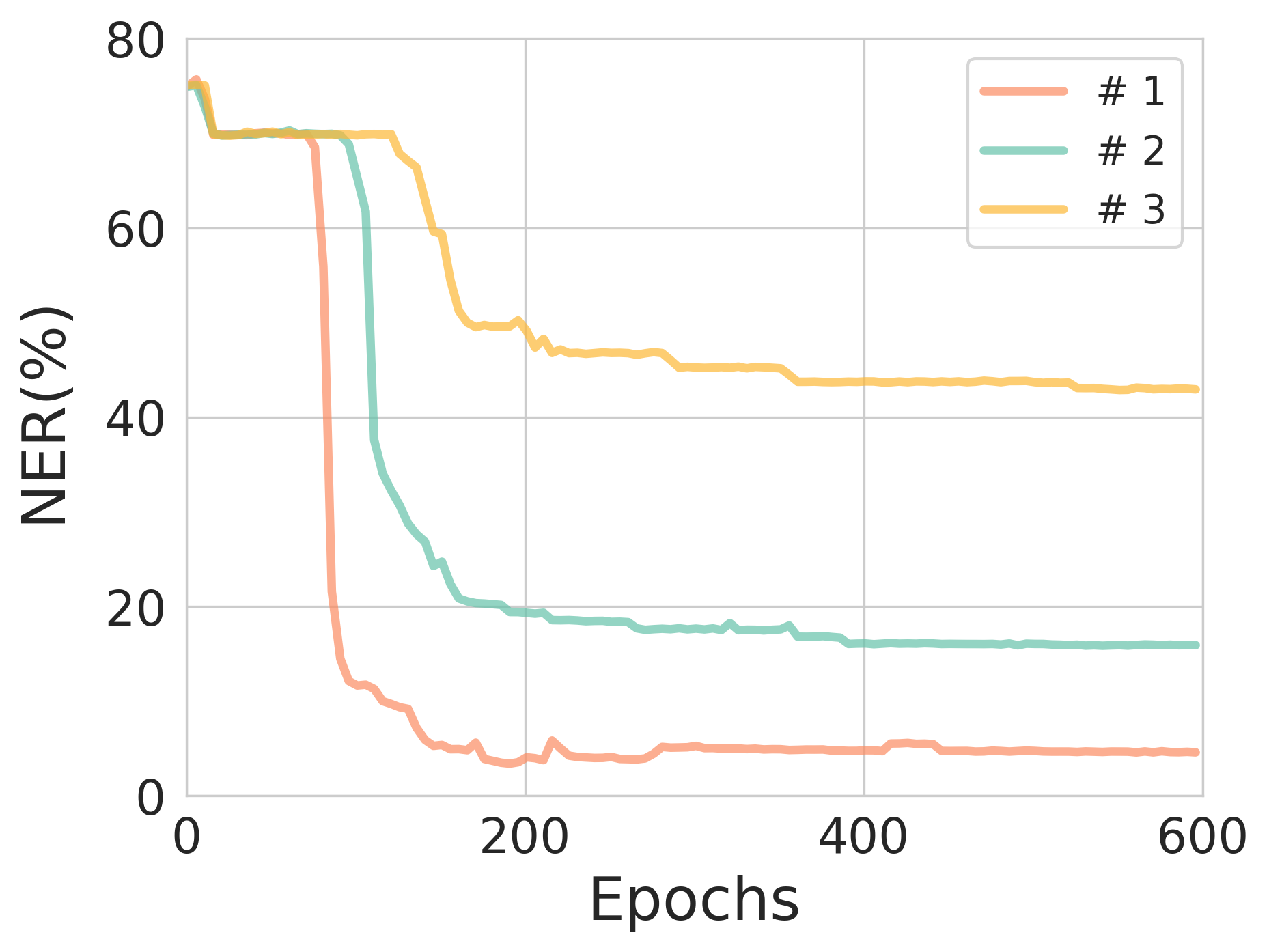}
    }%
    \subcaptionbox{$\tau=0.5$}[0.3\linewidth]
    {
        \includegraphics[width=1\linewidth]{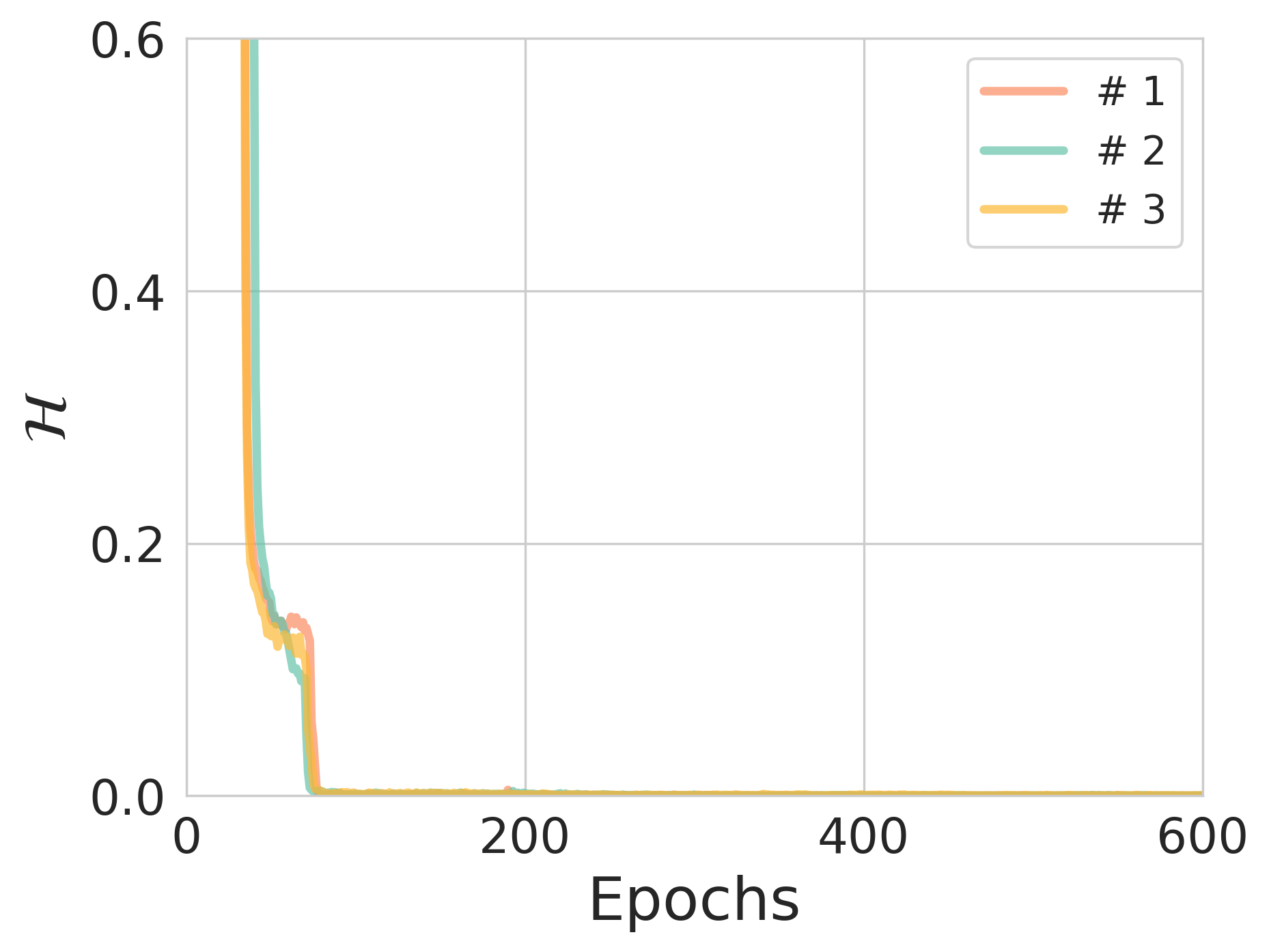}
        \includegraphics[width=1\linewidth]{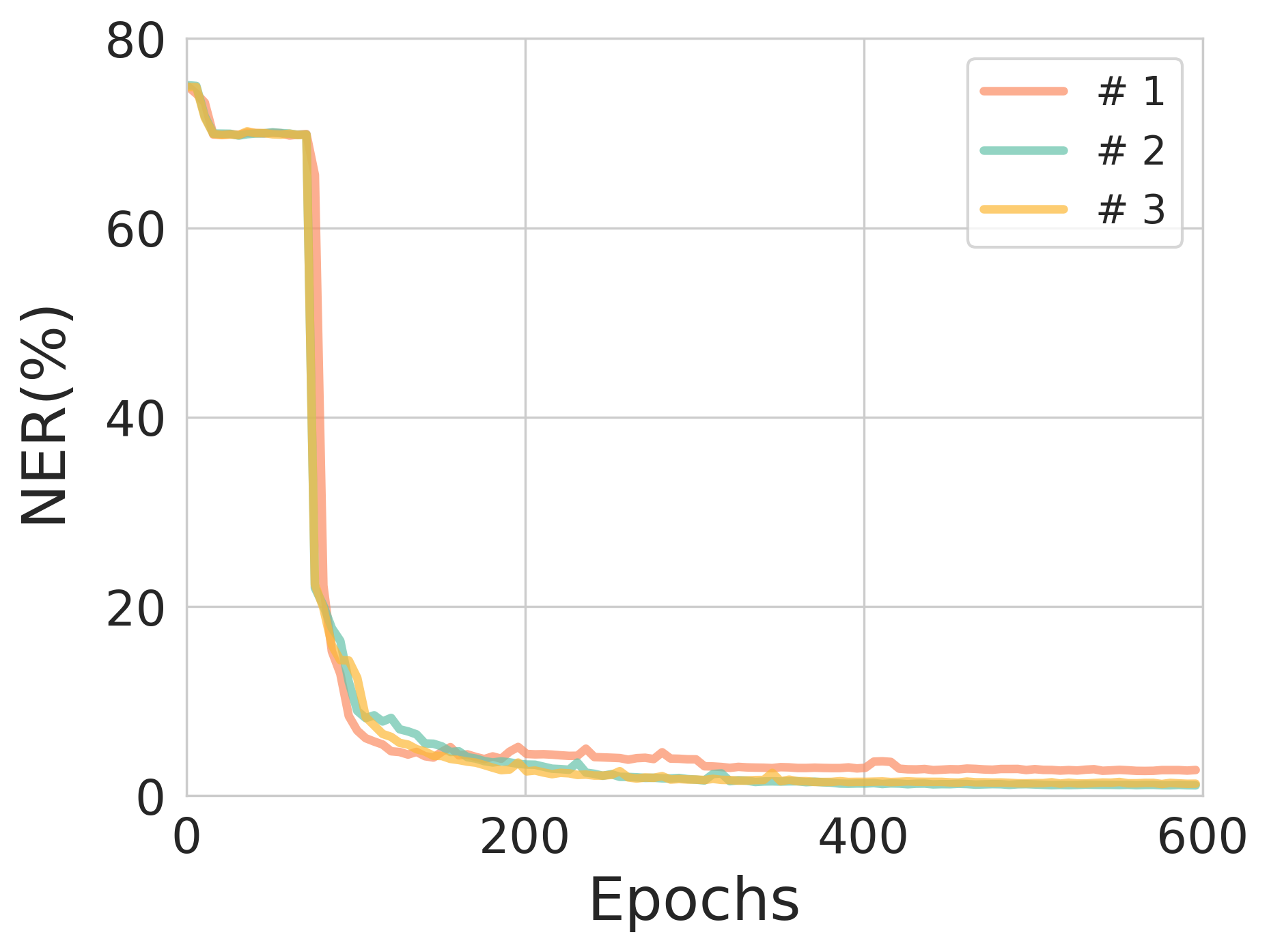}
    }%
    \subcaptionbox{$\tau=1$}[0.3\linewidth]
    {
        \includegraphics[width=1\linewidth]{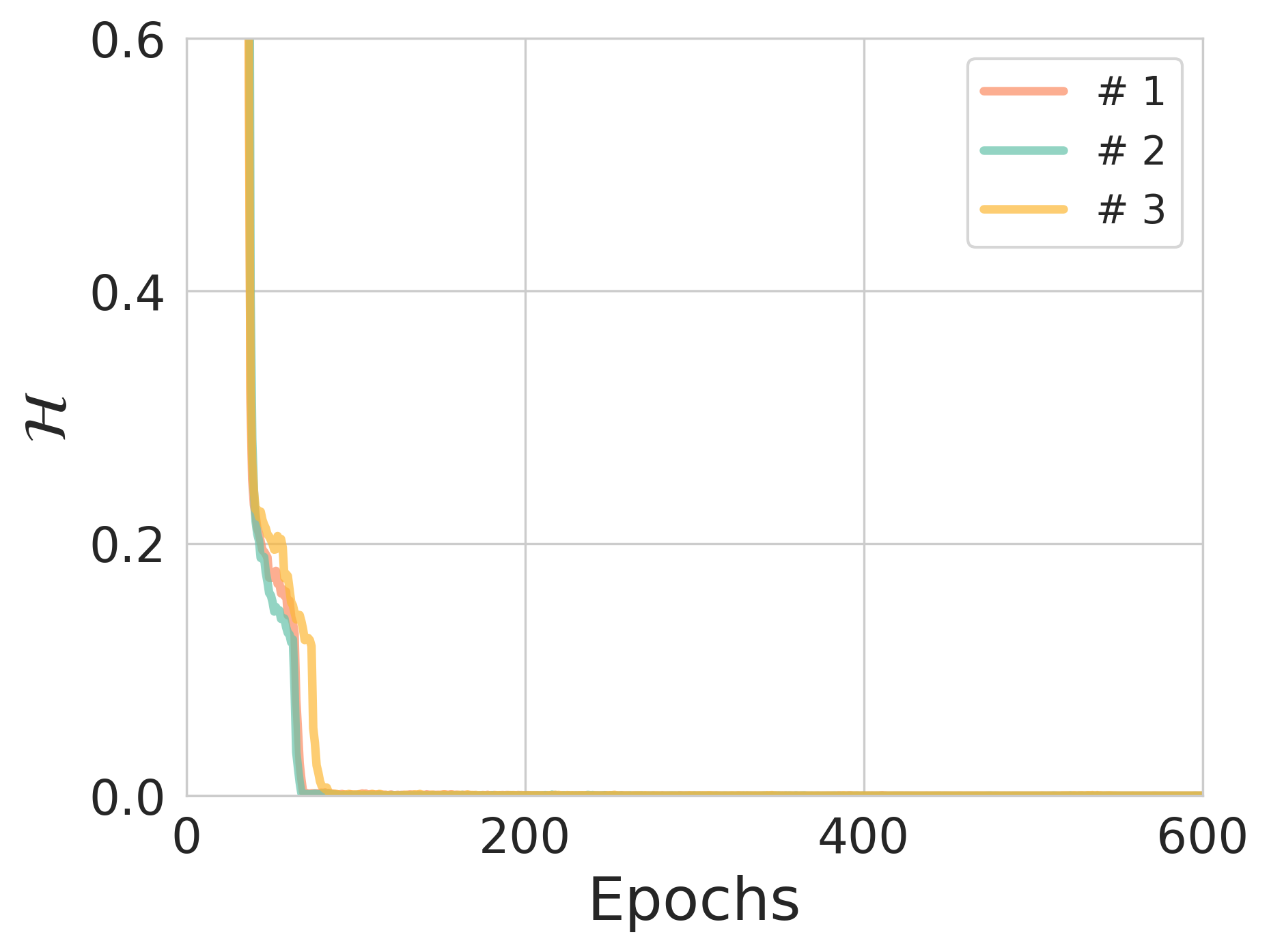}
        \includegraphics[width=1\linewidth]{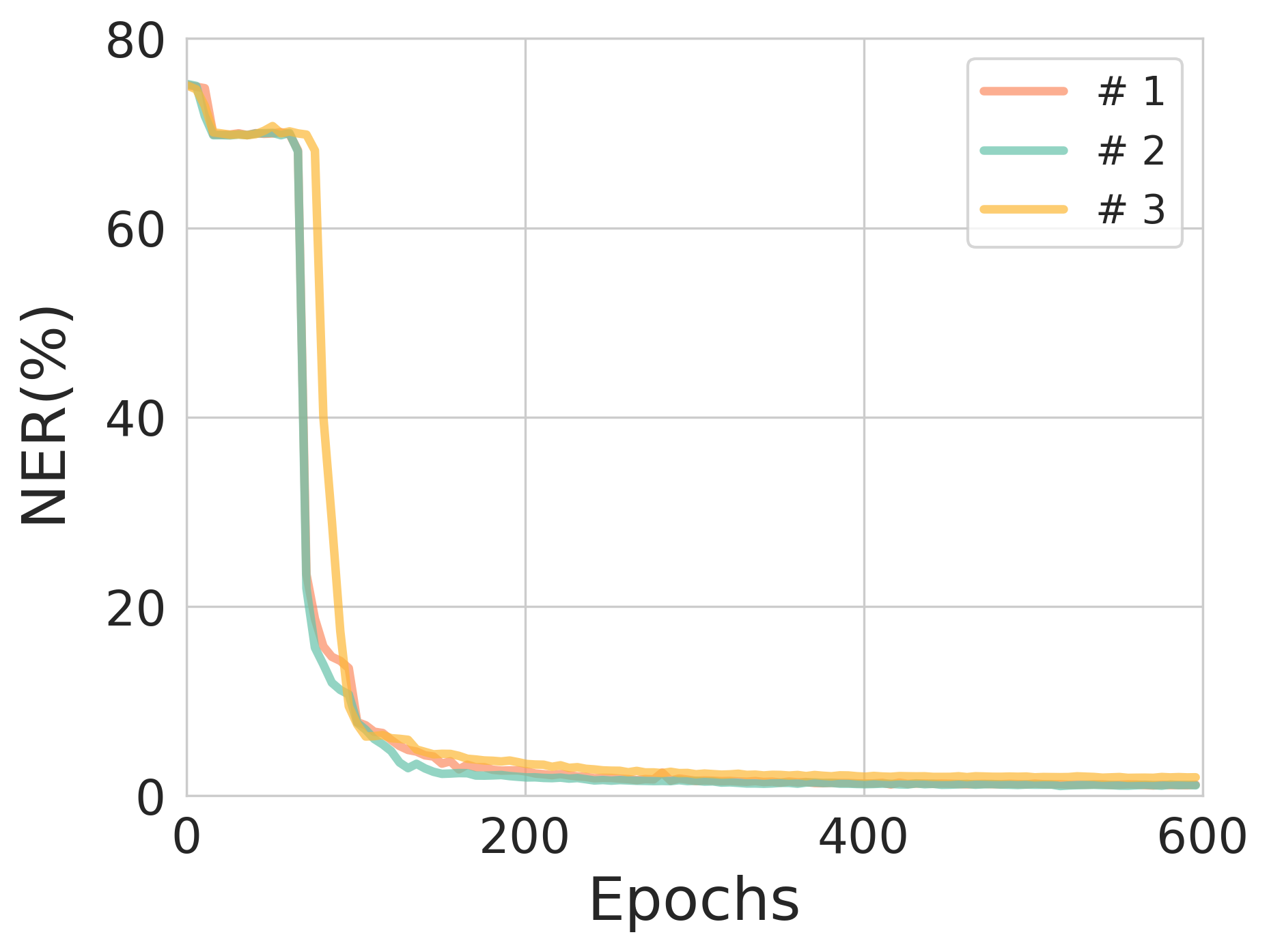}
    }\\
    \subcaptionbox{$\tau=2$}[0.3\linewidth]
    {
        \includegraphics[width=1\linewidth]{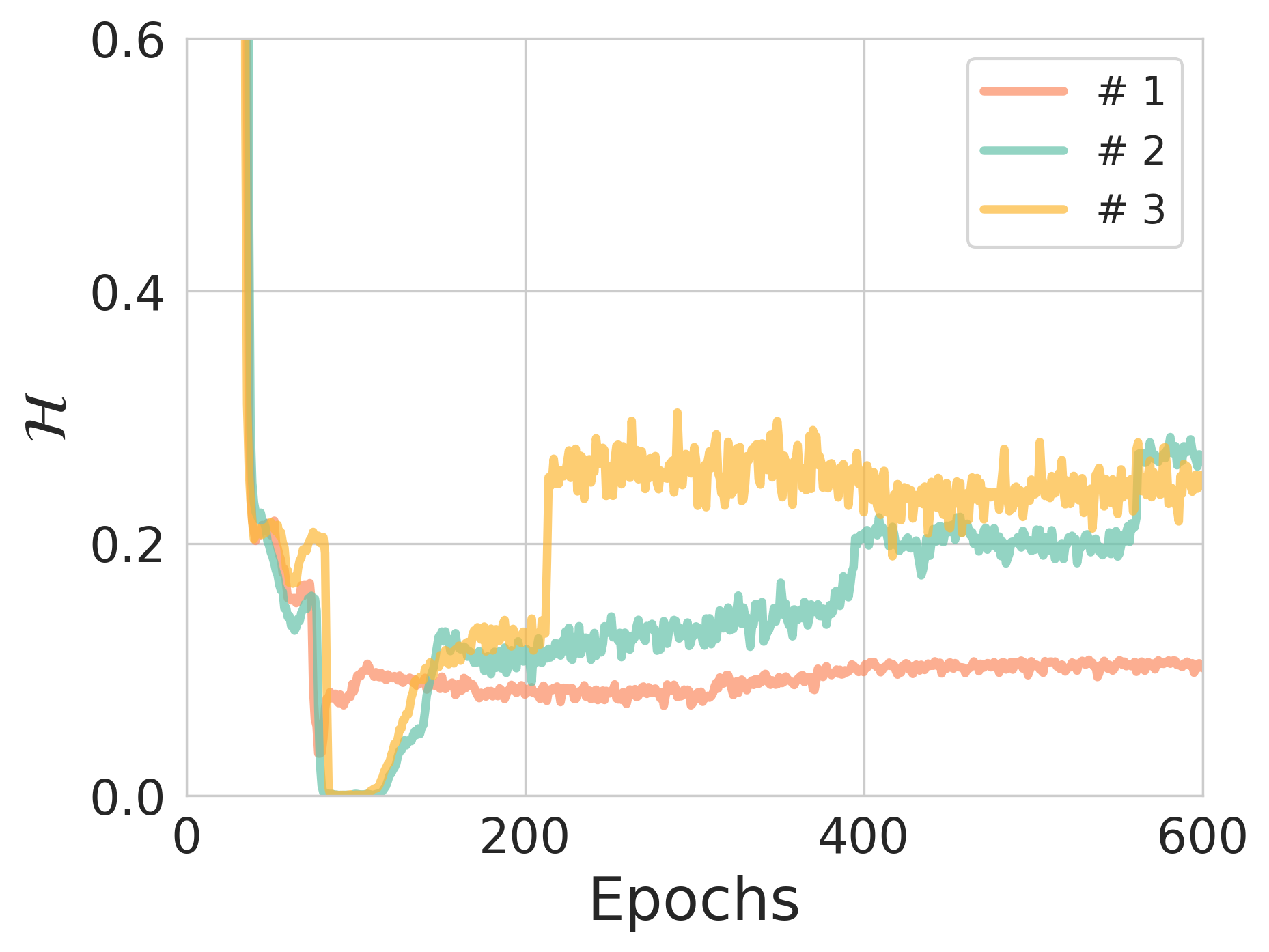}
        \includegraphics[width=1\linewidth]{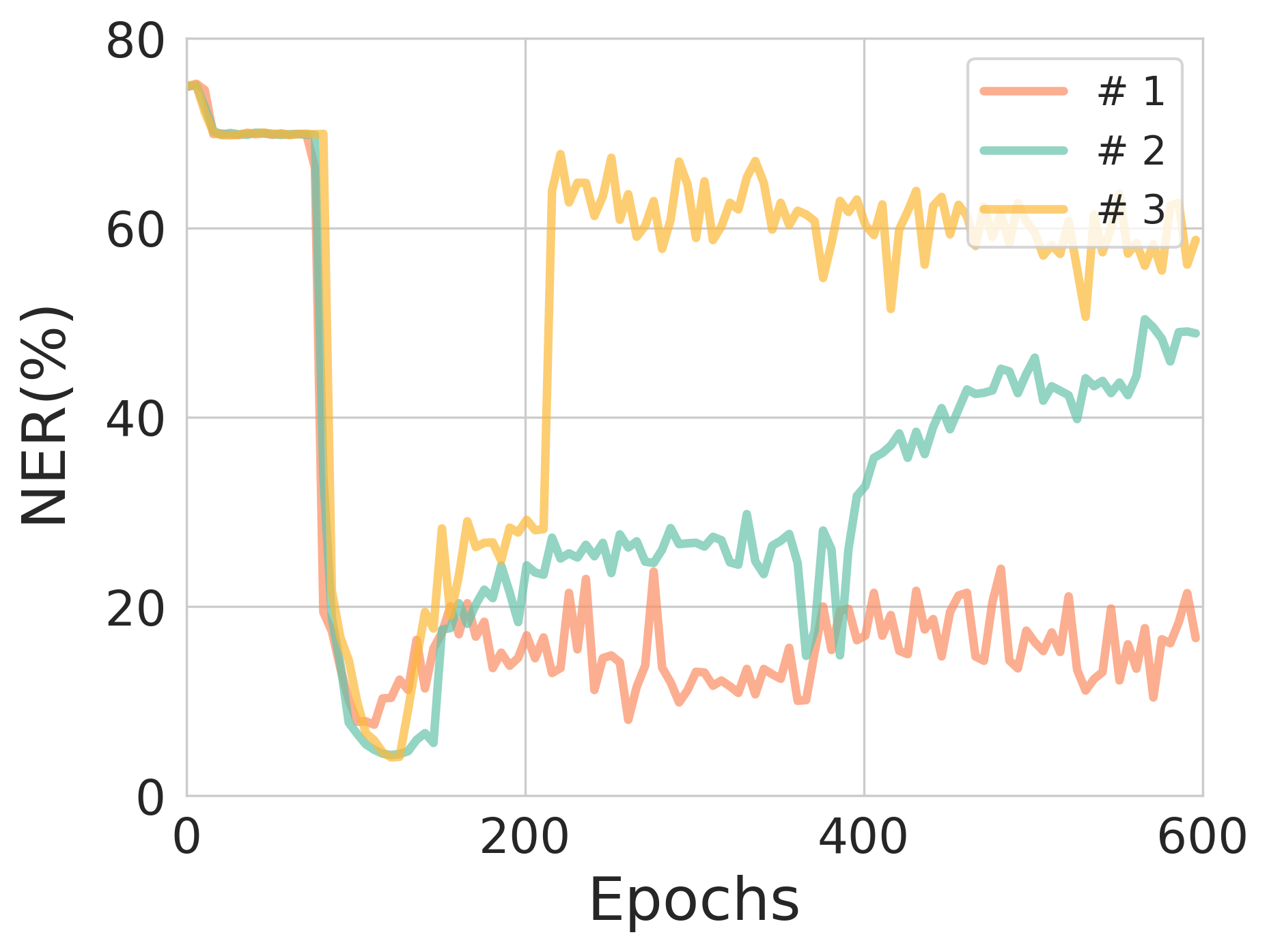}
    }%
    \subcaptionbox{$\tau=4$}[0.3\linewidth]
    {
        \includegraphics[width=1\linewidth]{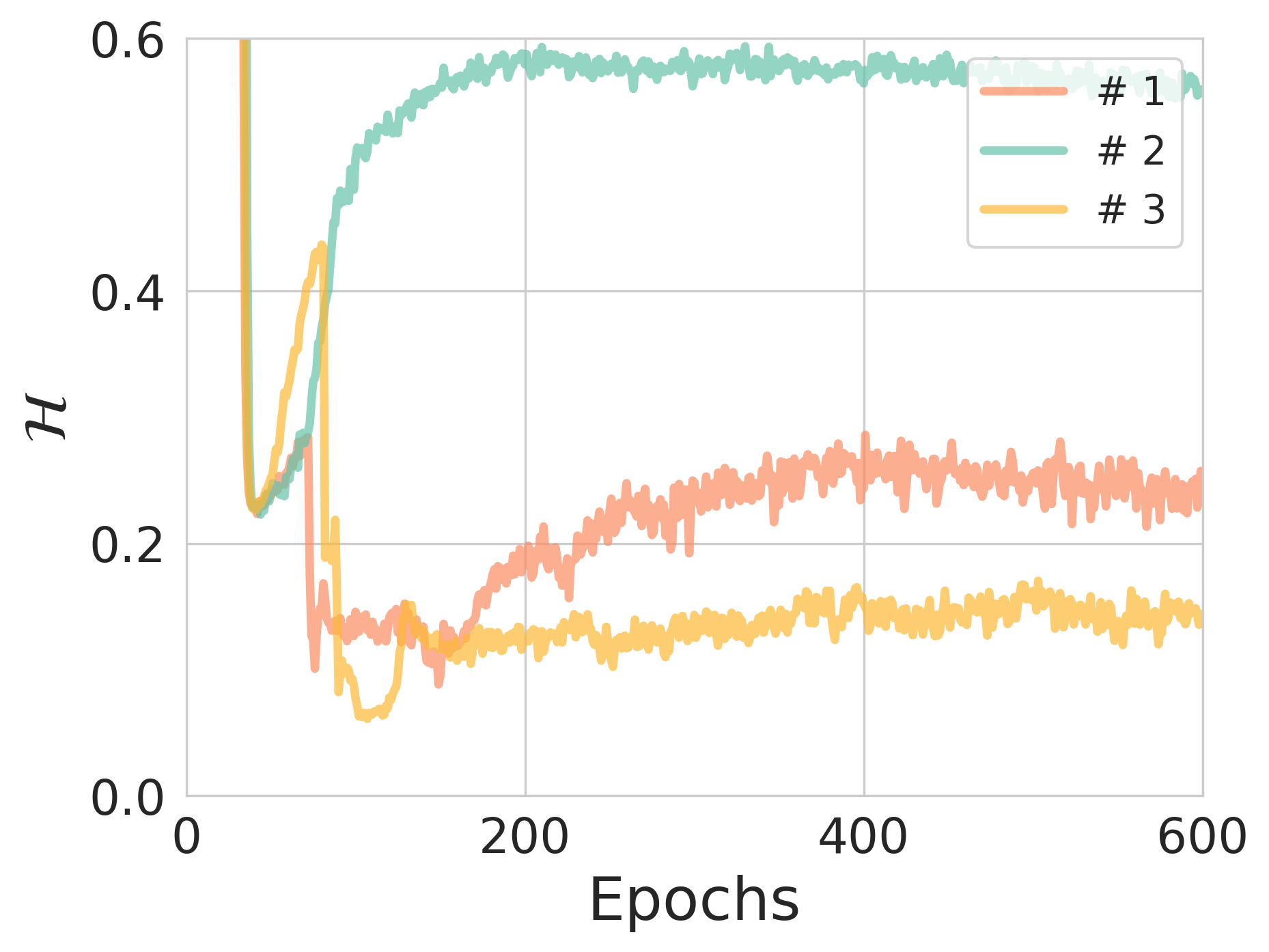}
        \includegraphics[width=1\linewidth]{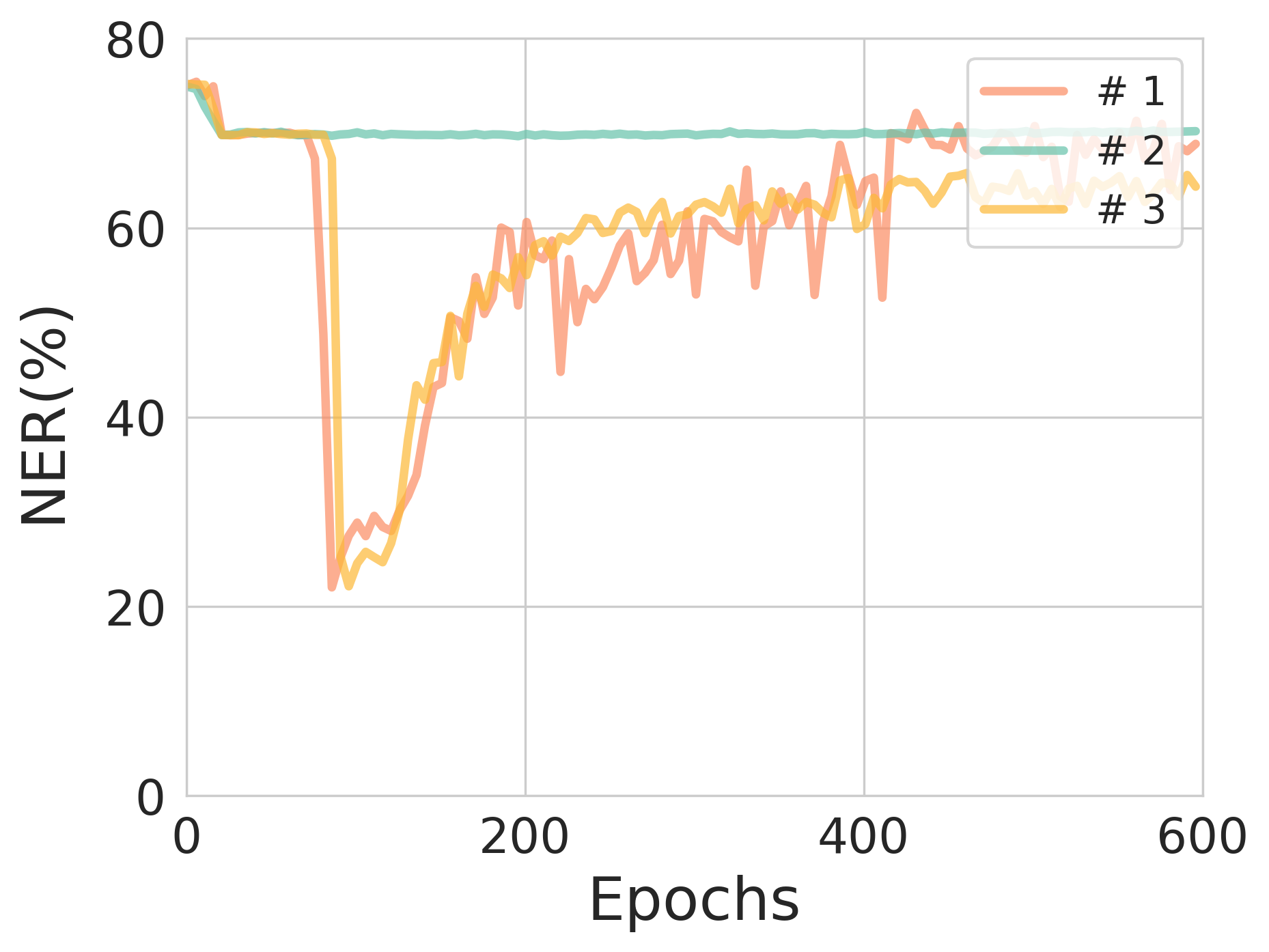}
    }%
    \subcaptionbox{$\tau=8$}[0.3\linewidth]
    {
        \includegraphics[width=1\linewidth]{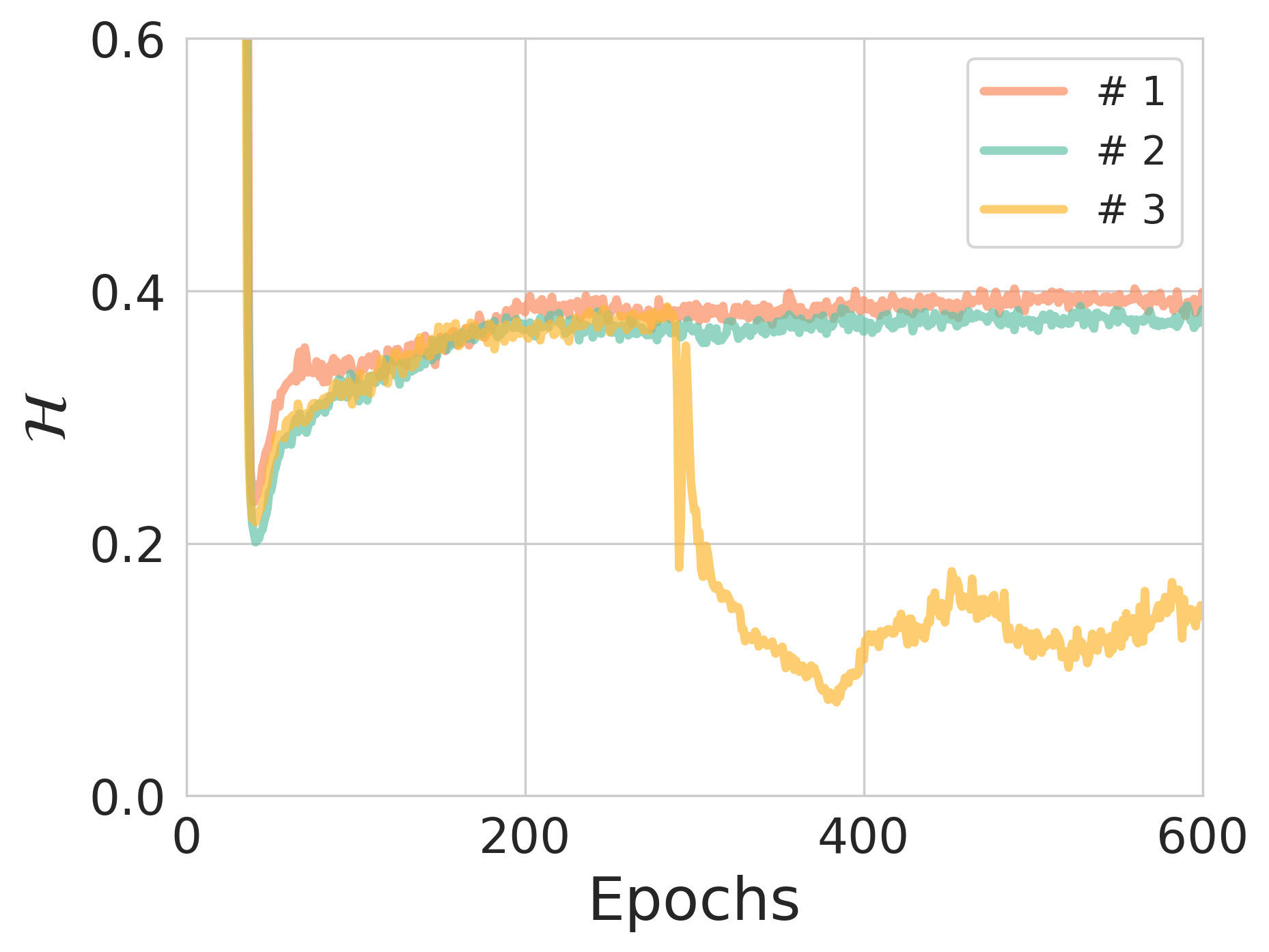}
        \includegraphics[width=1\linewidth]{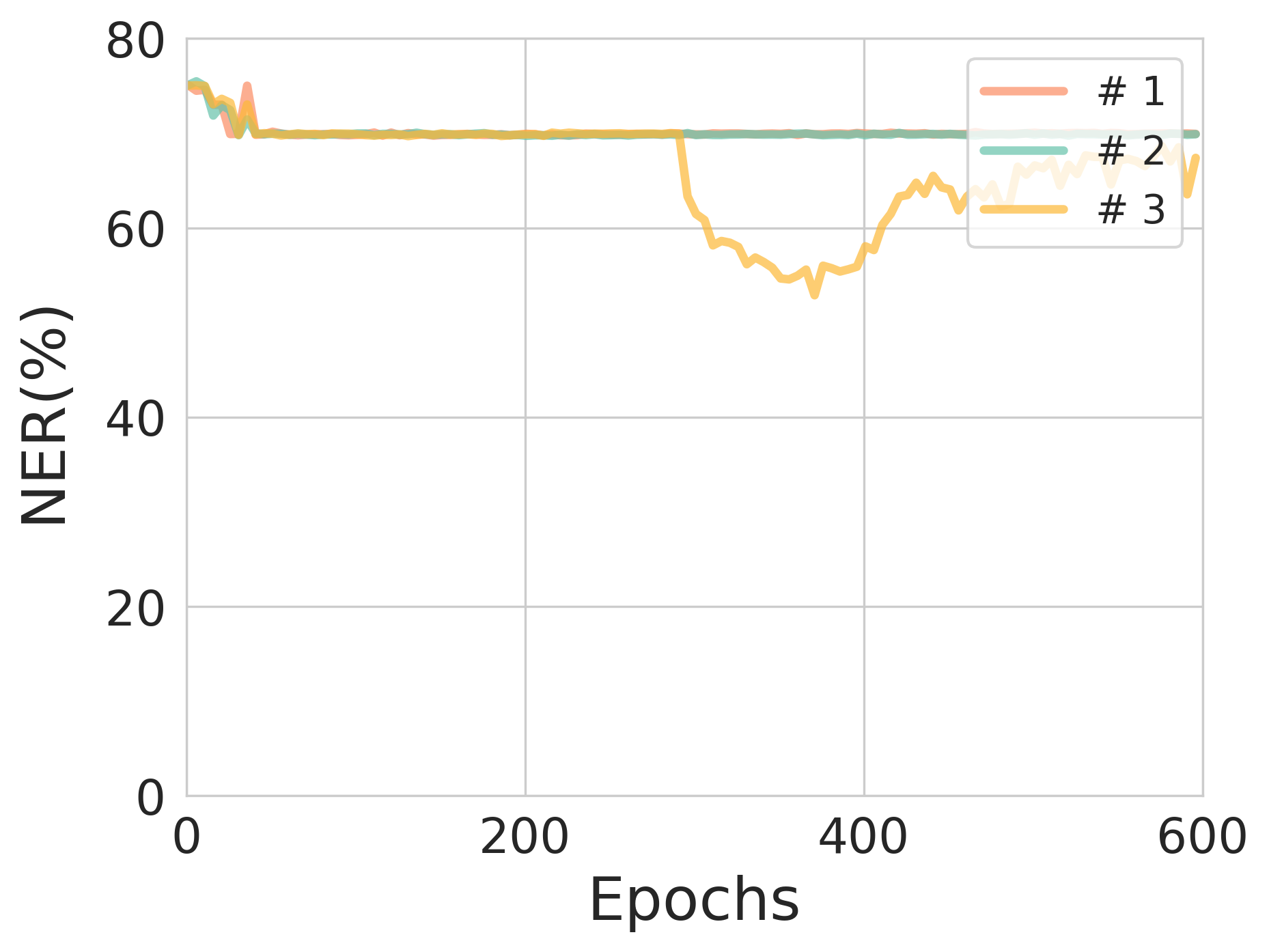}
    }%
    \caption{The codeword entropy $\mathcal{H}$ and 
    the validation NER for various choices of $\tau\in\{0.25,0.5,1,2,4,8\}$. 
    Each curve in the subfigures represents one of the $3$ runs conducted in the experiment 
    and is plotted against the training epochs. 
    }
    \label{fig:gumbelt}
    \vskip -0.in
\end{figure*}

\subsection{Potential alternative: plain entropy constraint}\label{app:alter}
We explore applying the entropy constraint~\cite{kossentini1993entropy} 
on the codeword $\bm{c}$ in the form of a probability vector. 
This constraint penalizes probability vectors deviating from the one-hot style. 
The Shannon entropy~\cite{shannon1948mathematical} in \cref{eqn:entropy} of a discrete distribution $P(x)$ represents the 
average level of ``information'' associated with the possible outputs of the variable $x$. 
It is easy to verify that \cref{eqn:entropy} is non-negative and 
equals zero only when the random variable produces a certain output. 
Considering this, the entropy constraint is defined on the codeword 
$\bm{c}=(\bm{c}_1,\bm{c}_2,\ldots,\bm{c}_k)$ 
as follows
\begin{equation}\label{eqn:lossce}
    \mathcal{L}_{\mathrm{EN}}(\bm{c}) = -\sum_{i}\sum_j c_{ij}\log{c_{ij}}.
\end{equation}

In the following experiments, the \cref{eqn:lossce} is integrated into the 
overall optimization objective with a weighting parameter $\lambda$, 
and the disturbance-based discretization is removed from the framework. 

Imposing the entropy constraint on the codewords compels them to adopt a one-hot style 
from a probability vector style, which is a double-edged sword. 
This approach, compared to plain quantization, preserves the capability of gradient propagation. 
Moreover, 
codewords that closely resemble a one-hot representation mitigate the domain difference between the decoder's input 
during the training and the testing phases. 

\begin{figure*}[htb!]
    \vskip 0.in
    \centering
    \subcaptionbox{$\lambda=0$}[0.2\linewidth]
    {
        \includegraphics[width=1\linewidth]{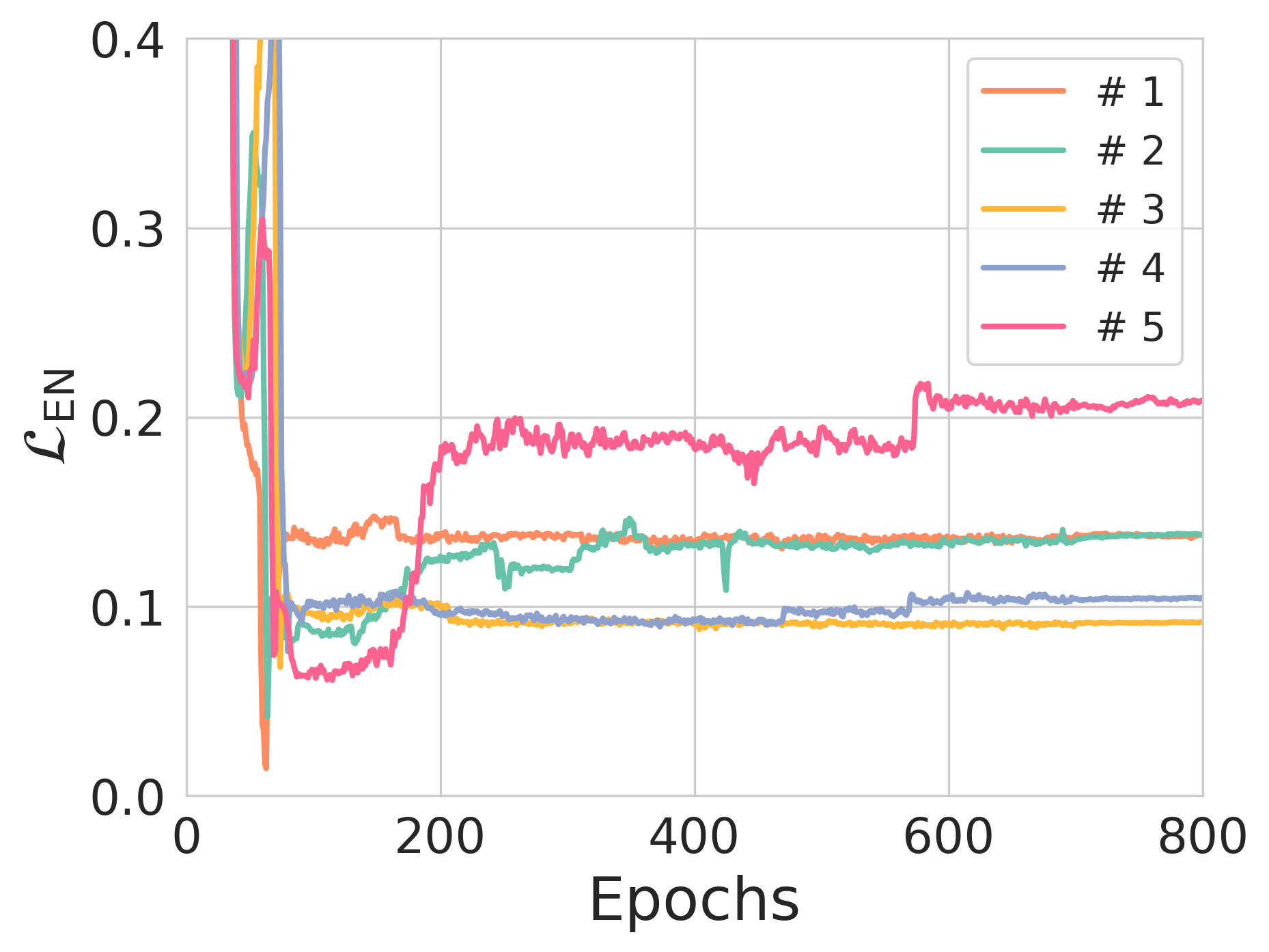}
        \includegraphics[width=1\linewidth]{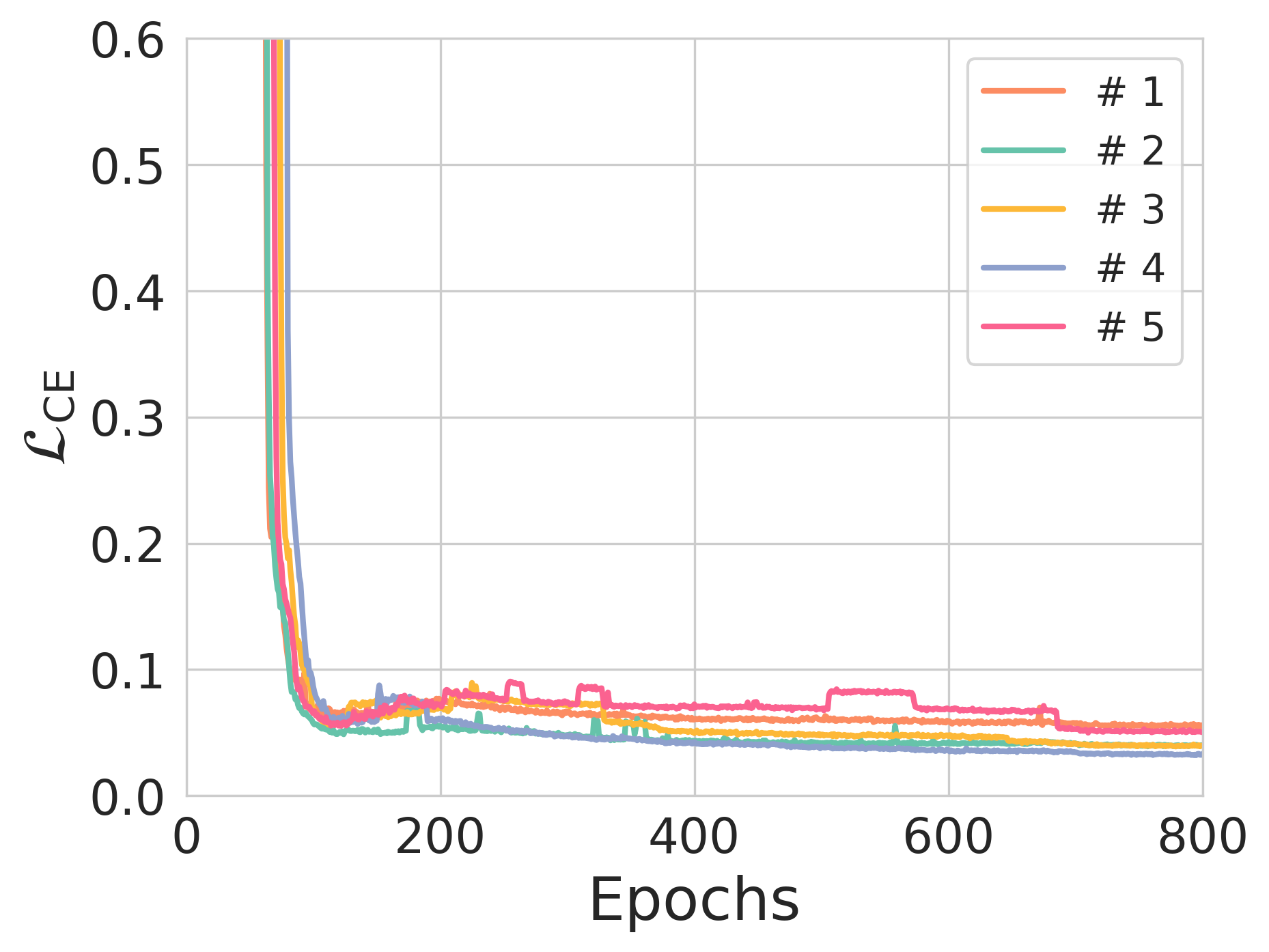}
        \includegraphics[width=1\linewidth]{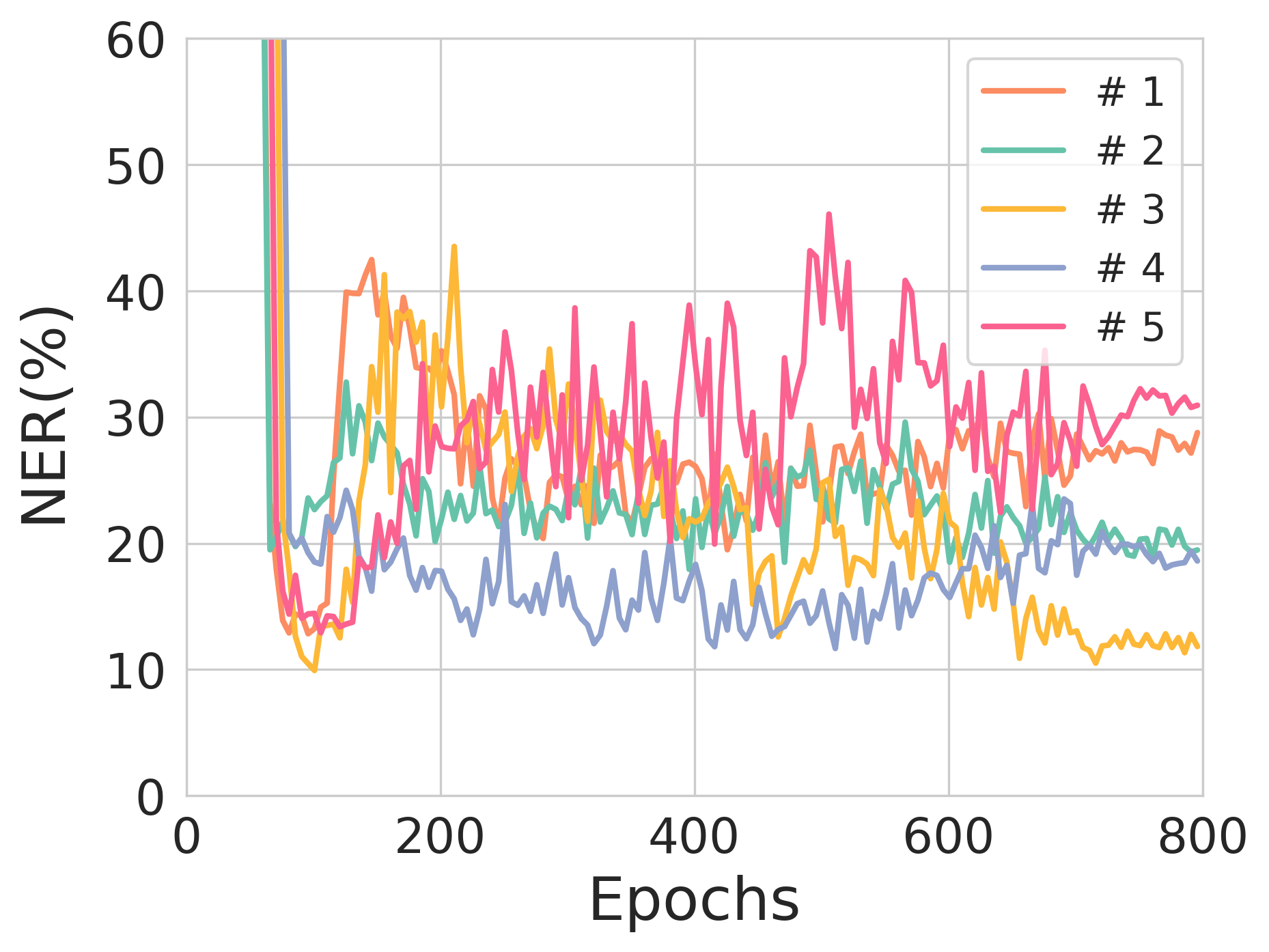}
    }%
    \subcaptionbox{$\lambda=0.001$}[0.2\linewidth]
    {
        \includegraphics[width=1\linewidth]{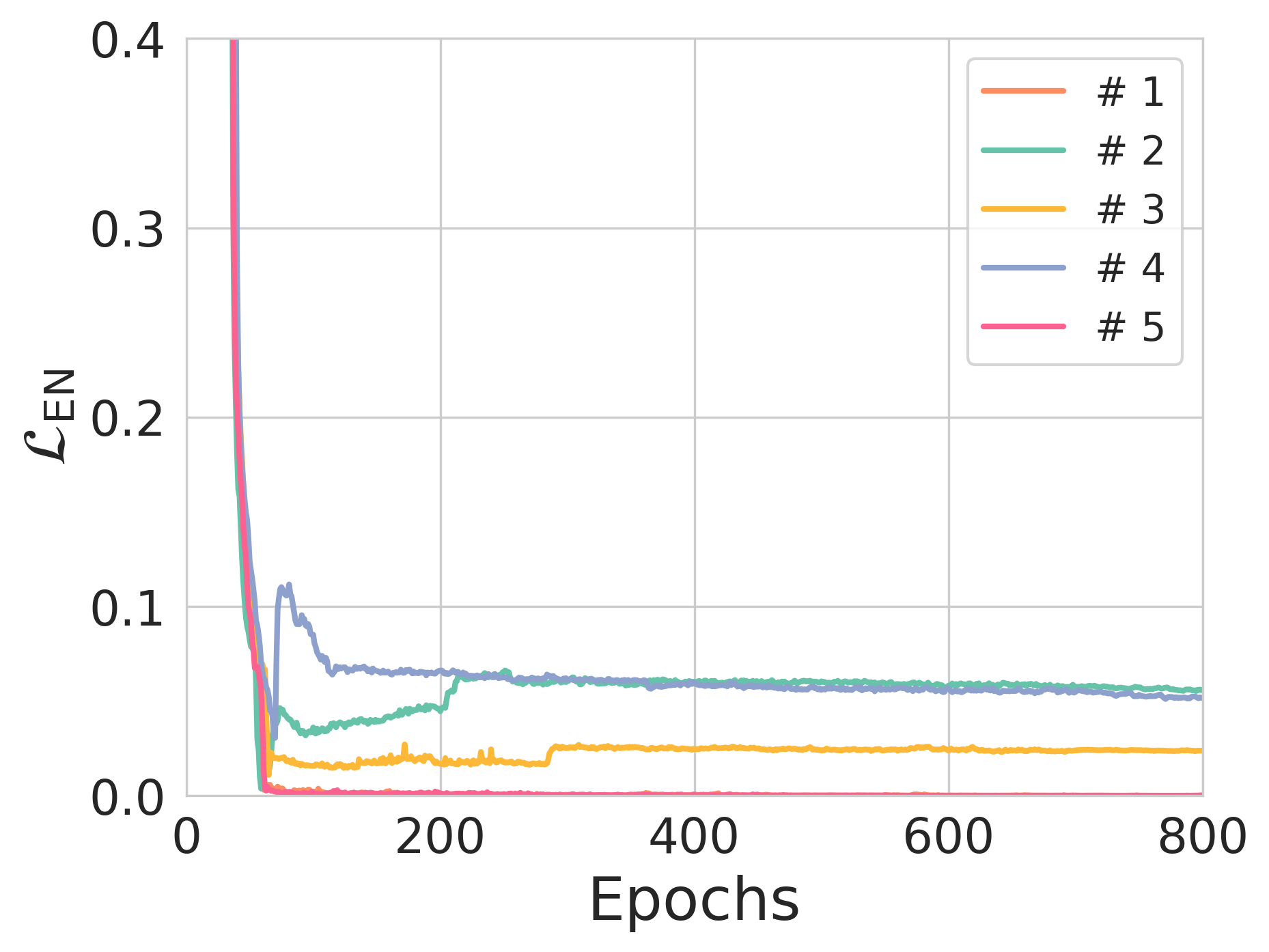}
        \includegraphics[width=1\linewidth]{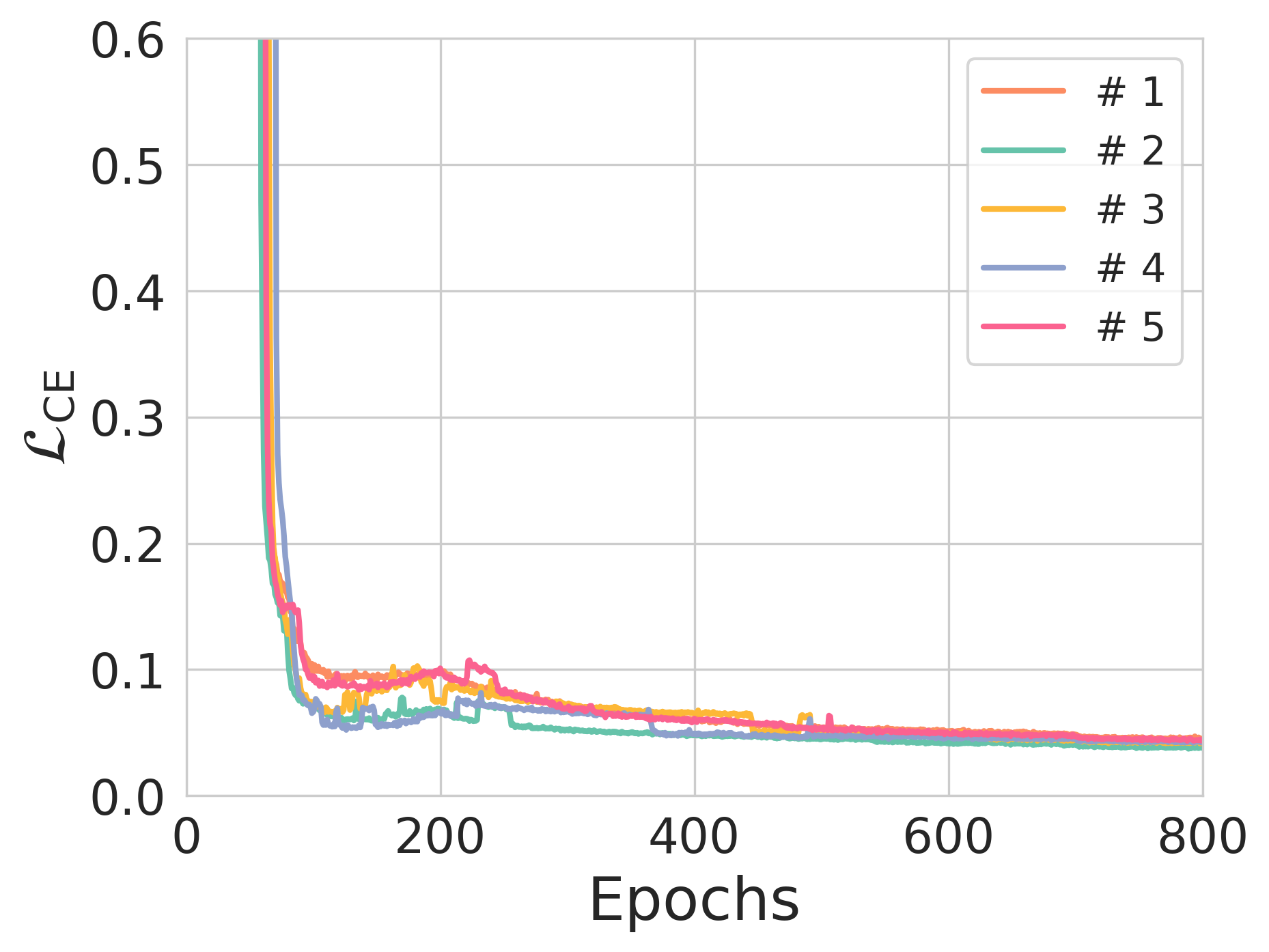}
        \includegraphics[width=1\linewidth]{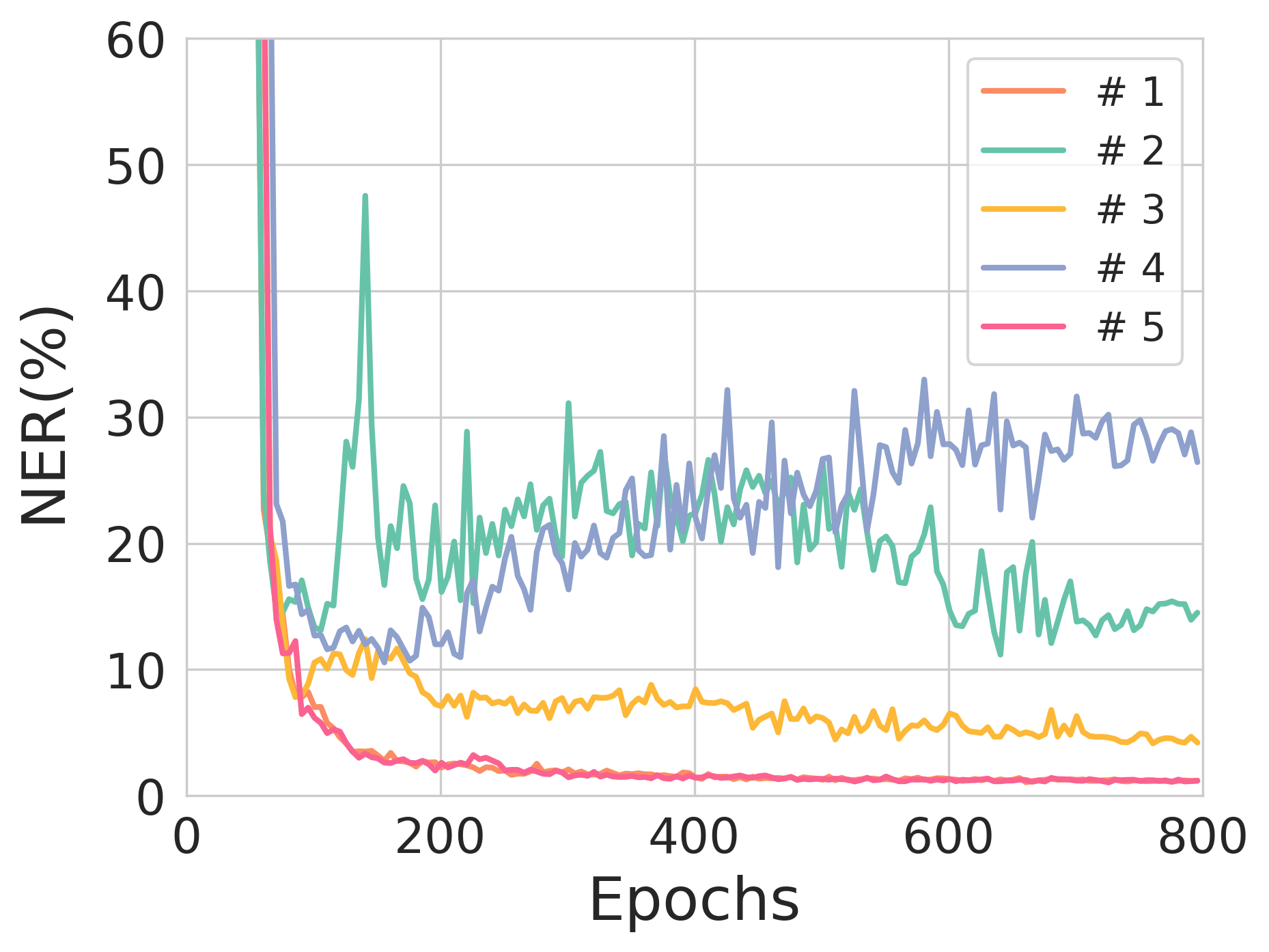}
    }%
    \subcaptionbox{$\lambda=0.01$}[0.2\linewidth]
    {
        \includegraphics[width=1\linewidth]{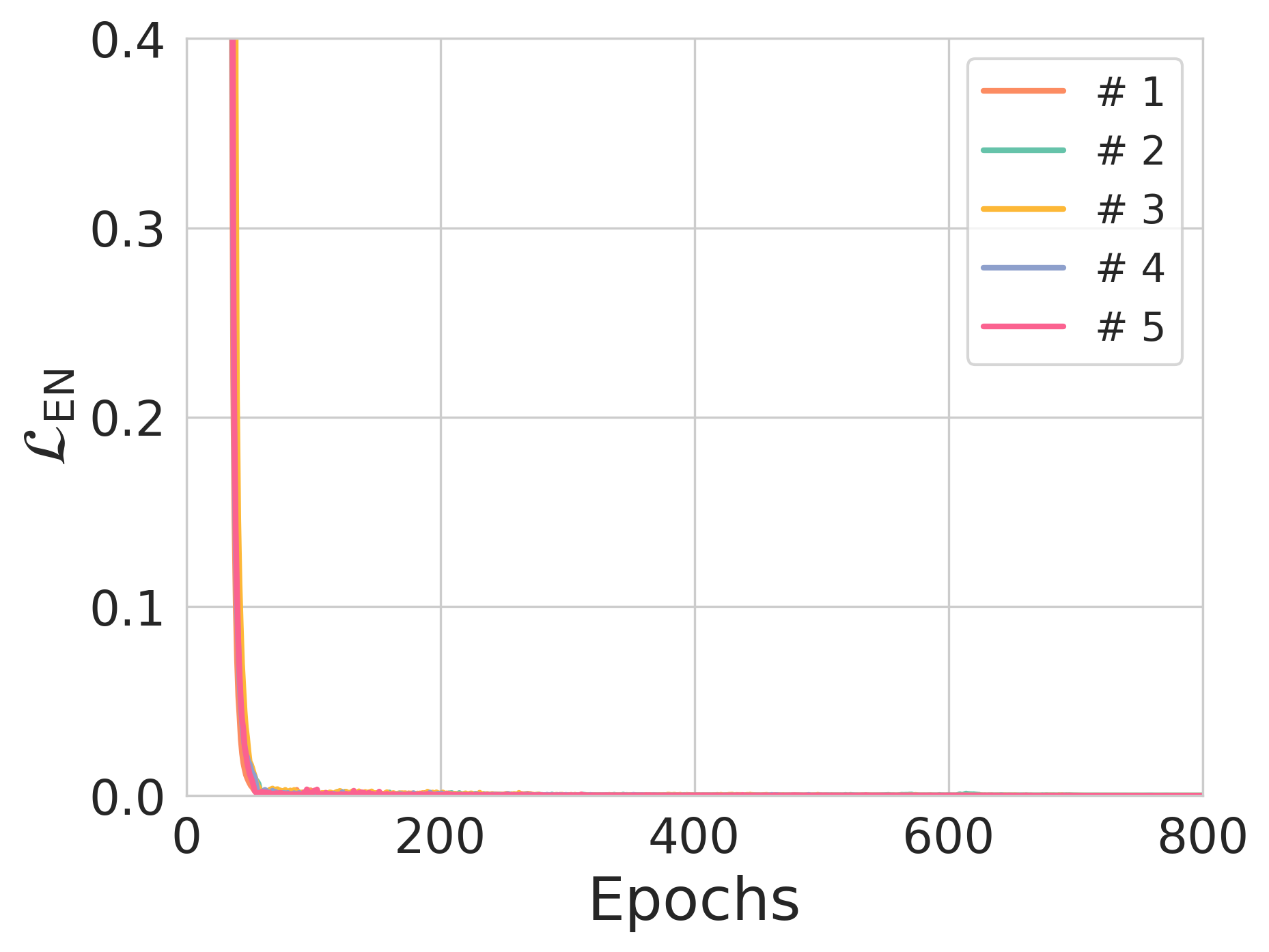}
        \includegraphics[width=1\linewidth]{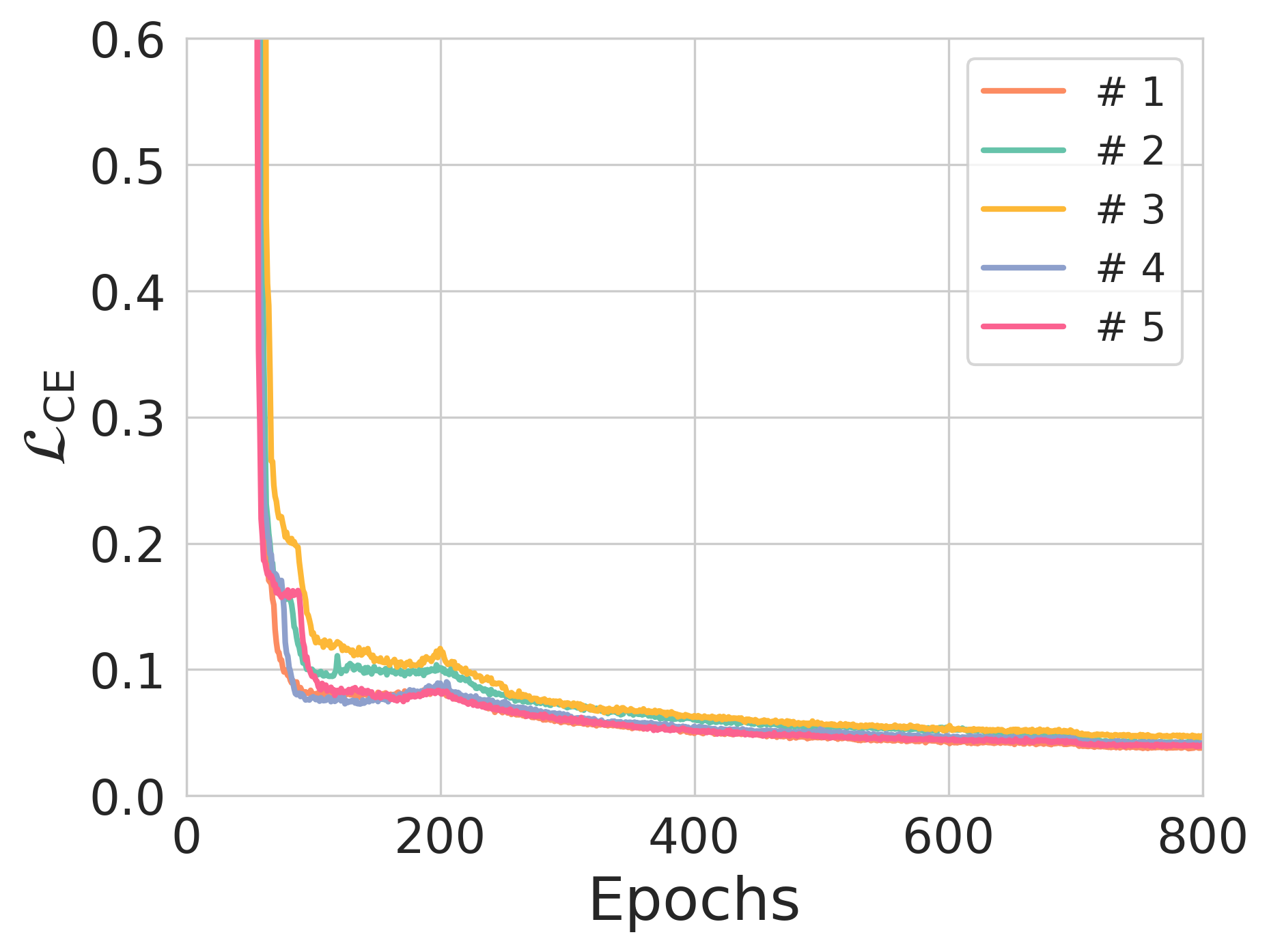}
        \includegraphics[width=1\linewidth]{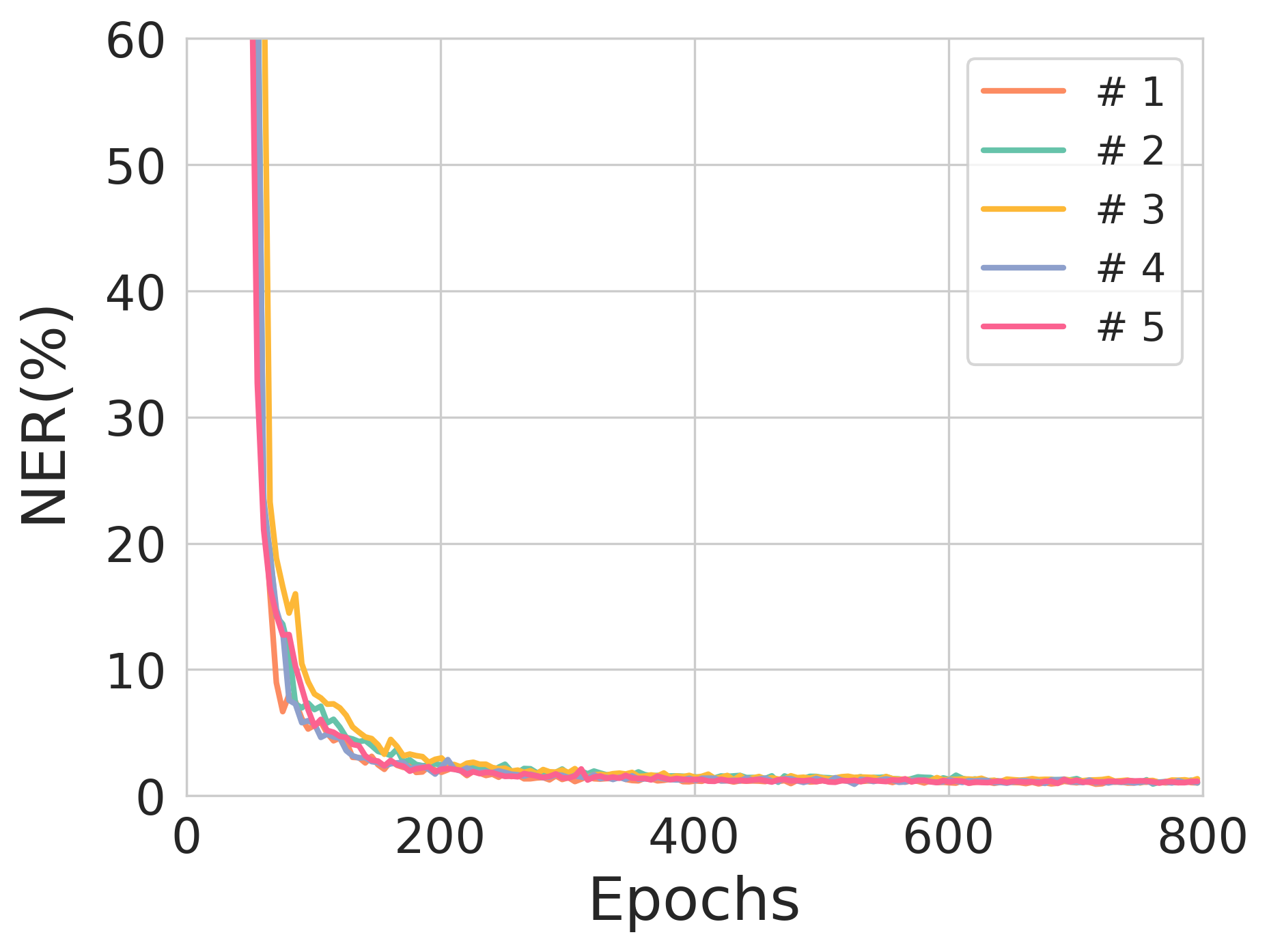}
    }%
    \subcaptionbox{$\lambda=0.1$}[0.2\linewidth]
    {
        \includegraphics[width=1\linewidth]{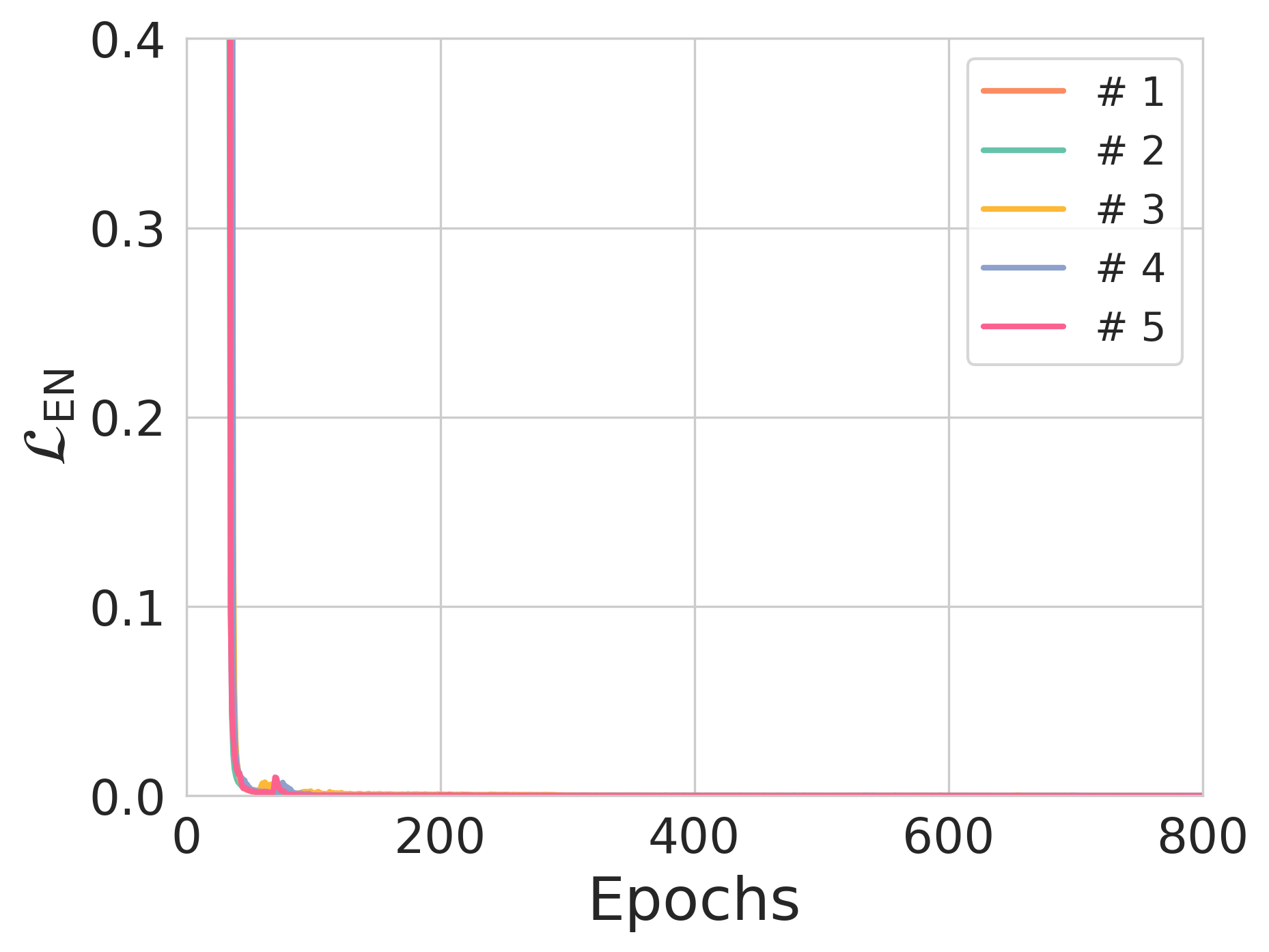}
        \includegraphics[width=1\linewidth]{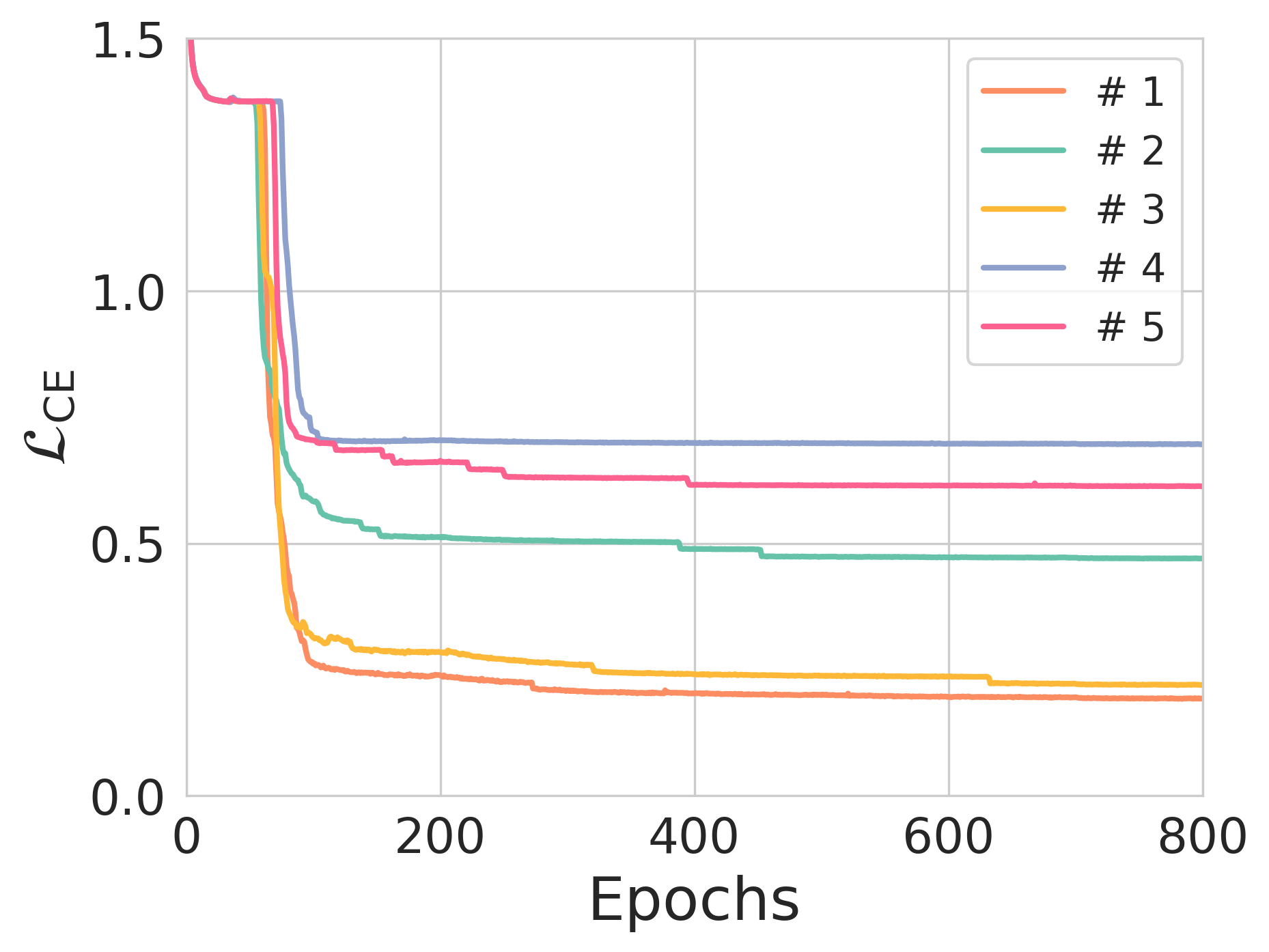}
        \includegraphics[width=1\linewidth]{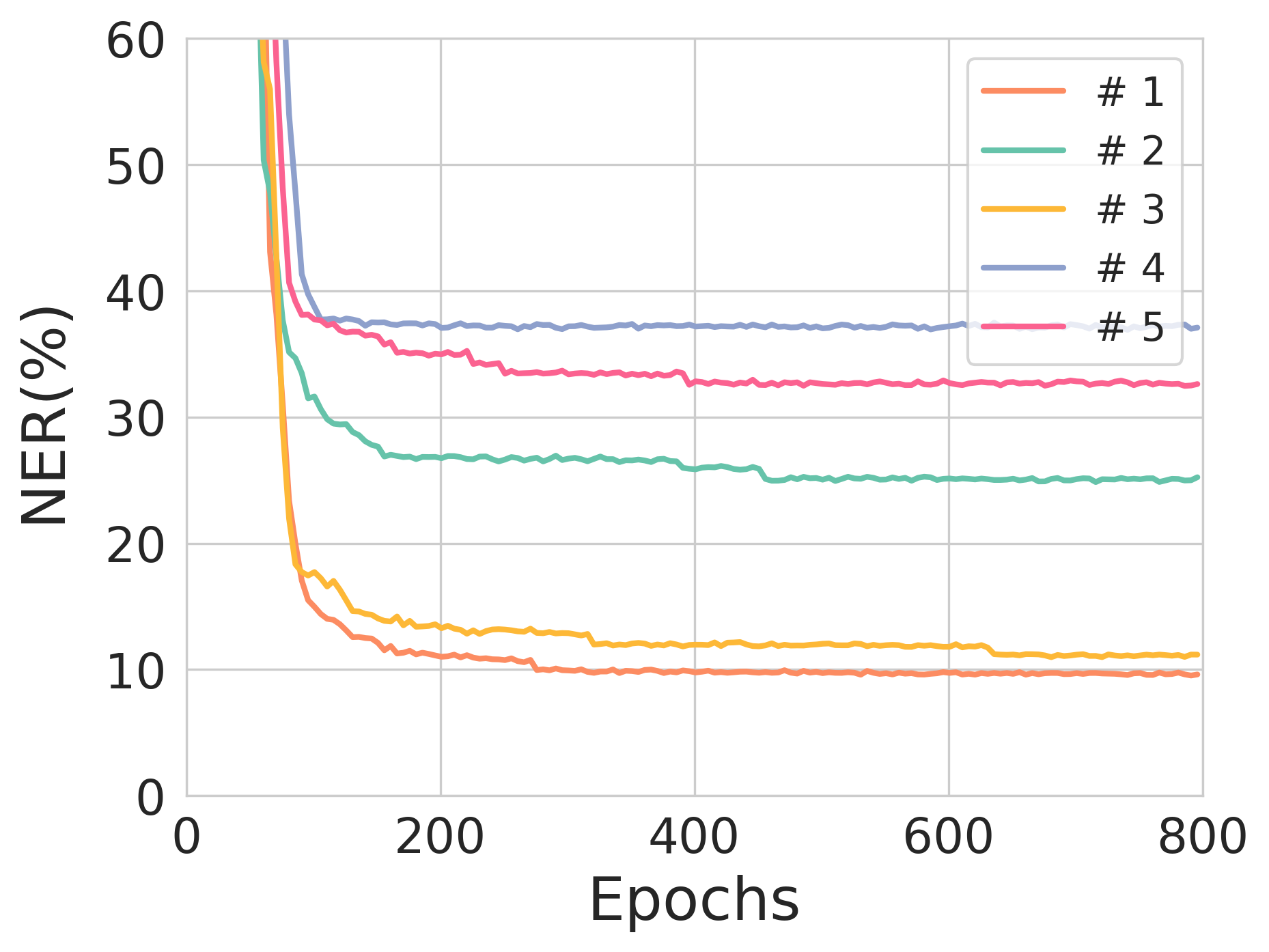}
    }%
    \subcaptionbox{$\lambda=1$}[0.2\linewidth]
    {
        \includegraphics[width=1\linewidth]{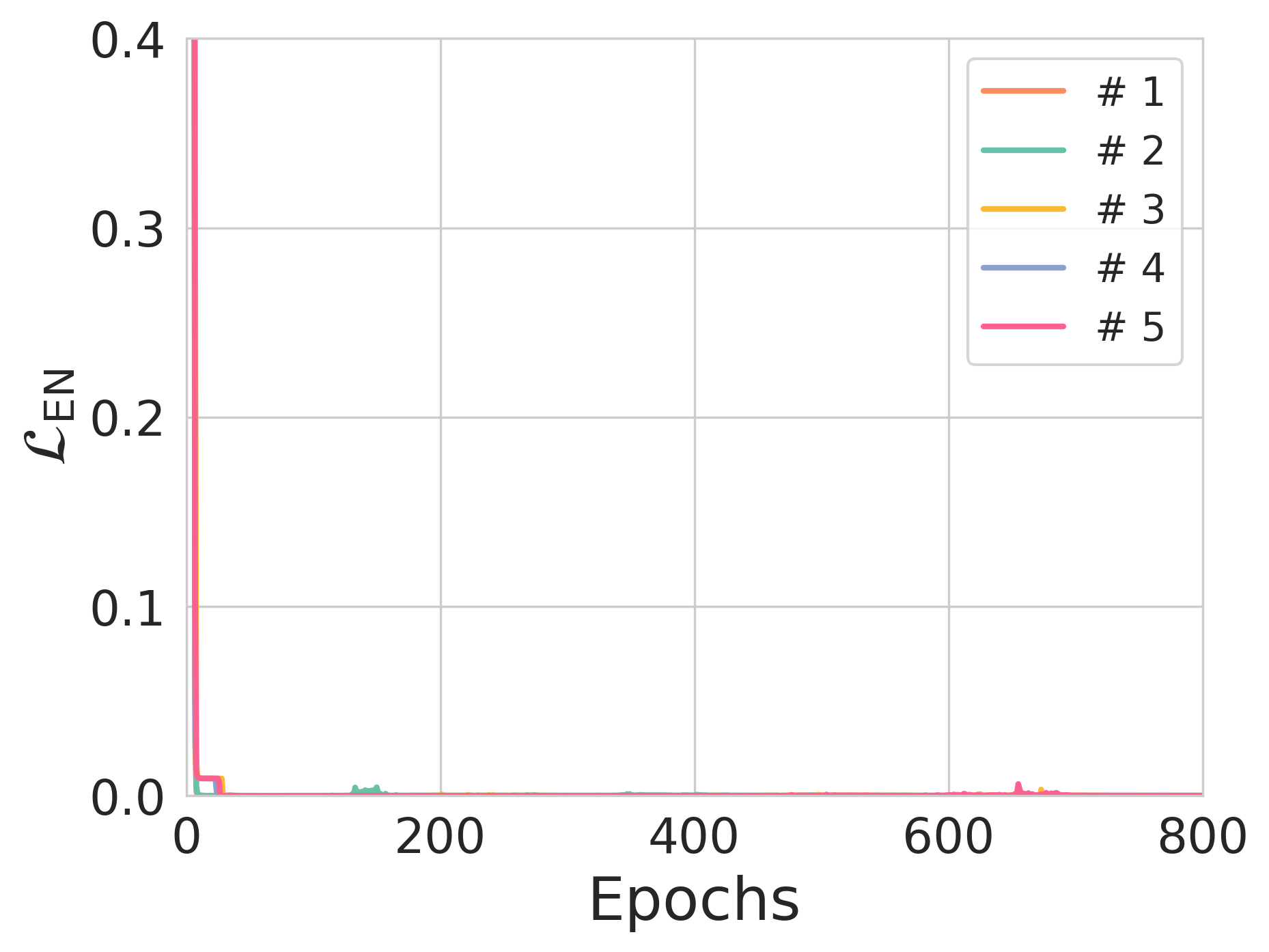}
        \includegraphics[width=1\linewidth]{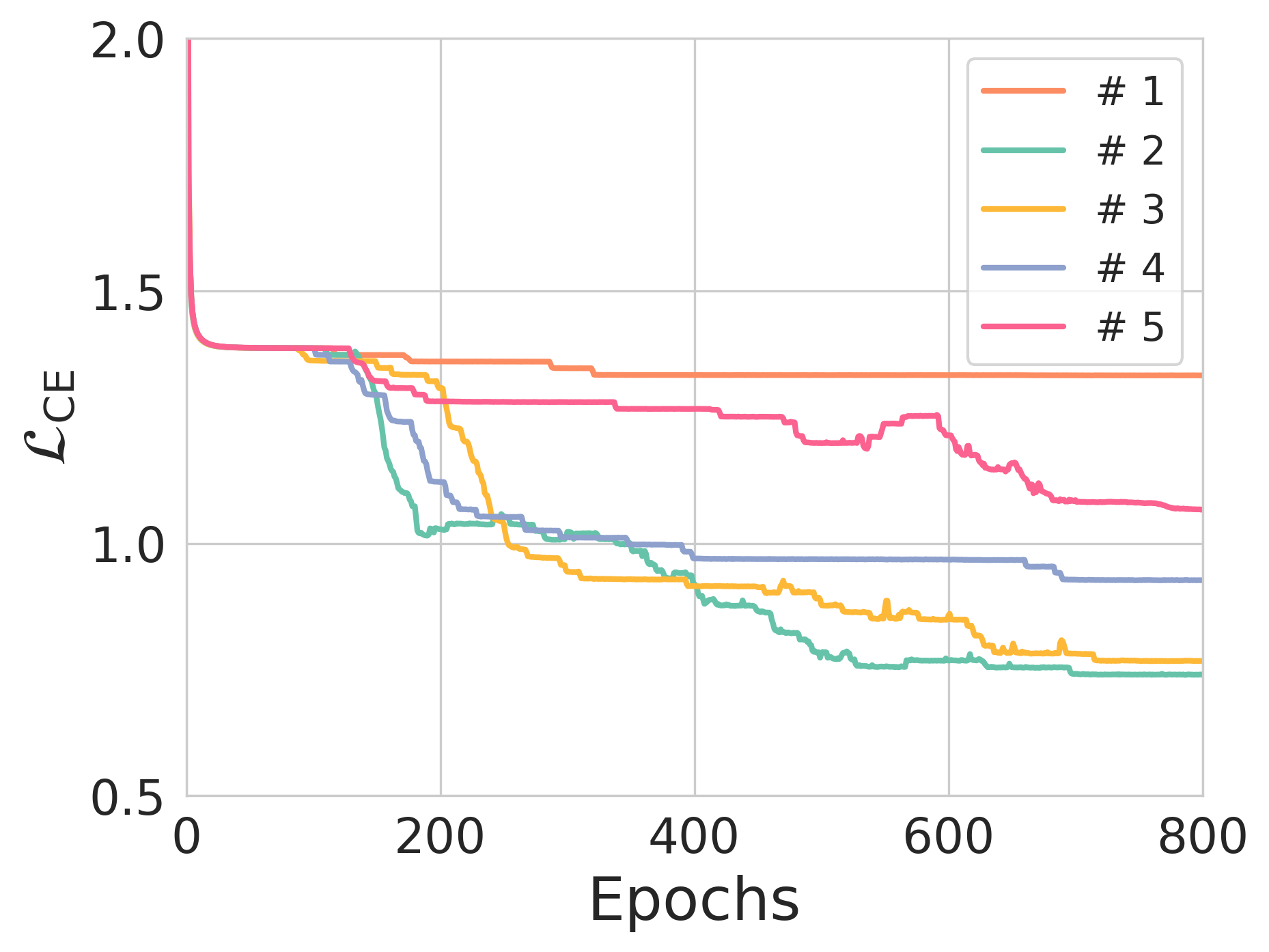}
        \includegraphics[width=1\linewidth]{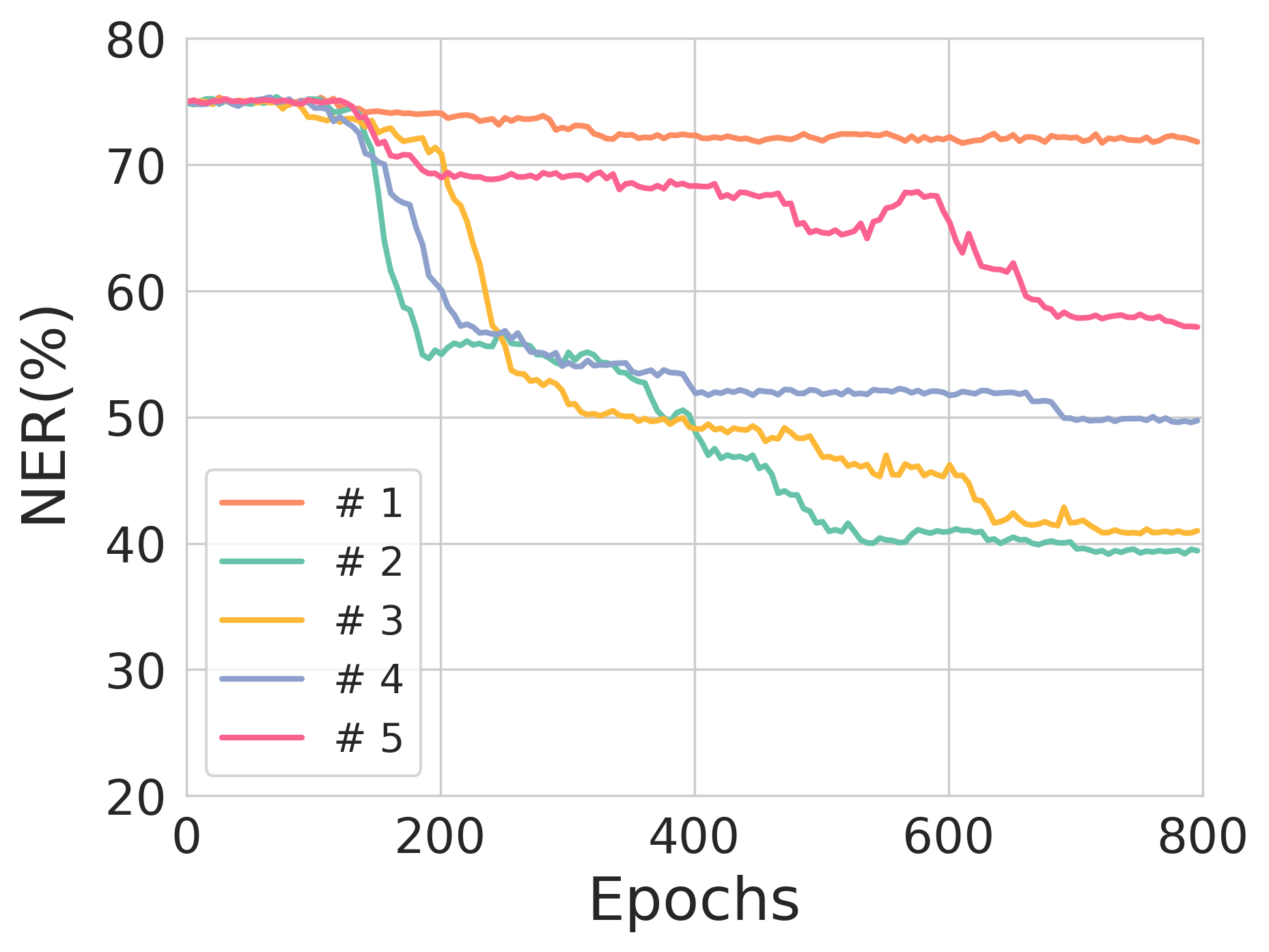}
    }%
    \caption{The entropy constraint $\mathcal{L}_\mathrm{EN}$, 
    the reconstruction loss $\mathcal{L}_\mathrm{CE}$, 
    and the validation NER for various choices of $\lambda\in\{0,0.001,0.01,0.1,1\}$. 
    Each curve in the subfigures represents one of the $5$ runs conducted in the experiment 
    and is plotted against the training epochs.
    }
    \label{fig:lambda}
    \vskip -0.in
\end{figure*}

However, the drawback is also non-negligible. 
If the entropy constraint is overapplied and dominates the model training before the autoencoder 
is effectively tuned to its intended function, 
the entropy constraint may tend to produce a gradient opposite to the loss-propagated gradient, 
leading the model towards a local minimum convergence point. 
As a thought experiment, consider a letter $\bm{c}_i=(0,0,0.1,0.9)$ from a codeword $\bm{c}$ 
where the letter is not aligned with the IDS-correcting aim. 
The optimization on the reconstruction loss $\mathcal{L}_\mathrm{CE}$ of \eqref{eqn:ce-loss} 
may propagate a gradient decreasing the fourth dimension $c_{i4}=0.9$. 
However, the entropy constraint on $\bm{c}_i$ will faithfully produce an opposite gradient, 
attempting to increase $c_{i4}=0.9$ to achieve low entropy, 
potentially hindering the optimization process. 

Therefore careful tuning of the entropy constraint weight is needed. 
Experiments were conducted with different choices of weight $\lambda$ in $\{0,0.001,0.01,0.1,1\}$, 
and the results are presented in \cref{fig:lambda}. 
Each column in this figure corresponds to a specific $\lambda$ value, 
and the rows depict the entropy $\mathcal{L}_{\mathrm{EN}}$, 
the reconstruction $\mathcal{L}_{\mathrm{CE}}$ loss between $\hat{\bm{s}}$ and $\bm{s}$, 
and the validation NER, from top to bottom. 

The first row suggests that the entropy $\mathcal{L}_\mathrm{EN}$ is controlled by enlarging the constraint weight, 
while the second row shows that the reconstruction loss $\mathcal{L}_\mathrm{CE}$ diverges with an overapplied entropy constraint. 
The third row of NER indicates that the performance is improved by introducing the entropy constraint and worsened 
by further enlarging the constraint weight. 

Column-wise, when using $\lambda=0$ and $\lambda=0.001$, 
the reconstruction $\mathcal{L}_\mathrm{CE}$ converges well, 
indicating that the autoencoder is well-trained. 
However, the entropy on the codeword remains relatively high during the training phase, 
suggesting that the codewords are diverse from one-hot style. 
The curves of validation NERs also support this speculation. 
Although the autoencoder is well-trained, the quantization of the codeword with high entropy 
alters the input domain for the decoder, and fails the testing phase. 
Regarding the columns corresponding to $\lambda=0.1$ and $\lambda=1$ in \cref{fig:lambda}, 
we observe that the entropy drops to a low level fast in the first few epochs, 
and the reconstruction loss $\mathcal{L}_\mathrm{CE}$ does not decline to an appropriate interval during the training phase. 
This verifies our conjecture that an overapplied entropy constraint will lead the model to a local minimum convergence point. 
Overall, the $\lambda=0.01$ is a proper choice for the constraint weight in our experiments. 
The model converges well, 
the entropy maintains at a low level, and the decoder keeps its performance with the quantized codewords. 

Based on the above analysis, while applying a plain entropy constraint may serve a similar purpose as disturbance-based discretization, 
it is less robust and requires careful tuning.

\section{Experiments on the Differentiable IDS Channel}\label{app:experimentDIDS}
\subsection{Gradients to the differentiable IDS channel}\label{app:diffids}
\begin{figure*}[htb!]
   \vskip 0.in
   \centering
       \includegraphics[width=0.20\linewidth,trim={0 0 0 10},clip]{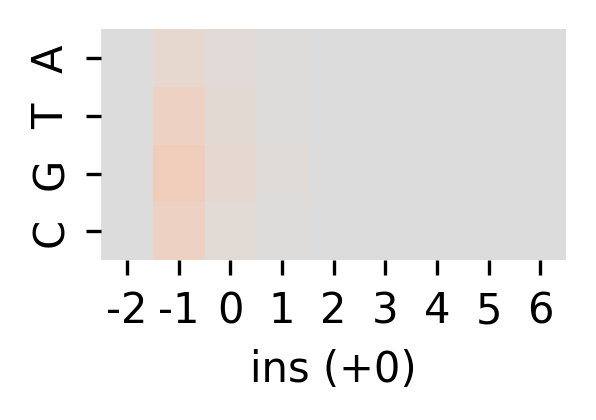}
       \includegraphics[width=0.20\linewidth,trim={0 0 0 10},clip]{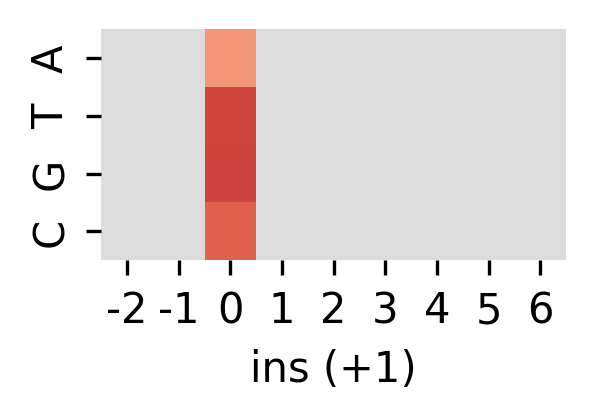}
       \includegraphics[width=0.20\linewidth,trim={0 0 0 10},clip]{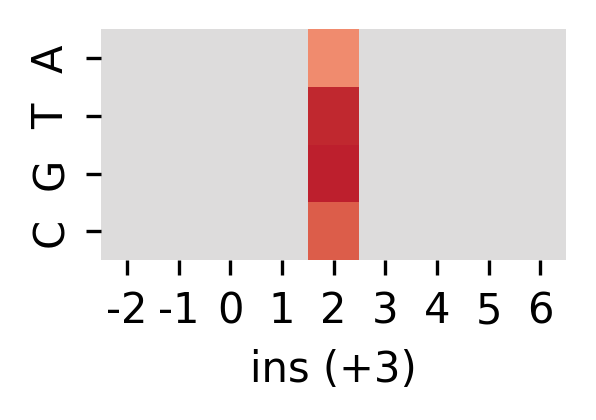}
       \includegraphics[width=0.20\linewidth,trim={0 0 0 10},clip]{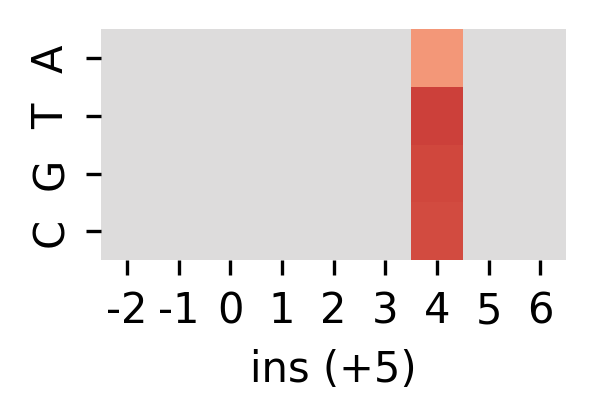}\\
       \includegraphics[width=0.20\linewidth,trim={0 0 0 10},clip]{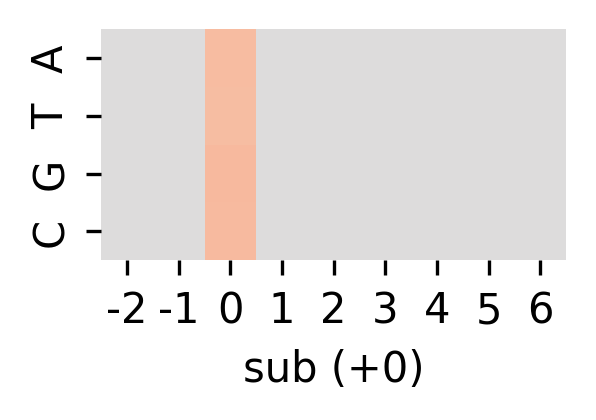}
       \includegraphics[width=0.20\linewidth,trim={0 0 0 10},clip]{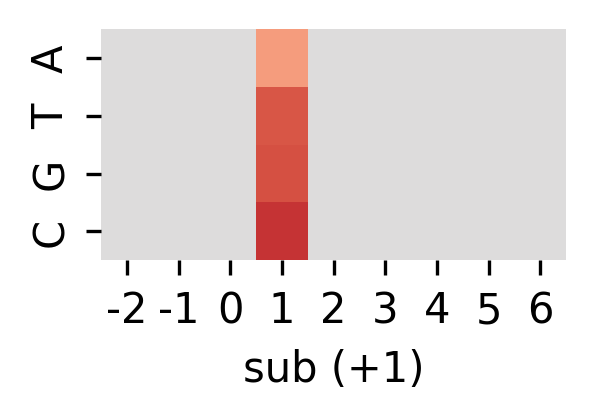}
       \includegraphics[width=0.20\linewidth,trim={0 0 0 10},clip]{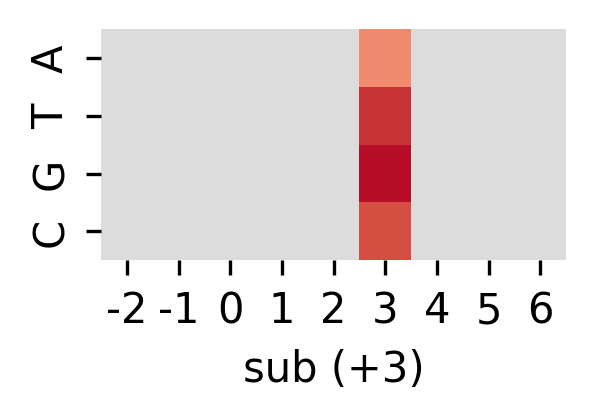}
       \includegraphics[width=0.20\linewidth,trim={0 0 0 10},clip]{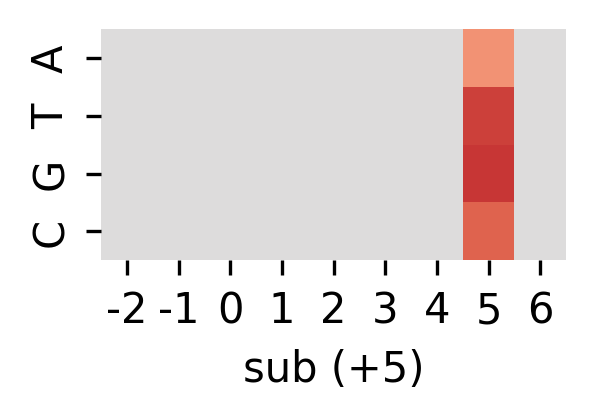}\\
       \includegraphics[width=0.20\linewidth,trim={0 0 0 10},clip]{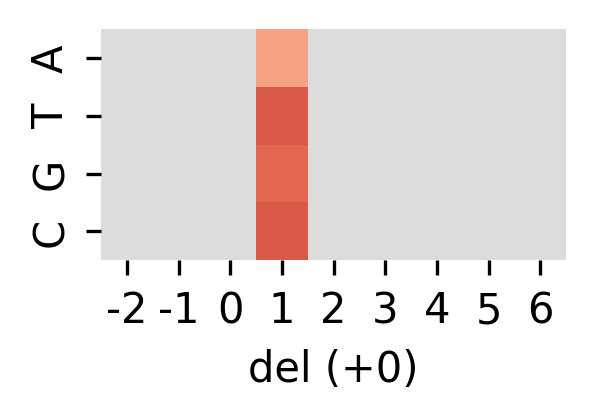}
       \includegraphics[width=0.20\linewidth,trim={0 0 0 10},clip]{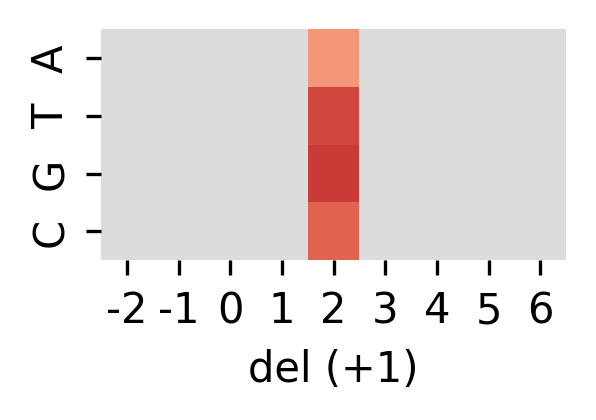}
       \includegraphics[width=0.20\linewidth,trim={0 0 0 10},clip]{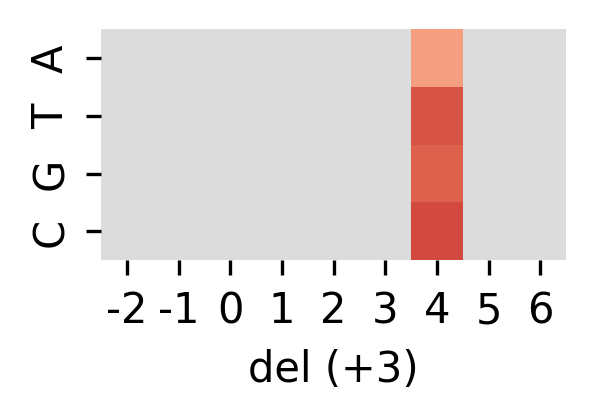}
       \includegraphics[width=0.20\linewidth,trim={0 0 0 10},clip]{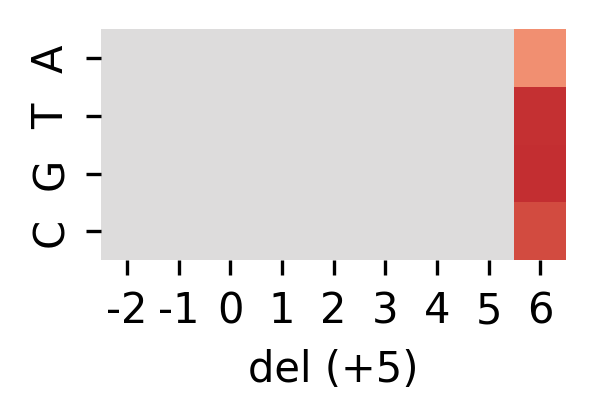}
   \caption{The averaged absolute gradients with respect to the input $\bm{c}$ over 100 runs. 
   The corresponding IDS operations were performed at an aligned $\mathrm{index}=0$ by the simulated differentiable IDS channel, 
   the gradients were back-propagated from position $+k$ of the channel output $\hat{\bm{c}}$. 
   It is suggested that the gradients identify their corresponding position in the input: $+k-1$ for insertion, 
   $+k$ for substitution, and $+k+1$ for deletion. 
   }
   \label{fig:grad}
   \vskip -0.in
\end{figure*}

To investigate whether the simulated IDS channel back-propagates the gradient reasonably, 
the channel output $\hat{\bm{c}} = \mathrm{DIDS}(\bm{c})$ is modified by altering one base to produce $\hat{\bm{c}}'$. 
The absolute values of the gradients 
of $\mathcal{L}(\hat{\bm{c}},\hat{\bm{c}}')$ with respect to the input $\bm{c}$ after back-propagation are presented in 
\cref{fig:grad}. 
For instance, subfigure $\mathrm{del (+3)}$ indicates that the IDS channel modifies $\bm{c}$ to $\hat{\bm{c}}$ 
by performing a deletion at index $0$. 
The output $\hat{\bm{c}}$ is then manually modified by applying a substitution at position 
$+3$. 
The gradients of $\mathcal{L}(\hat{\bm{c}},\hat{\bm{c}}')$ 
with respect to $\bm{c}$ are plotted over the window $[-2,+6]$. 

It is suggested in \cref{fig:grad} that the proposed differentiable IDS channel back-propagates gradients reasonably. 
The gradients shift by one base to the left (resp. right) 
when the IDS channel performs an insertion (resp. deletion) on $\bm{c}$. 
When the IDS channel operates $\bm{c}$ with a substitution, the gradients stay at the same index. 
This behavior demonstrates that the channel is able to trace the gradients through the IDS operations. 
Specifically, in the case $\mathrm{ins (+0)}$, the channel-inserted base in $\hat{\bm{c}}$ at $\mathrm{idx}$ 
is manually modified. As a result, no specific 
base in $\bm{c}$ has a connection to the manually modified base, 
leading to a diminished gradient in this scenario. 


\subsection{More on the gradients to differentiable IDS channel}\label{app:idsgrad}
Above, we illustrated that the differentiable IDS channel can effectively 
trace gradients through the IDS operations.
In this section, we focus on evaluating the channel's capability to recover the error profile through gradient-based optimization.

Given a codeword $\bm{c}$, an empty profile $\bm{p}_{0}$ which defines the identity transformation of the IDS channel such that 
$\hat{\bm{c}} = \bm{c} = \mathrm{DIDS}(\bm{c},\bm{p}_0)$, 
and a modified codeword $\hat{\bm{c}}'$ which is produced by manually modifying $\bm{c}$ through 
an insertion, deletion, or substitution at position $\mathrm{idx}$, 
the gradients of $\mathcal{L}(\hat{\bm{c}}, \hat{\bm{c}}')$ are computed 
with respect to both the input codeword $\bm{c}$ and the empty profile $\bm{p}_{0}$. 
The average gradients, calculated over 100 runs, are plotted in \cref{appfig:idsgrad} 
with position $\mathrm{idx}$ aligned to $0$. 

\begin{figure*}[htb!]
   \vskip 0.in
   \centering
   \subcaptionbox{insertion}[0.3\linewidth]
   {
       \includegraphics[width=1\linewidth]{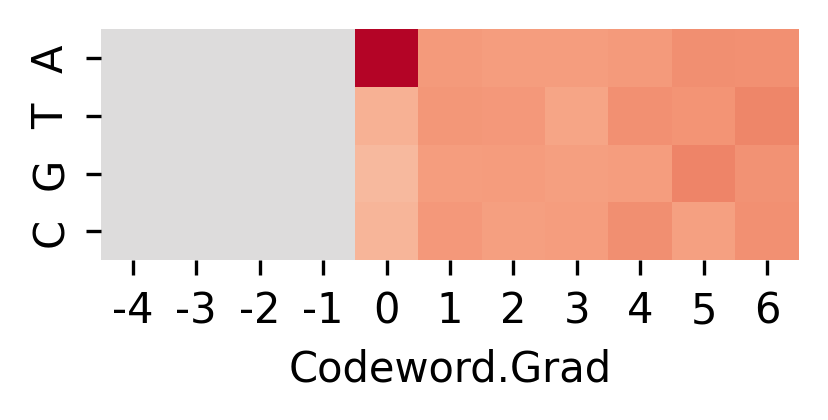}
       \includegraphics[width=1\linewidth]{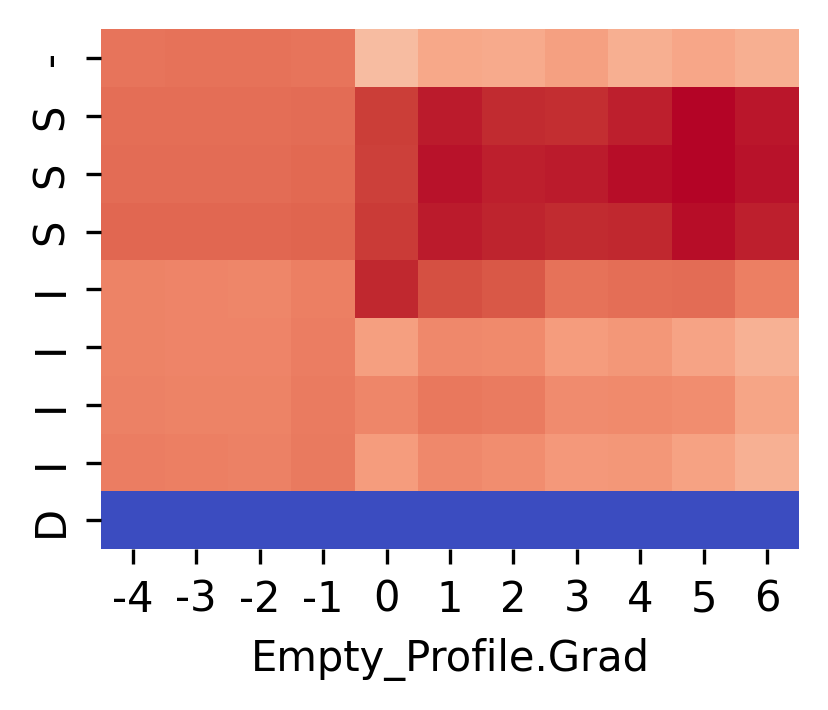}
       \includegraphics[width=1\linewidth]{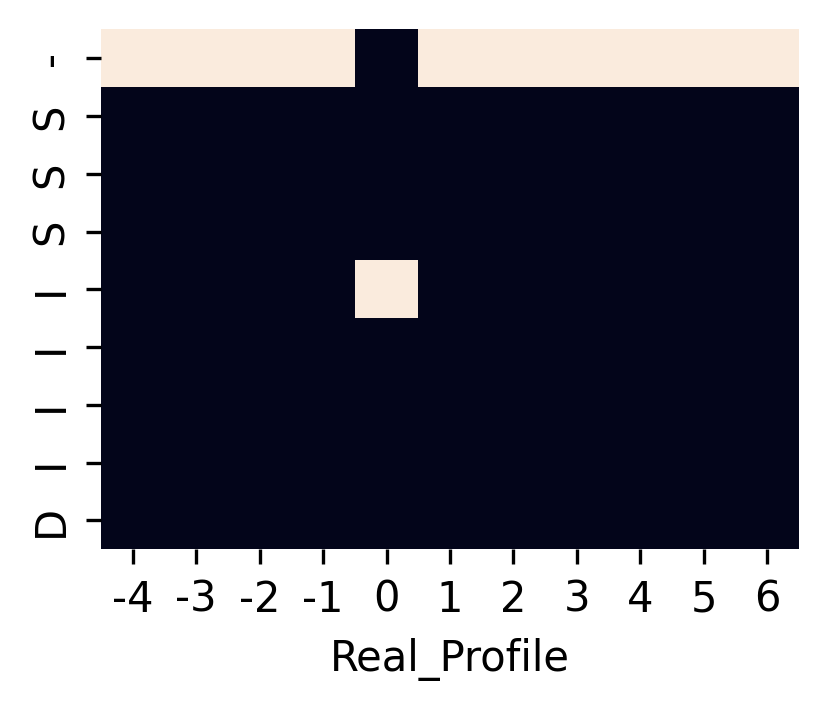}
   }%
   \subcaptionbox{deletion}[0.3\linewidth]
   {
      \includegraphics[width=1\linewidth]{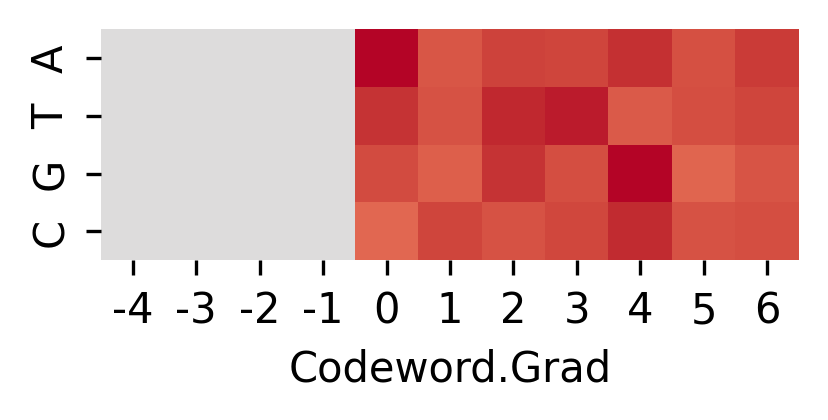}
      \includegraphics[width=1\linewidth]{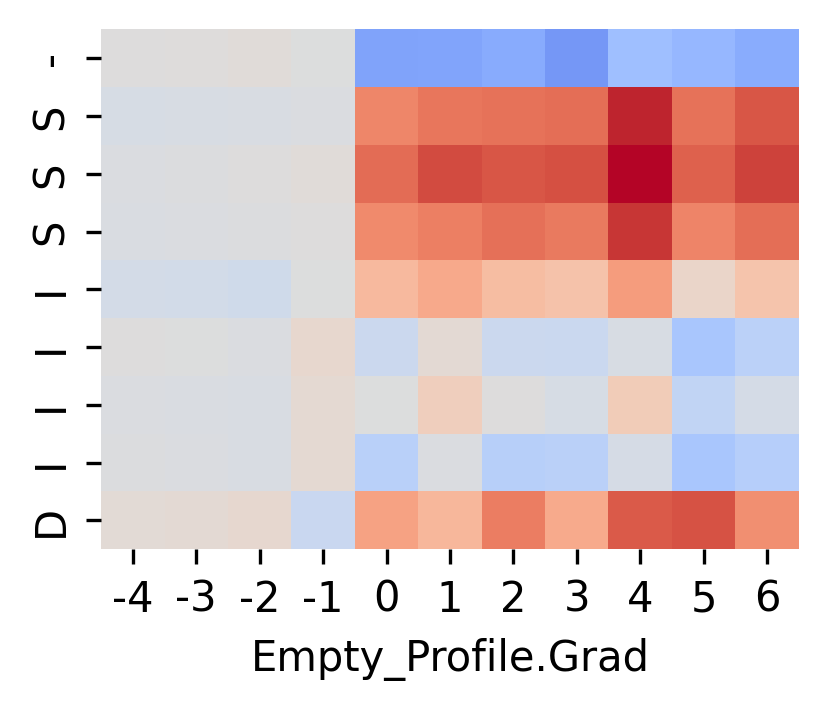}
      \includegraphics[width=1\linewidth]{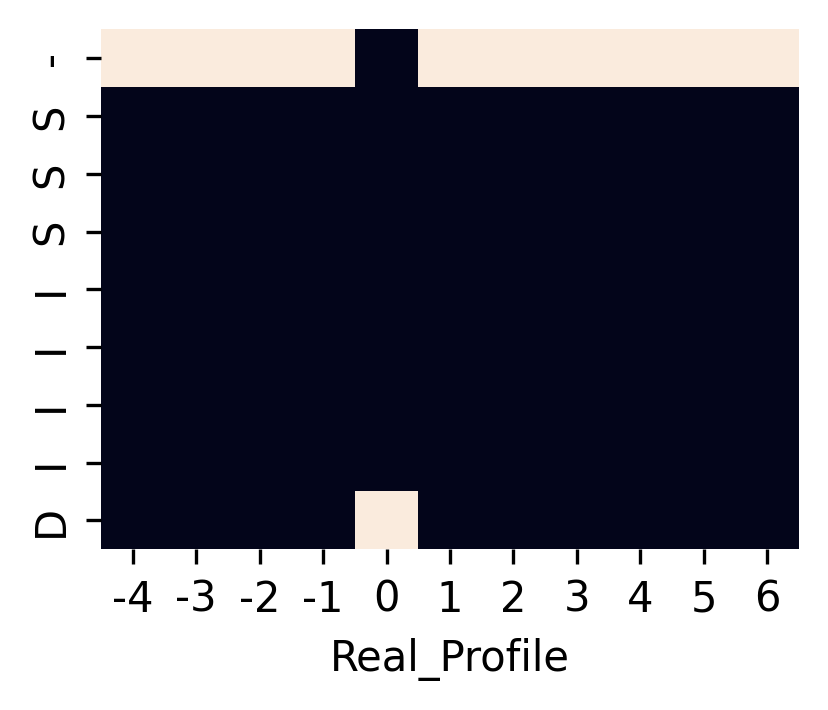}
   }%
   \subcaptionbox{substitution}[0.3\linewidth]
   {
      \includegraphics[width=1\linewidth]{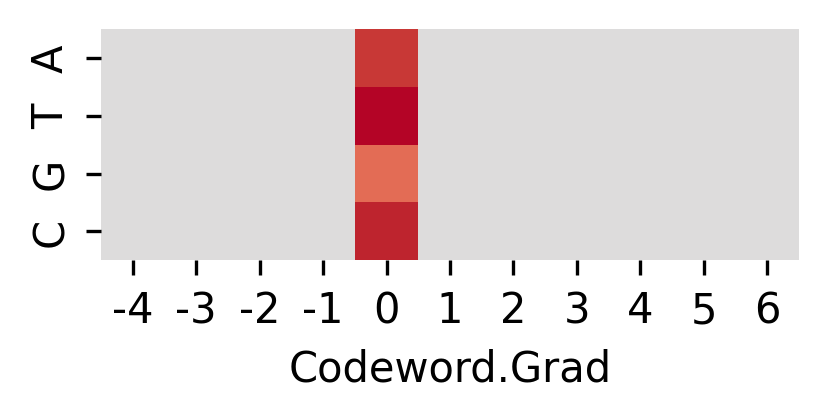}
      \includegraphics[width=1\linewidth]{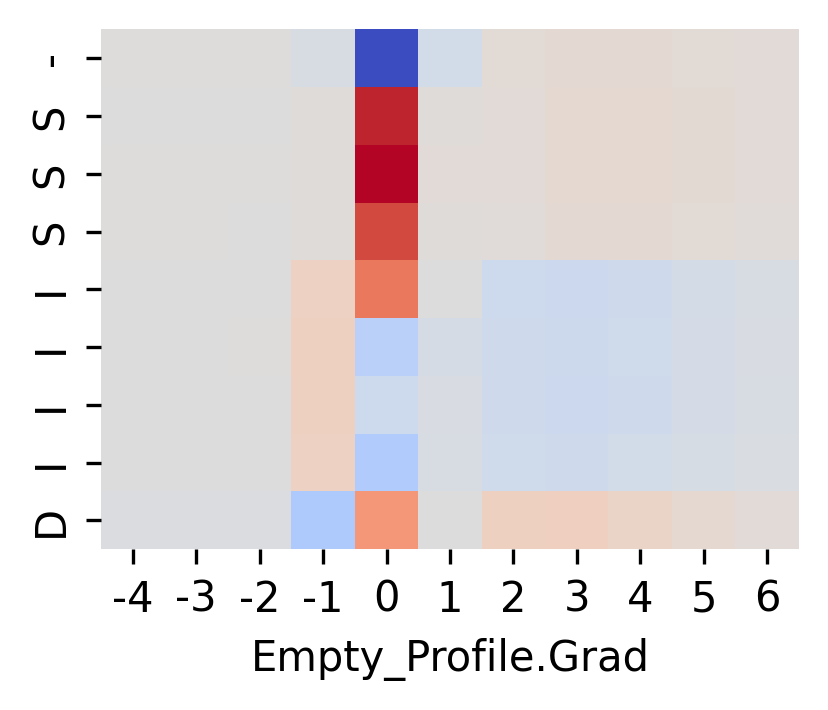}
      \includegraphics[width=1\linewidth]{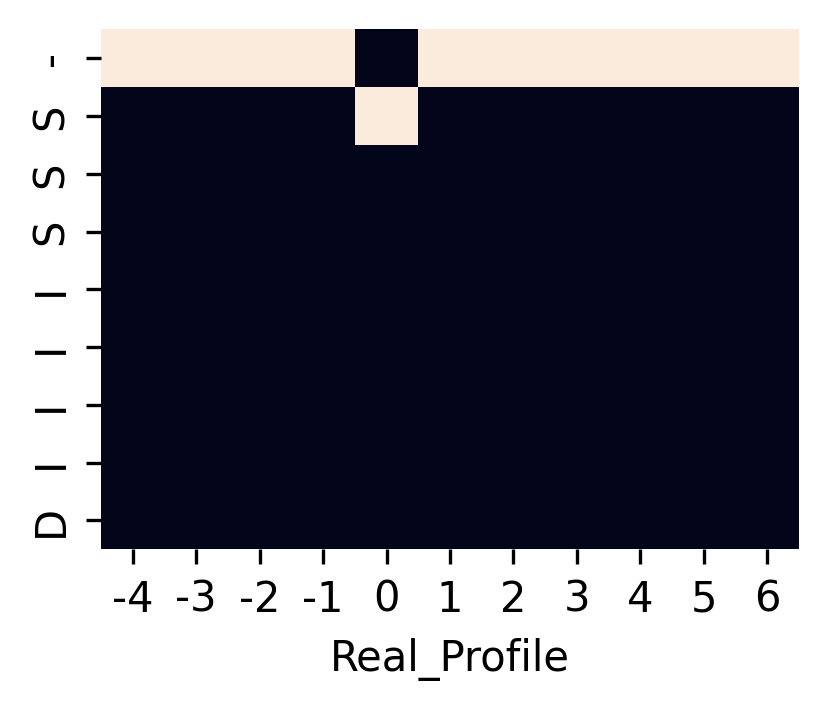}
   }%
   \caption{The gradient distribution with respect to the input codeword and the empty profile, 
   when the output codeword is manually modified. 
   The figures display the averaged gradients over 100 runs, visualizing how the gradients were back-propagated in different cases of insertions, deletions, and substitutions in the output codeword.
   }
   \label{appfig:idsgrad}
   \vskip -0.in
\end{figure*}

In~\cref{appfig:idsgrad}, it is suggested that, 
when performing an insertion or deletion, the gradients with respect to the codeword are distributed after the 
error position $\mathrm{idx}$. 
This aligns with the fact that synchronization errors (insertions or deletions) can be interpreted as 
successive substitutions starting from the error position, 
especially when the actual error profile is unknown. 
When performing a substitution, the gradients naturally concentrate at the error position $\mathrm{idx}$. 

Regarding the empty profile, $\bm{p}_0 = \bm{0}$, 
the gradients also exhibit meaningful patterns. 
For an insertion, the substitution area after $\mathrm{idx}$ is lighted by the gradients, 
supporting the view that an insertion can be seen as a sequence of substitutions if error constraints are absent. 
Additionally, the insertion area of the profile is also lighted, which makes sense since an insertion may also 
be interpreted as a series of substitutions followed by an ending insertion. 
For deletion errors, similar patterns are observed: 
the gradients are distributed in the areas of substitutions and deletions after the error position $\mathrm{idx}$, since 
the deletion can also be viewed as a series of substitutions, or as several substitutions and an ending deletion. 
For substitution errors, the gradients again concentrate at the error position $\mathrm{idx}$, 
as substitutions do not cause sequence mismatches. 

Utilizing energy constraints on the profile may be helpful for specific profile applications. 
In this work, only the gradients with respect to the codeword participate in the training phase, 
the existing version of the simulated differentiable IDS channel is assumed to be adequate.

\section{Ablation Study on the Auxiliary Reconstruction Loss}\label{app:ablationaux}
\subsection{Effects of the auxiliary reconstruction loss}\label{subsec:auxexperiment}
\begin{figure*}[htb!]
    \vskip 0.in
    \centering
    \subcaptionbox{$\mu=0$}[0.24\linewidth]
    {
        \includegraphics[width=1\linewidth]{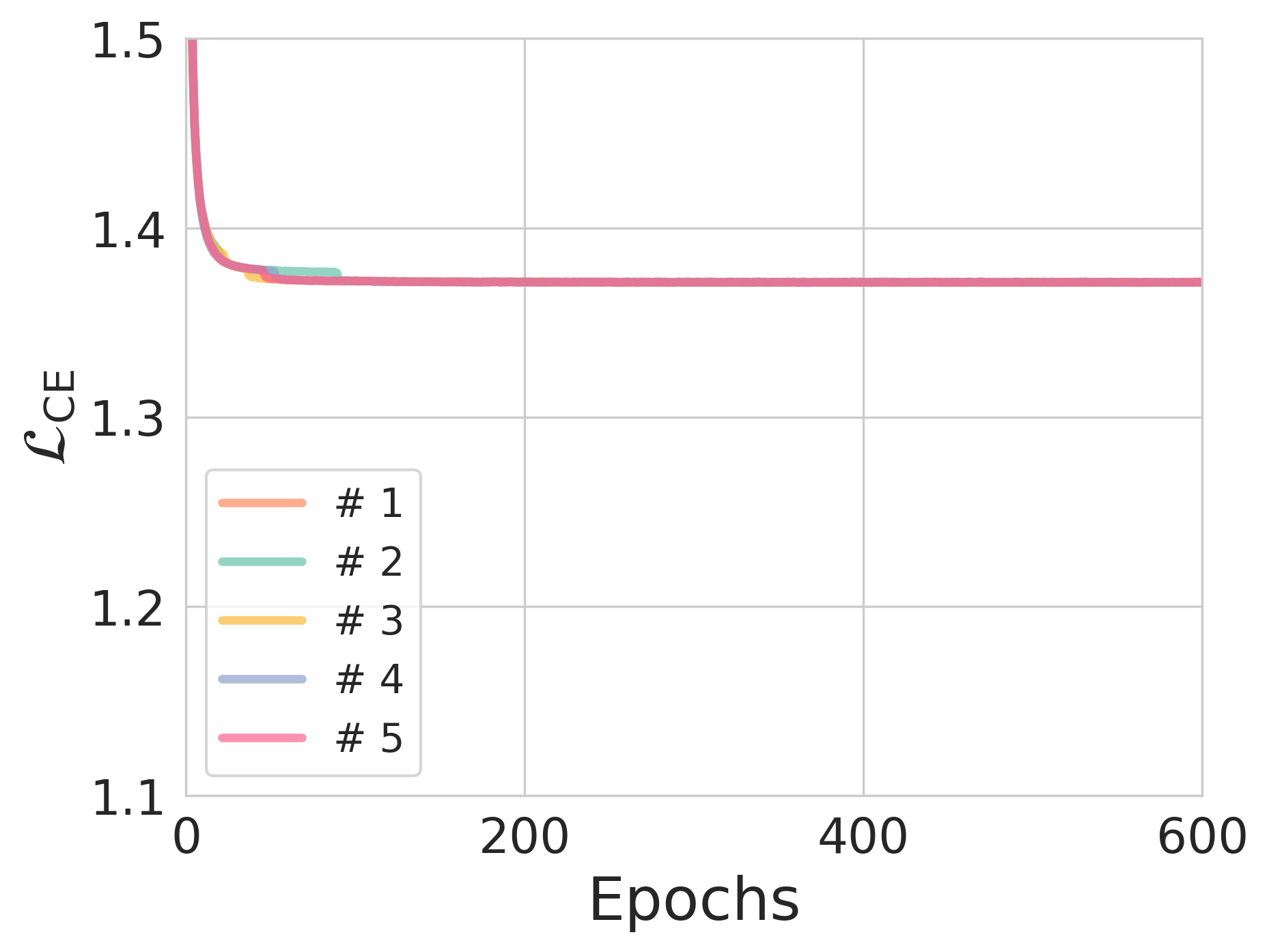}
        \includegraphics[width=1\linewidth]{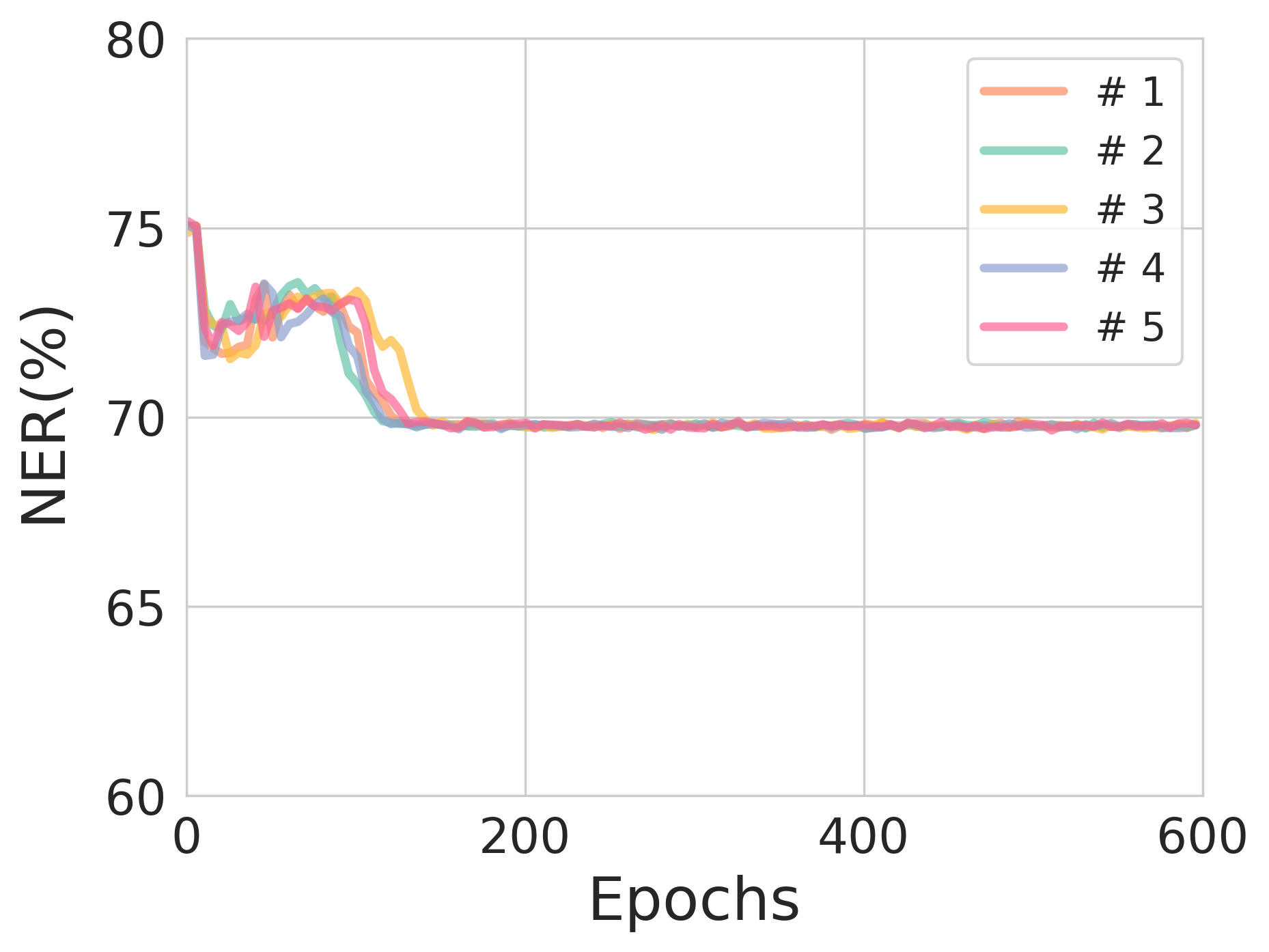}
    }%
    \subcaptionbox{$\mu=0.5$}[0.24\linewidth]
    {
        \includegraphics[width=1\linewidth]{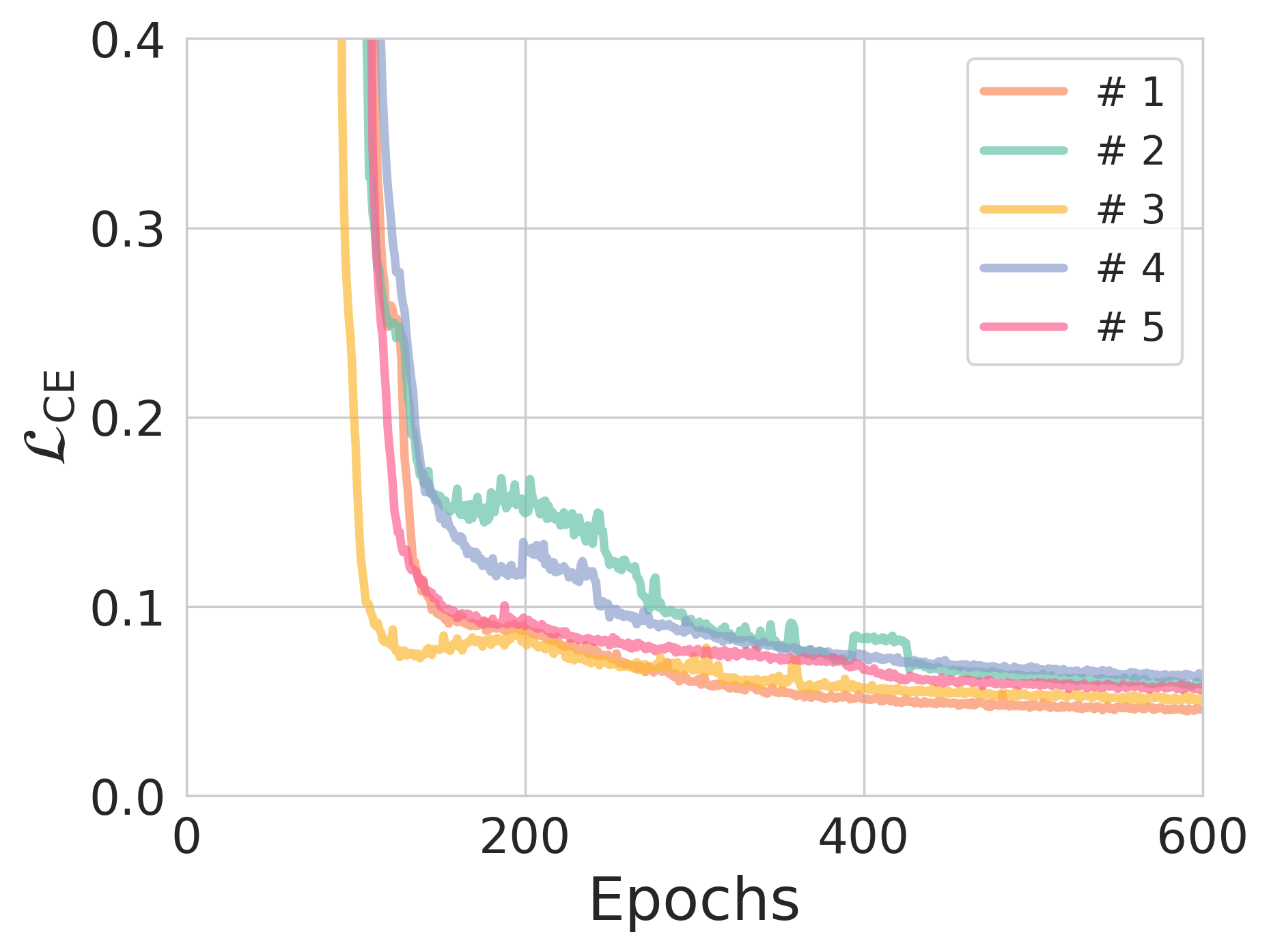}
        \includegraphics[width=1\linewidth]{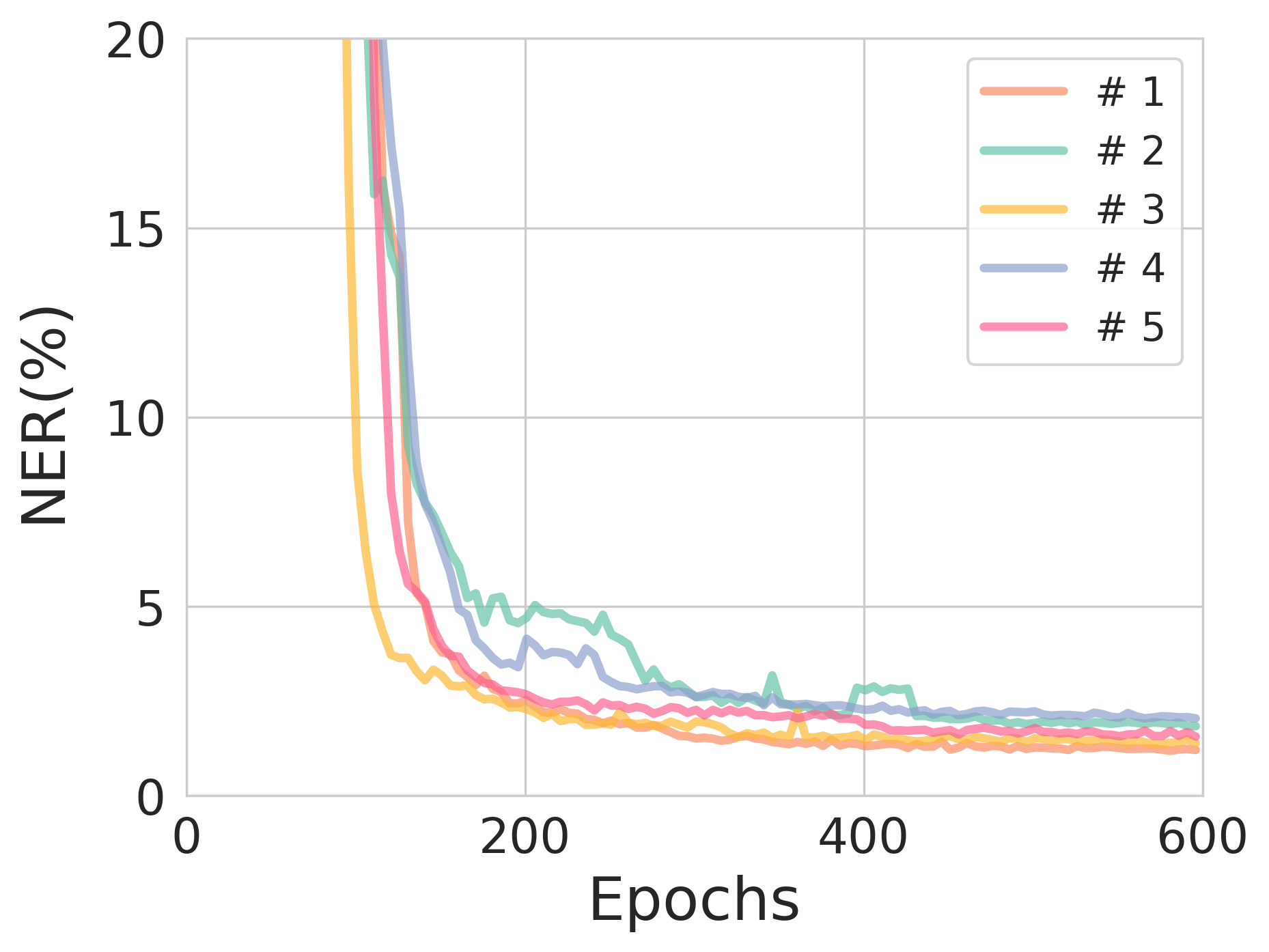}
    }%
    \subcaptionbox{$\mu=1$}[0.24\linewidth]
    {
        \includegraphics[width=1\linewidth]{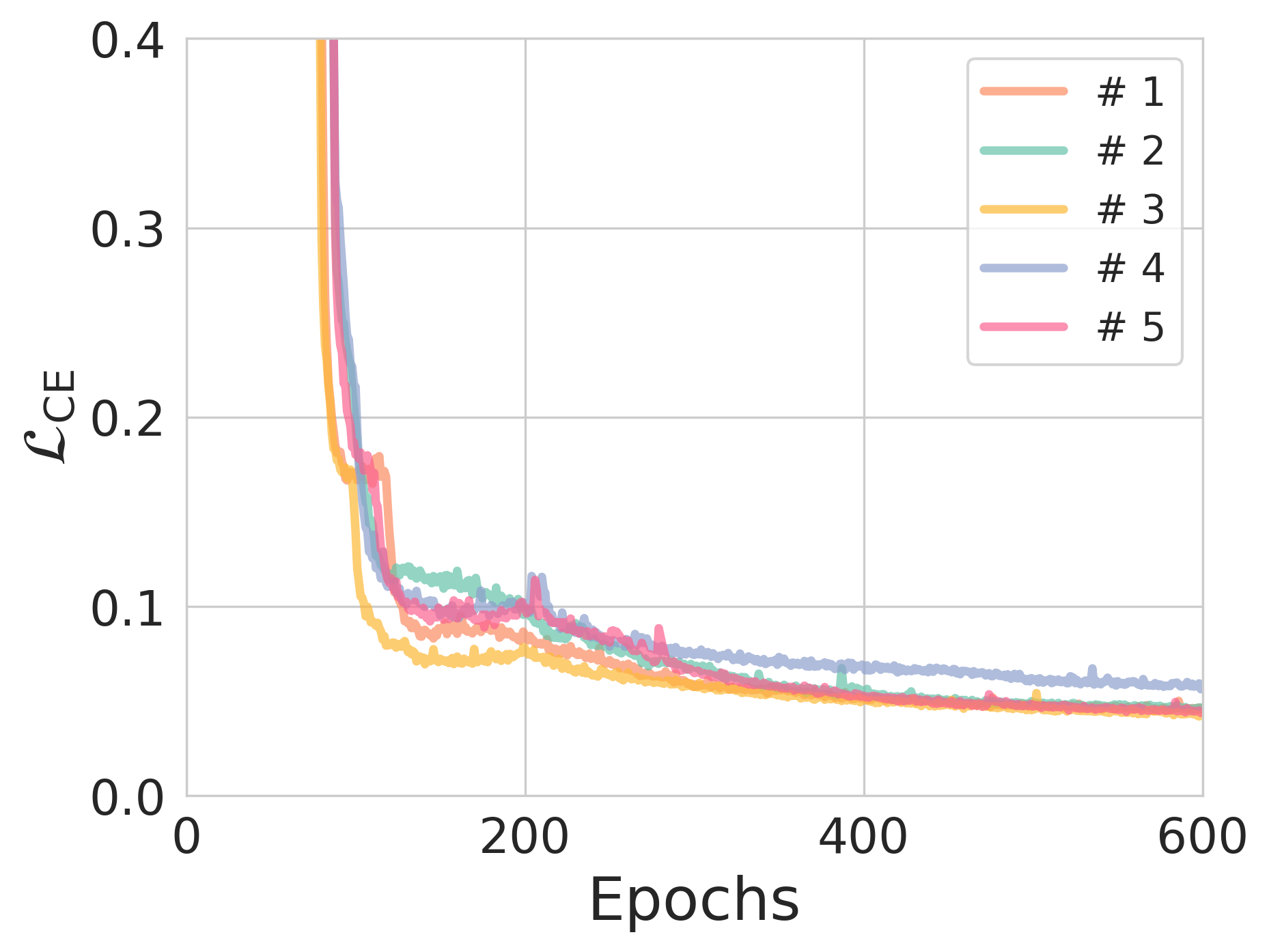}
        \includegraphics[width=1\linewidth]{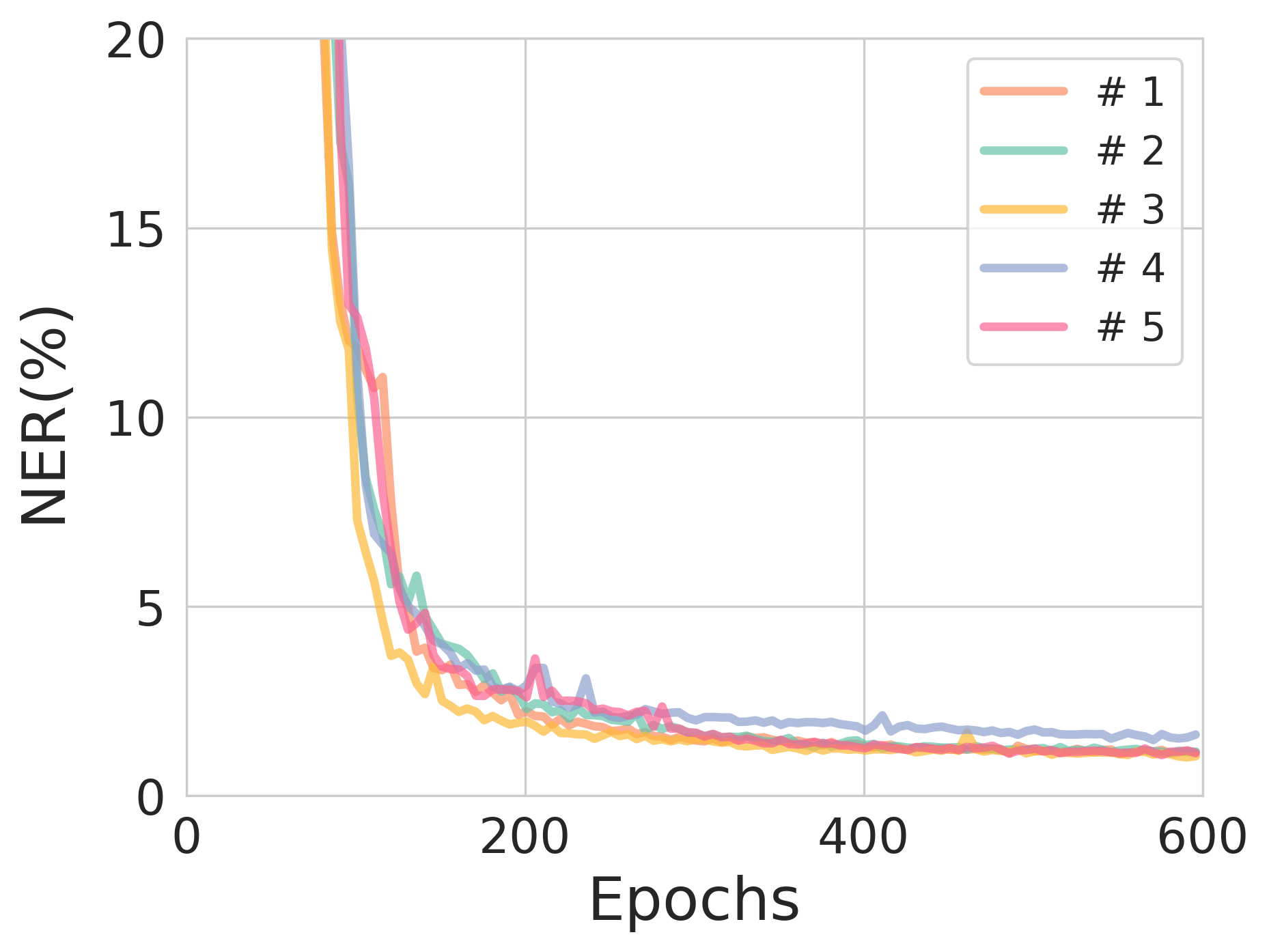}
    }%
    \subcaptionbox{$\mu=1.5$}[0.24\linewidth]
    {
        \includegraphics[width=1\linewidth]{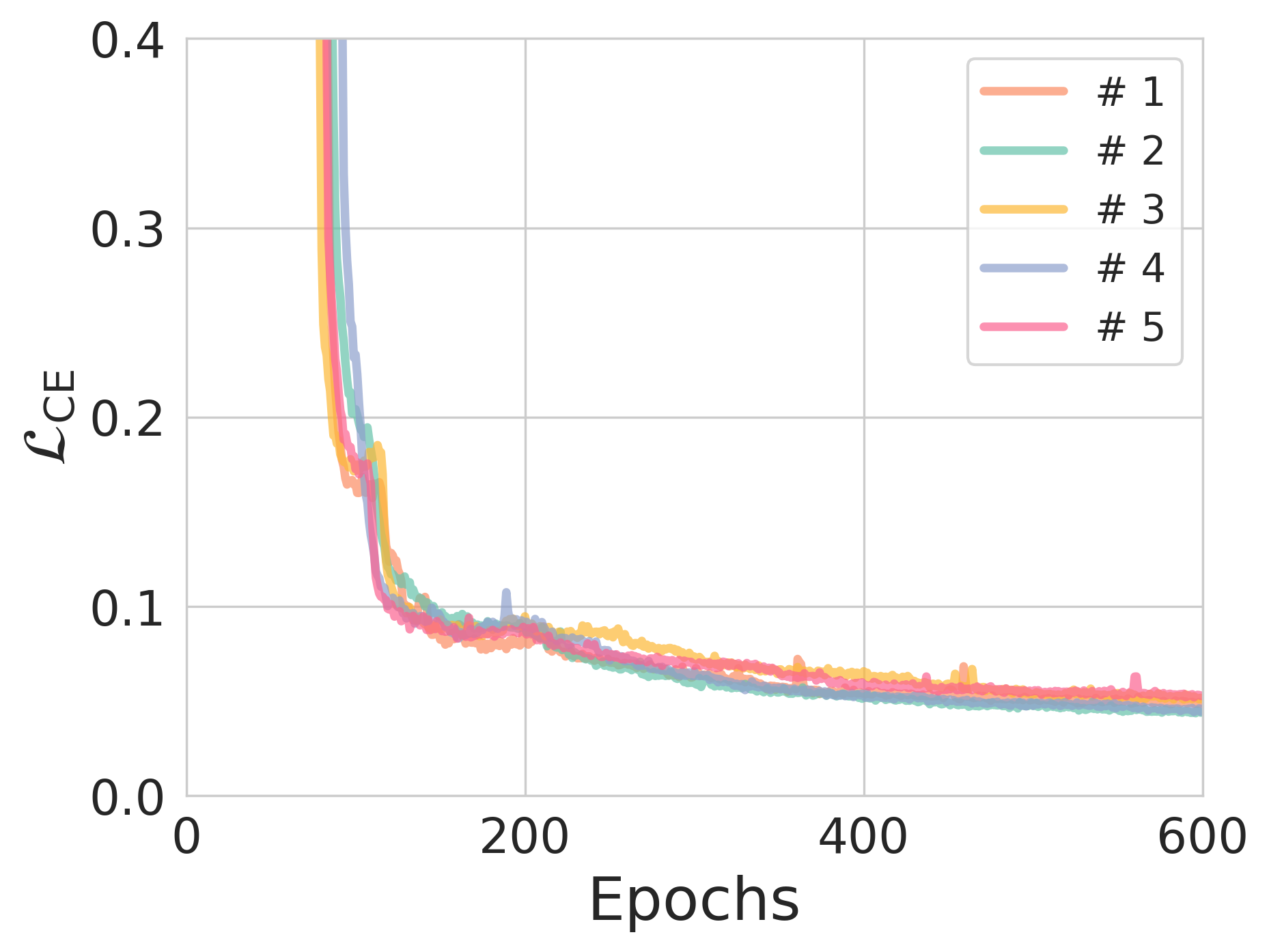}
        \includegraphics[width=1\linewidth]{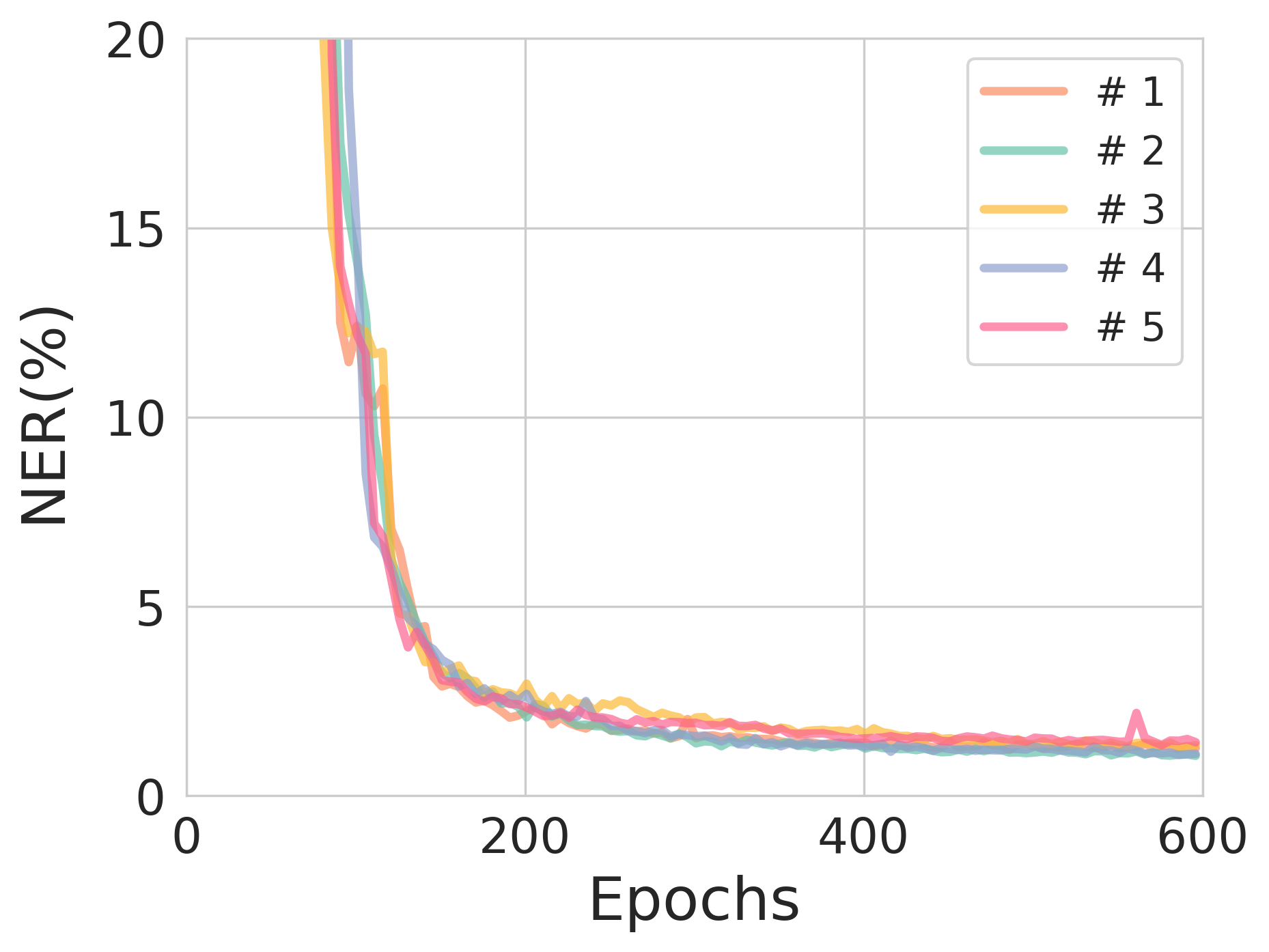}
    }%
    \caption{The reconstruction loss $\mathcal{L}_\mathrm{CE}$ between the source and recovered sequences, 
    and the validation NER for various choices of $\mu\in\{0,0.5,1,1.5\}$. 
    Each curve in the subfigures represents one of the $5$ runs conducted in the experiment 
    and is plotted against the training epochs.
    }
    \label{fig:mu}
    \vskip -0.in
\end{figure*}

Experiments with different choices of the hyperparameter $\mu$ were conducted, which are 
$\mu=0$ indicating the absence of the auxiliary reconstruction loss, 
and $\mu\in\{0.5,1,1.5\}$ for different weights for the auxiliary loss. 
The validation NER and the reconstruction loss between source and recovered sequences are 
plotted against the training epochs. 

The first column of \cref{fig:mu} indicates that without the auxiliary loss, 
all $5$ runs of the training fail, producing random output. 
By comparing the first column with the other three, the effectiveness of introducing the auxiliary loss can be inferred. 
In the subfigures corresponding to $\mu\in\{0.5,1,1.5\}$, all the models converge well, 
and the NERs also exhibit a similar convergence. 
This suggests the application of the auxiliary loss is essential, 
but the weight of this loss has minimum influence on the final performance. 

\subsection{Auxiliary loss on patterns beyond sequence reconstruction}\label{app:aux}
In \cref{subsec:auxexperiment}, the necessity of introducing a auxiliary reconstruction 
task to the encoder is verified. 
After these experiments, a natural question arises: 
How about imparting the encoder with higher initial logical ability 
through a more complicated task rather than replication. 
Motivated by this, we adopted commonly used operations from existing IDS-correcting codes 
and attempted to recover the sequence from these operations using the encoder. 
In practice, we employed the forward difference $\mathrm{Diff}(\bm{s})$, where
\begin{equation}
    \mathrm{Diff}(\bm{s})_i = s_i-s_{i+1} \mod{4}, 
\end{equation}
the position information-encoded sequence $\mathrm{Pos}(\bm{s})$, where 
\begin{equation}
    \mathrm{Pos}(\bm{s})_i = s_i+i \mod{4}, 
\end{equation}
and their combinations as the reconstructed sequences. 

The evaluation NERs against training epochs are plotted in \cref{fig:beyond} under different 
combinations of the identity mapping $\mathrm{I}$, $\mathrm{Diff}$, and $\mathrm{Pos}$. 
It is clear that the reconstruction of the identity mapping $\mathrm{I}$ outperforms 
$\mathrm{Diff}$ and $\mathrm{Pos}$. 
Introducing the identity mapping $\mathrm{I}$ to $\mathrm{Diff}$ and $\mathrm{Pos}$ helps improving the 
convergence of the model, but final results have illustrated that they are still worse than simple applying 
the identify mapping $\mathrm{I}$ as the auxiliary task. 
These variations may be attributed to the capabilities of the Transformers in our setting or 
the disordered implicit timings during training. 

\begin{figure*}[htb!]
   \vskip 0.in
   \centering
   \subcaptionbox{$\mathrm{I}$}[0.30\linewidth]
   {
       \includegraphics[width=1\linewidth]{gumbel/gumbelNER.png}
   }%
   \subcaptionbox{$\mathrm{Diff}$}[0.30\linewidth]
   {
       \includegraphics[width=1\linewidth]{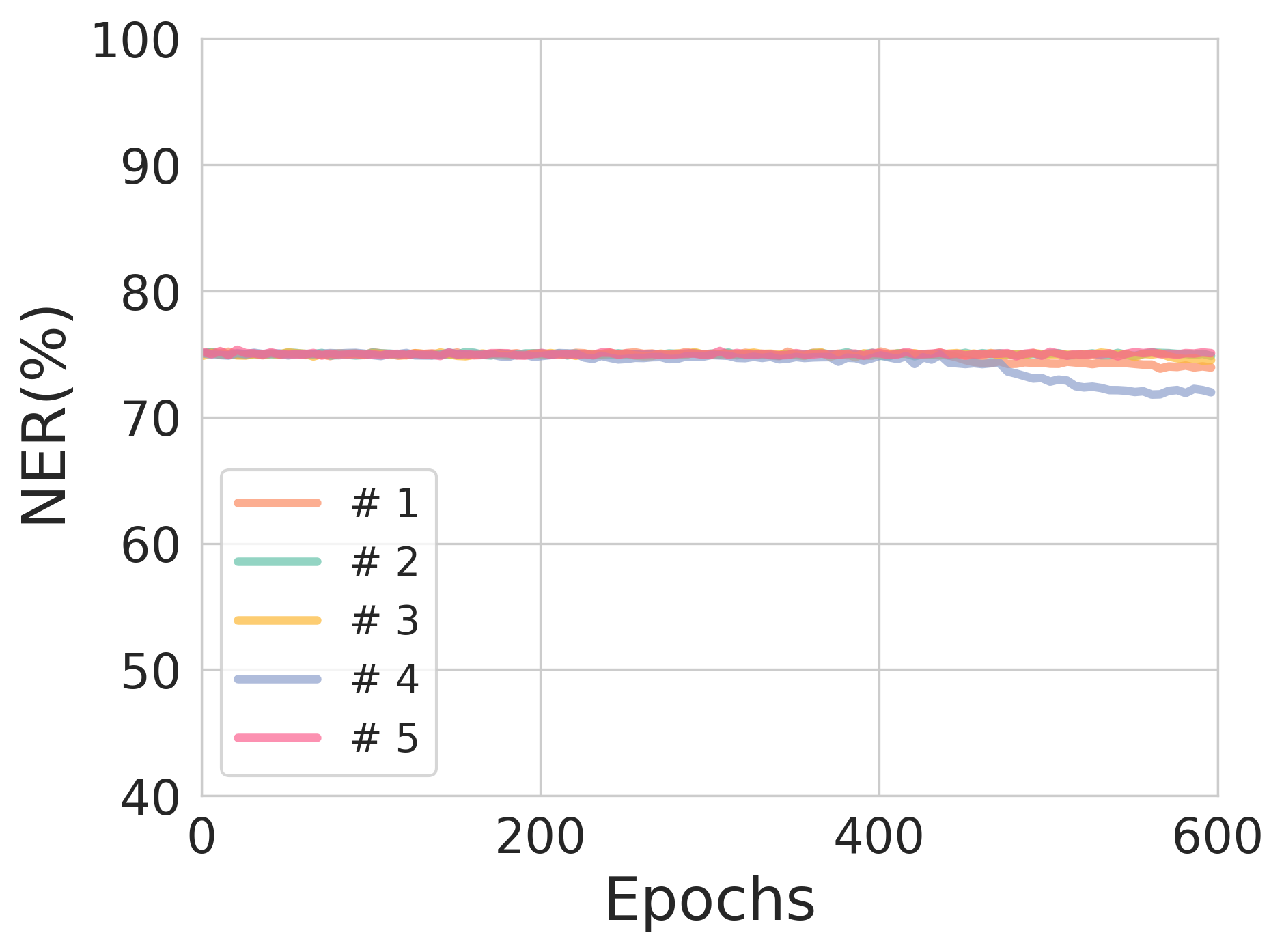}
   }%
   \subcaptionbox{$\mathrm{Pos}$}[0.30\linewidth]
   {
       \includegraphics[width=1\linewidth]{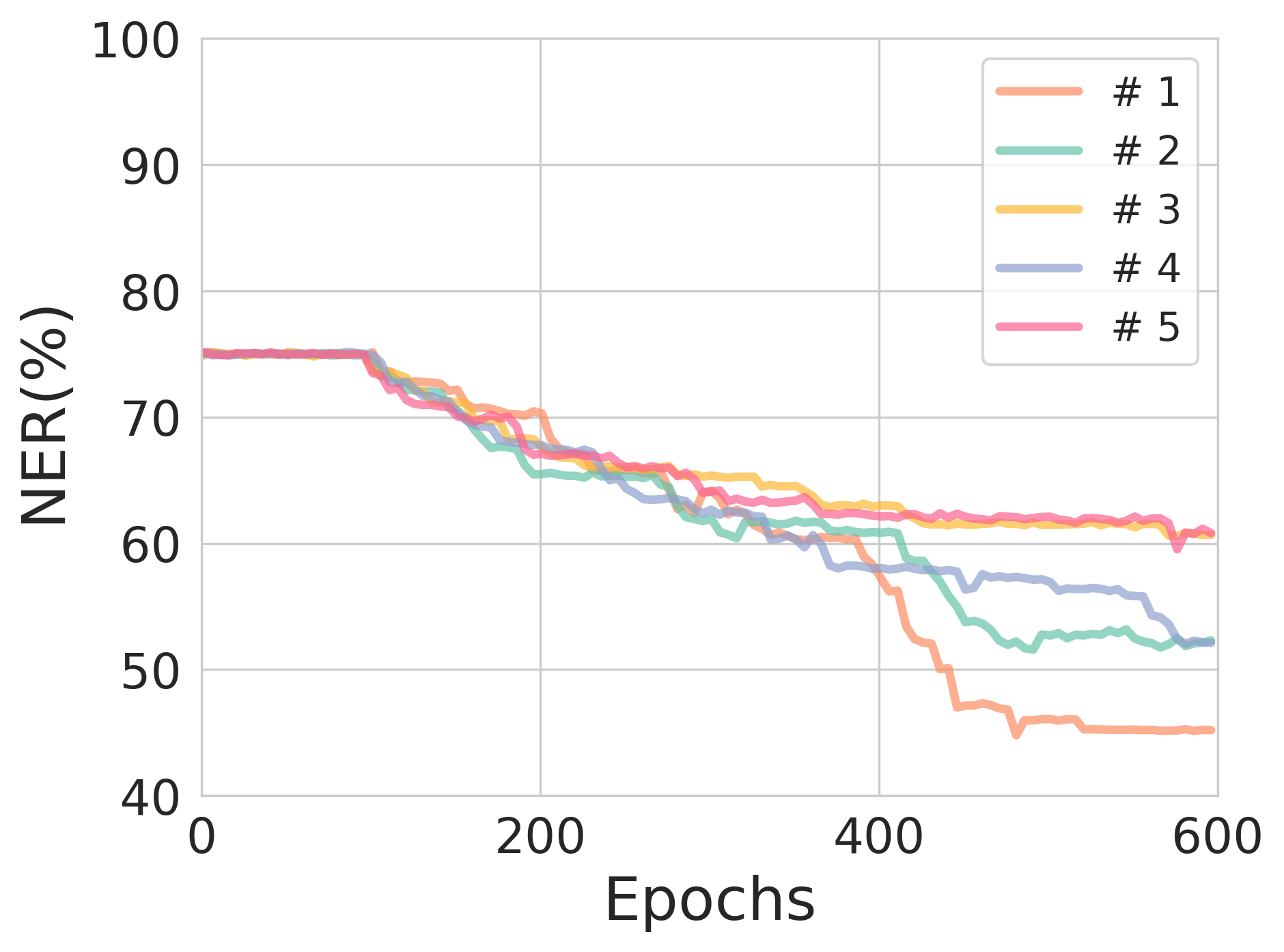}
   }%
   \\
   \subcaptionbox{$\mathrm{I}+\mathrm{Diff}$}[0.30\linewidth]
   {
       \includegraphics[width=1\linewidth]{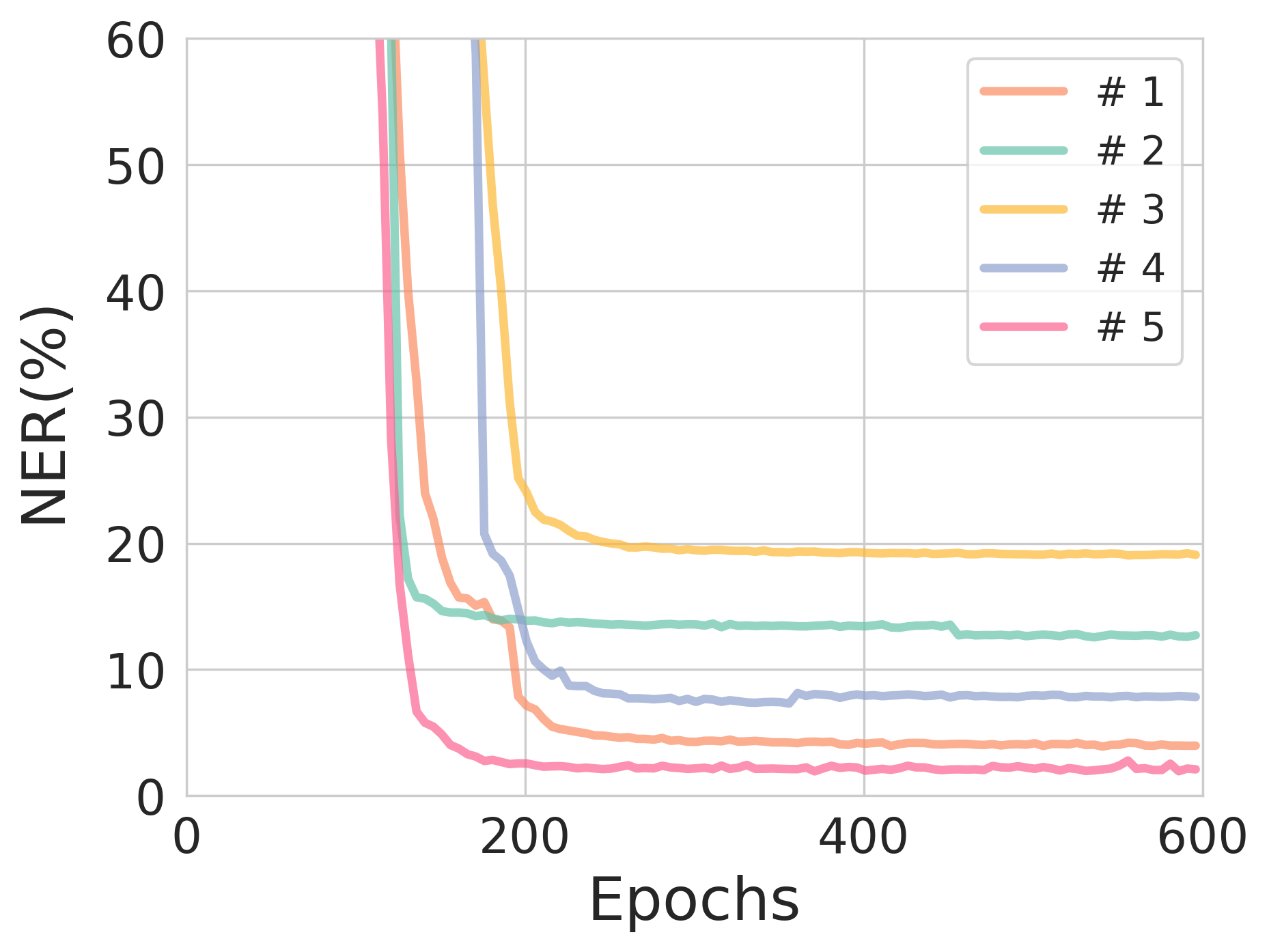}
   }%
   \subcaptionbox{$\mathrm{I}+\mathrm{Pos}$}[0.30\linewidth]
   {
       \includegraphics[width=1\linewidth]{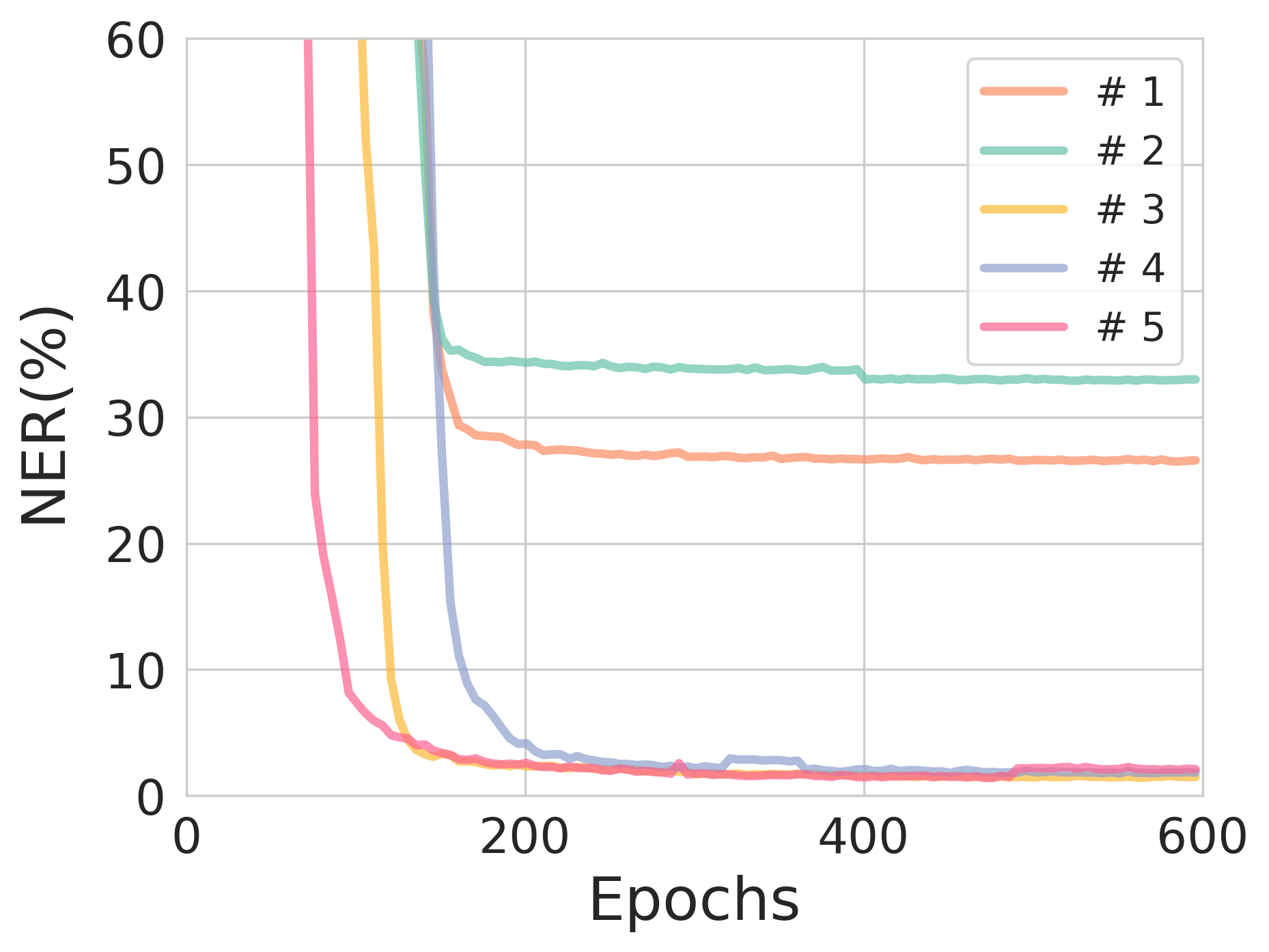}
   }%
   \caption{The validation NER against the training epochs with different choices of auxiliary reconstruction. 
   The reconstructed sequences are produced by combinations of the identity mapping $\mathrm{I}$, $\mathrm{Diff}$, 
   and $\mathrm{Pos}$, where $+$ denotes sequence concatenating. 
   }
   \label{fig:beyond}
   \vskip -0.in
\end{figure*}

\section{Preparation of the Datasets}\label{app:dataset}
This work focuses on DNA-based storage, which involves storing and retrieving information from 
synthesized DNA molecules \emph{in vitro}, typically composed of four bases: 
$\{\mathrm{A},\mathrm{T},\mathrm{G},\mathrm{C}\}$. 
Although genomes also use DNA molecules \emph{in vivo} 
to store the information necessary for an organism, 
DNA-based information storage in this context is largely unrelated to genetic data.

DNA-based information storage is a generic storage method, capable of storing arbitrary information.  
The source information is randomly generated with equal probabilities 
over the set $\{\mathrm{A},\mathrm{T},\mathrm{G},\mathrm{C}\}$, 
which is in a trivial bijection with binary sequences. 

\begin{table*}[htb!]
    \centering
    {\footnotesize
    \begin{tabular}{r|l}\toprule
        Default ($\mathrm{Hom}$) & total error: $1$\%, $\mathrm{Ins}:\mathrm{Del}:\mathrm{Subs} = 1:1:1$, position-insensitive                     \\
        $\mathrm{Asc}$           & total error: $1$\%, $\mathrm{Ins}:\mathrm{Del}:\mathrm{Subs} = 1:1:1$, position-sensitive: start: $0$\%- end: $2$\%  \\
        $\mathrm{Des}$           & total error: $1$\%, $\mathrm{Ins}:\mathrm{Del}:\mathrm{Subs} = 1:1:1$, position-sensitive: start: $2$\%- end: $0$\%                      \\
        $k$\% error channel      & total error: $k$\%, $\mathrm{Ins}:\mathrm{Del}:\mathrm{Subs} = 1:1:1$, position-insensitive                                        \\
        C111                     & total error: $10.36$\%, $\mathrm{Ins}:\mathrm{Del}:\mathrm{Subs} = 1:1:1$, position-insensitive                                       \\
        C253                     & total error: $10.36$\%, $\mathrm{Ins}:\mathrm{Del}:\mathrm{Subs} = 1.66:5.31:3.38$, position-insensitive                                  \\
        MemSim                   & total error: $10.36$\%. \\
        &\makecell[l]{For each codeword $\bm{c}$, it is passed through the released MemSim software \cite{hamoum2021channel} \\
        to generate the output sequence $\hat{\bm{c}}$ from the channel. 
        The error profile $\bm{p}$ is then \\
        inferred by comparing $\bm{c}$ with $\hat{\bm{c}}$.}          \\
        \bottomrule
        \end{tabular}
    }
    \caption{Details of all employed channel configurations, including the specific parameters for generating the error profiles.
    }\label{tab:channelsetting}
\end{table*}

\begin{figure}[htb!]
    \vskip 0.in
    \centering
        \includegraphics[width=0.87\linewidth,trim={0 0 0 0},clip]{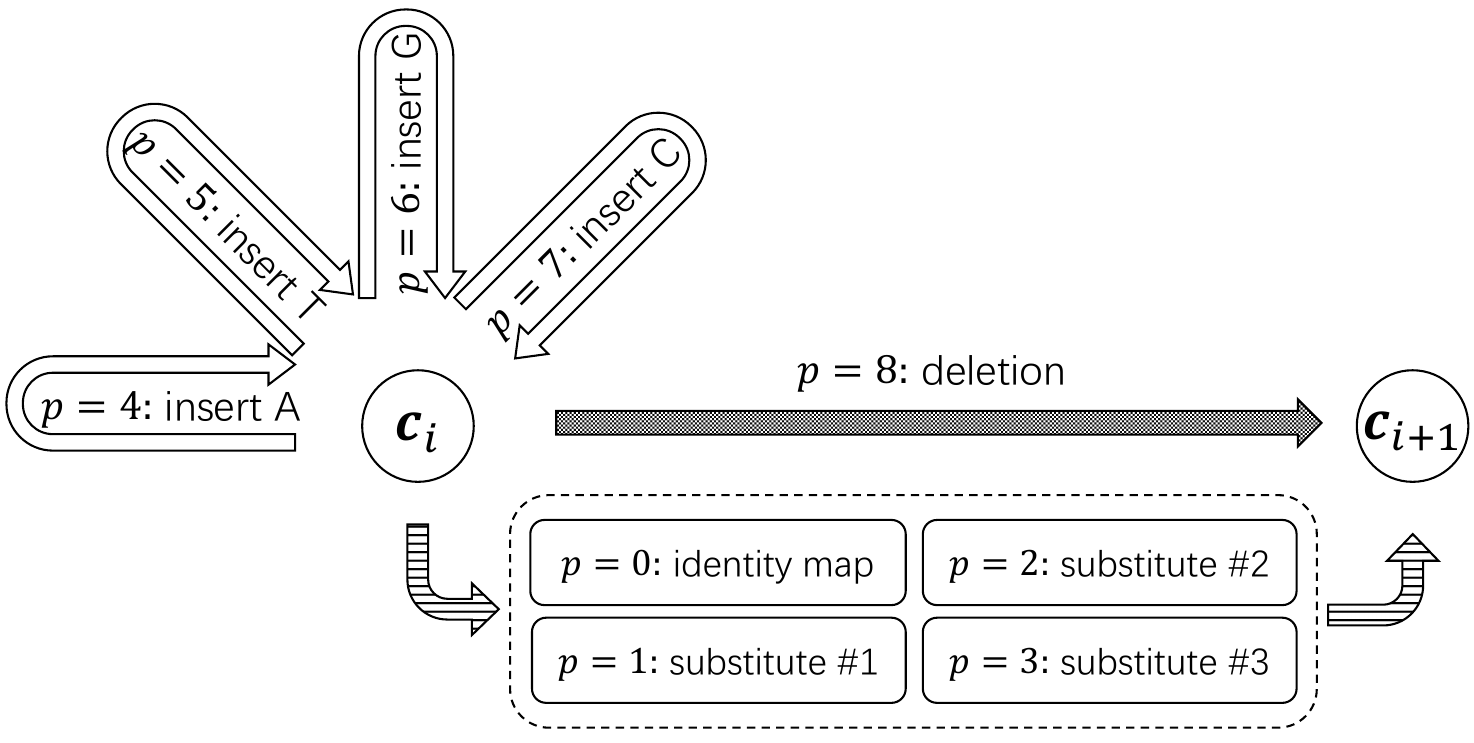}
    \caption{The transmission of a single symbol in the IDS channel involves errors of types \#1-\#8, 
    which correspond to three types of substitutions, four types of insertions, and deletions. 
    }
    \label{fig:profile}
    \vskip -0.in
\end{figure}

The error profile records the difference between the input and output of the IDS channel.  
An error profile is a sequence of symbols in the range $0$ to $8$. 
In practice, two pointers are used for the DNA sequence and the profile sequence, respectively. 
When $p=0$, both pointers advance to the next position, 
appending the DNA letter in the output. 
When $p=1,2,3$, both pointers also advance, 
but the DNA letter is substituted by one of the other three bases in the output, 
corresponding to the rolling operation described in \cref{sec:ids}.
When $p=4,5,6,7$, the profile pointer moves to the next while the DNA pointer keeps stationary, 
inserting one base in the output sequence according to $p$. 
When $p=8$, both pointers advance to the next position, but nothing is appended in the output, 
which is a deletion. 
One step of this operation is illustrated in \cref{fig:profile}. 
This is a common strategy as in \cite{davey2001reliable,yan2022segmented}. 

In the experiments, for each channel, the source sequence is generated with equal probabilities 
over the set $\{\mathrm{A},\mathrm{T},\mathrm{G},\mathrm{C}\}$, 
and the corresponding error profiles are generated as shown in \cref{tab:channelsetting}. 

\section{Transformer, Complexity, and Time Consumption}\label{app:complexity}
Transformers~\cite{vaswani2017attention}, well-known deep learning architectures, rely on the attention mechanism. 
Each head of a Transformer model processes features according to the following formula:
\begin{equation}\label{eqn:attention}
    \mathrm{Attention}(Q,K,V) = \mathrm{softmax}\left(\frac{QK^{T}}{\sqrt{d_k}}\right)V.
\end{equation}
In this work, each layer comprises $16$ attention heads with an embedding dimension $512$, 
and a total of $3+3$ attention layers are used for the sequence-to-sequence model. 
Both the encoder and decoder are implemented as such sequence-to-sequence models. 
For the differentiable IDS channel, a $1+1$ layered sequence-to-sequence model is employed. 

Since attention is calculated globally over the sequence in \cref{eqn:attention}, 
it has a complexity of $O(n^2)$. 
Without delving into the many efficient Transformer architectures, 
the time consumption 
was measured by decoding $1,280,000$ codewords using an RTX3090. 
The encoder, which shares the same structure, exhibits similar performance. 
The results are acceptable 
and are presented in \cref{tab:timecons}. 
\begin{table}[htb!]
    \vskip 0.in
    \centering
    {\small
    \begin{tabular}{rcccc}\toprule
                & $\ell_s=50$ & $\ell_s=75$ & $\ell_s=100$ & $\ell_s=125$ \\\midrule
        time (s)      & 521.94 & 573.87 & 623.92  & 687.76  \\\bottomrule
    \end{tabular}
    }
    \caption{Time consumption of decoding $1,280,000$ codewords for each source length $\ell_s$ by an RTX3090. 
    }\label{tab:timecons}
    \vskip -0.in
\end{table}

\end{document}